\documentclass{amsart}
\usepackage{hyperref}
\usepackage{url}
\usepackage[latin1]{inputenc}
\usepackage[boxed,algoruled,linesnumbered,titlenotnumbered]{algorithm2e}
\usepackage[final]{graphicx}
\usepackage[round]{natbib}

\usepackage{amsthm}
\usepackage{amsmath}
\usepackage{amssymb}

\newtheorem{proposition}{Proposition}
\newtheorem{theorem}{Theorem}

\newtheorem*{theorem*}{Theorem}
\newtheorem{lemma}{Lemma}
\newtheorem{corollary}{Corollary}
\newtheorem{definition}{Definition}
\newtheoremstyle{example}{2ex}{\topsep}%
   {\normalfont\footnotesize}
   {}
   {\bfseries}
   {}
   { }
   {\thmname{#1}\thmnumber{ #2} [\thmnote{#3}].}
\theoremstyle{example}
\newtheorem{example}{Example}

\newcommand{\hypo}[1]{\sc{(#1)}}

\makeatletter
\newenvironment{hypothesis}[1]{\removelastskip\smallbreak\noindent%
  \def\@currentlabel{\hypo{#1}}%
  {\textbf{Hypothesis}}~{\hypo{#1}}\ \normalfont\normalsize}{\smallbreak\smallskip}
\newenvironment{convention}[1]{\removelastskip\smallbreak\noindent%
  \def\@currentlabel{\hypo{#1}}%
  {\textbf{Convention}}~{\hypo{#1}}\ \normalfont\normalsize}{\smallbreak\smallskip}
\makeatother

\newcounter{DuviiiLastNumberLineofAlgo}

\mathchardef\slGamma="0100
\mathchardef\slDelta="0101
\mathchardef\slTheta="0102
\mathchardef\slLambda="0103
\mathchardef\slXi="0104
\mathchardef\slPi="0105
\mathchardef\slSigma="0106
\mathchardef\slUpsilon="0107
\mathchardef\slPhi="0108
\mathchardef\slPsi="0109
\mathchardef\slOmega="010A

\def\bA{\mathbb{A}}

\def\bC{\mathbb{C}}

\def\bF{\mathbb{F}}

\def\bK{\mathbb{K}}

\def\bN{\mathbb{N}}

\def\bR{\mathbb{R}}

\def\bZ{\mathbb{Z}}
\def\cA{\mathcal{A}}
\def\cB{\mathcal{B}}
\def\cC{\mathcal{C}}

\def\cL{\mathcal{L}}

\def\cS{\mathcal{S}}

\def\cV{\mathcal{V}}

\def\cX{\mathcal{X}}

\def\frakg{\mathfrak{g}}

\def\boldA{\mathbf{A}}
\def\boldB{\mathbf{B}}

\def\boldF{\mathbf{F}}
\def\boldG{\mathbf{G}}

\def\boldS{\mathbf{S}}

\def\boldU{\mathbf{U}}
\def\boldV{\mathbf{V}}

\def\boldGamma{{\mathbf\Gamma}}
\def\boldSigma{{\mathbf\Sigma}}
\def\boldPhi{{\mathbf\Phi}}

\newcommand{\base}{B}
\newcommand{\length}[1]  {\lvert #1\rvert}
\newcommand{\norm}[1]   {\left \|#1\right \|}
\newcommand{\trpse}{{\mathrm{tr}}}
\newcommand{\Bint}[1]{(#1)_{{\base}}}

\newcommand{\Bexp}[1]{(0.#1)_{{\base}}}

\newcommand{\Ostar}{{O\tilde{\phantom{\imath}}}}

\newcommand{\lambdajsr}{\lambda_{*}}

\author{Philippe Dumas}
\date{\today}
\title[Asymptotic behaviour of radix-rational sequences]{Mean
  asymptotic behaviour of radix-rational sequences and dilation
  equations\\ (Extended version)}

\address{Algorithms Project,
INRIA Rocquencourt,
78153 Le Chesnay Cedex, France
}
\email{Philippe.Dumas@inria.fr}


\begin{document}

\begin{abstract}
  The generating series of a radix-rational sequence is a rational
  formal power series from formal language theory viewed through a
  fixed radix numeration system.  For each radix-rational sequence
  with complex values we provide an asymptotic expansion for the
  sequence of its Ces\`aro means. The precision of the asymptotic
  expansion depends on the joint spectral radius of the linear
  representation of the sequence; the coefficients are obtained
  through some dilation equations. The proofs are based on elementary
  linear algebra.
\end{abstract}

\maketitle

\tableofcontents

\section{Introduction}
Radix-rational sequences come to light in various domains of knowledge
and almost every example has been first studied for itself and by
elementary methods. As a result there does not exist a general theorem
about their asymptotic behaviour. Flajolet~\citep{FlGo94,FlGrKiPrTi94}
has developed a method based on the analytic theory of numbers. The
used frame of divide-and-conquer recurrences is wider than ours and
permits to deal with many examples, but even if the underlying idea
has a geometric evidence the application remains delicate. 

Yet these sequences are a mere generalization of classical rational
sequences, that is sequences which satisfy a linear homogeneous
recurrence with constant coefficients. In the same manner they satisfy
a linear homogeneous recurrence with constant coefficients but where
the shift $n\mapsto n+1$ is replaced by the pair of scaling
transformations $n\mapsto 2n$ and $n\mapsto 2n+1$ (or more generally
$n\mapsto\base n+r$, $0\leq r<\base$, if the used radix
is~$\base\geq2$). The first example which comes in
mind~\citep{Trollope68,Delange75} is the binary sum-of-digits
function, that is the number~$s_2(n)$ of~$1$'s in the binary expansion
of an integer~$n$. (\citet{DrGa98} give an extended bibliography.) The
sequence satisfies $s_2(0)=0$, $s_2(2n)=s_2(n)$,
$s_2(2n+1)=s_2(n)+1$. It may seem to the reader that such a sequence
is of limited interest, but it appears in many problems like the study
of the maximum of a determinant of a $n\times n$ matrix with
entries~$\pm1$~\citep{ClLi65} or in a merging process occurring in
graph theory~\citep{McIlroy74}. This example have been greatly
generalized with the number of occurrences of some pattern in the
binary code~\citep{BoCoMo89}, or in the Gray code: \citet{FlRa80}
study the average case of Batcher's odd-even merge by using the
sum-of-digits function of the Gray code.  Among the sequences directly
related to a numeration system the Thue-Morse sequence which writes
$u(n)=(-1)^{s_2(n)}$ is certainly the one which has caused the
greatest number of publications~\citep{AlSh99}. There exist variants
with some subsequences~\citep{Newman69,Coquet83} or with binary
patterns~\citep{BoCoMo89} other than the simple pattern~$1$: the
Rudin-Shapiro sequence associated to the pattern~$11$ was initially
designed to minimize the $L^{\infty}$ norm of a sequence of
trigonometric polynomials with
coefficients~$\pm1$~\citep{Rudin59,Shapiro51}.  The study of the
complexity of algorithms is another source of radix-rational
sequences. The cost~$c_n$ of computing the $n$-th power of a matrix by
binary powering satisfies $c_0=0$, $c_{2n}=c_n+1$,
$c_{2n+1}=c_n+2$. The idea of binary powering has been re-employed
by~\citet{MoOl90} in the context of computations on an elliptic curve
where the subtraction has the same cost than addition.  \citet{SuRe83}
have used the divide-and-conquer strategy to provide heuristics for
the problem of Euclidean matching and this leads them to a
$4$-rational sequence in the worst case. The theory of numbers is
another domain which provides examples like the number of odd binomial
coefficients in row~$n$ of Pascal's triangle~\citep{Stolarsky77} or
the number of integers which are a sum of three squares in the first
$n$ integers~\citep{OsSh89}.

It is the merit of \citet{AlSh92,AlSh03} to have put all these
scattered examples in a common framework. A sequence~$(u_n)$ is
rational with respect to the radix~$\base$ if all the sequences
obtained from it by applying the scaling transforms $n\mapsto\base
n+r$, $0\leq r<\base$, remain in a finite dimensional vector
space~$V$. This leads to the idea of a linear representation with
matrices of size~$d$ if~$V$ is $d$-dimensional. For the examples above
the dimension~$d$ is usually small, say from~$1$ to~$4$, but there
exist examples where $d$ is larger, like in the work
of~\citep{Cassaigne93} which uses $d=30$. In such cases a general
method of study is necessary.

Our aim is to provide asymptotic expansions for radix-rational
sequences. Because the sequences under consideration may have
a chaotic behaviour, we cannot expect an asymptotic expansion in a
usual asymptotic scale for all these sequences. We provide only an
asymptotic expansion in the mean, that is for the Ces\`aro mean
\[
\frac{1}{N}\sum_{n=0}^Nu_n
\]
if we study the sequence $(u_n)$. The way we will follow is based on
elementary linear algebra. This is natural because radix-rational
sequences are defined by linear recurrences and after all the mean
asymptotic behaviour of a radix-rational sequence originates in the
asymptotic expansion of a classical rational sequence, as we shall
prove. Eventually we obtain a general theorem, valid for all
radix-rational sequences. Evidently the asymptotic expansions which
result form this theorem are not significantly different from those of
our predecessors. Roughly speaking, they have the form
\[
\sum_{\alpha>\alpha_*,\,\ell\geq0}N^\alpha\log_\base^\ell(N)\sum_\omega
\omega^{\lfloor\log_\base N\rfloor}\Psi_{\alpha,\ell,\omega}(\log_\base N)+
O(N^{\alpha_*}),
\]
where the $\omega$'s are some complex numbers of modulus~$1$ and the
$\Psi$'s are $1$-periodic functions. The true difference lies in the
fact that the theorem has a precise framework.

The organization of the paper is as follows. In
Section~\ref{Du08vers09:sect:rrs+de} we present the definition of a
radix-rational series and a linear representation $L$, $(A_r)_{0\leq
r<\base}$, $C$ of such a series. We expose the link with the idea of a
rational formal power series in non-commutative variables from the
formal language theory.  Such a series is the main object of study and
a radix-rational series is only a metamorphosis of a formal power
series through the interpretation of integers as their radix
expansions in a given base.  We address the problem of the asymptotic
behaviour for~$K$ large of a vector-valued running sum
\[
\boldS_K(x)=\sum_{\begin{subarray}{c}\length w=K\\ (0.w)_{\base}\leq x\end{subarray}}A_wC,
\]
which depends on the length~$K$ of words and on a real number
$x\in[0,1]$ through the interpretation of words as $\base$-ary
expansions of real numbers from~$\left[0,1\right)$. We give technical
notations and hypotheses useful for this study. Particularly we recall
the idea of joint spectral radius for a family of matrices, because
the joint spectral radius~$\lambdajsr$ of a linear representation
governs the precision of the asymptotic expansion we will obtain. Next
we divide the running sum~$\boldS_K(x)$ by its dominant
behaviour. This gives a vector-valued function~$\boldF_K(x)$ and in
the limit a functional equation appears. Surprisingly we have a
dilation equation and such equations have been deeply studied because
they are useful in the theory of wavelets and in the theory of
refinement schemes.

Section~\ref{Du08vers09:sec:basiclimitthm} is the basic part of the
paper. Theorem~\ref{Du08vers09:thm:existenceF} shows that the
sequence~$(\boldF_K)$ converges uniformly towards the unique
solution~$\boldF$ of the dilation equation. Moreover in the special
case where~$C$ is an eigenvector of a certain matrix~$Q$ relative to
an eigenvalue with modulus~$\rho$ the speed of convergence is
essentially $O((\lambdajsr/\rho)^K)$.

In Section~\ref{Du08vers09:sec:fullexpansion} we elaborate on this
point by considering generalized eigenvectors.  The main result is
Theorem~\ref{Du08vers09:thm:formalasymptoticexpansion} which gives an
asymptotic expansion for~$\boldS_K(x)$ from the linear representation.
The technical hypotheses used as we work out the theorem reflect in
the error term of the expansion which is essentially
$O(\lambdajsr^K)$. This asymptotic expansion has coefficients which
are solutions of some dilation equations.

Section~\ref{Du08vers09:sec:main:periodicfunction} translates
Theorem~\ref{Du08vers09:thm:formalasymptoticexpansion} for rational
formal power series into
Theorem~\ref{Du08vers09:thm:sequenceasymptoticexpansion} for
radix-rational sequences. So we obtain an asymptotic expansion for the
running sums of a given radix-rational sequence (or for its Ces\`aro
mean, it is the same).  The solutions of the dilation equations become
functions which are practically periodic with respect to the logarithm
of the index.  The theorem even though it is general may give an
obvious formula. A typical example is the Thue-Morse sequence for
which we conclude that it is~$O(1)$. This flaw is quite normal because
the running sum of the Thue-Morse sequence is almost the sequence
itself. We illustrate the result with the study of the periodic
function which appears in the asymptotic expansion for the
Rudin-Shapiro sequence and specifically of its symmetries. We compare
our work and the work of Dumont where a similar result appears. At the
end we emphasize the fact that in full generality the coefficients of
the asymptotic expansion are not periodic functions but only
pseudo-periodic functions.

To summarize, this paper gives a theorem, which is new and general,
about the asymptotic behaviour of radix-rational sequences. Moreover
it shows a link with the well developed domain of dilation equations.

\section{Radix-rational series and dilation equations}\label{Du08vers09:sect:rrs+de}
\subsection{Radix-rational series}
For an integer~$n$ whose binary expansion is the word
$n_{K-1}\ldots n_0$, with figures $n_{K-1}$, $\ldots$, $n_0$ in
the set $\cB=\{0,1\}$, the number of ones is
\[
s_2(n)=\sum_{k\geq 0}n_k
\]
and the generating function is
\[
S_2(z)=\sum_{n\geq 0}s_2(n)z^n.
\]
It is easy to verify the following relationships
\[
\forall n\in\bN_{\geq0},\quad s_2(2n)=s_2(n),\quad s_2(2n+1)=s_2(n)+1.
\]
The constant sequence $e=1$ satisfies obviously
\[
\forall n\in\bN_{\geq0},\quad e(2n)=e(n),\quad e(2n+1)=e(n).
\]
It appears that the sequences~$s_2$ and~$e$ generate a
$\bZ_{\geq0}$-module which is left stable by the following linear
right action of the monoid of words~$\cB^*$: the image of a
sequence~$u$ under the action of a figure~$r\in\{0,1\}$ is the
sequence $u.r=(u(2n+r))$. The example leads to the following
definition. We consider a radix $\base\geq 2$ and the associated
alphabet $\cB=\{0,\ldots,\base-1\}$. For a semi-ring~$\bK$, the
monoid~$\cB^*$, equipped with the concatenation, operates on the
$\bK$-module of all sequences with values in~$\bK$ by the map
$(u(n))_{n\geq 0}\mapsto (u(\base n+r))_{n\geq 0}$ for each figure~$r$
in~$\cB$.
\begin{definition}
  A generating function $U(z)=\sum_{n\geq 0}u_nz^n$ whose coefficients
  take their values in a semi-ring~$\bK$ is $\base$-rational (or
  $\base$-recognizable, or $\base$-regular) if the $\bK$-module
  generated by the sequence under the right action of the monoid of
  words~$\cB^*$ is of finite type. It is radix-rational (or
  radix-recognizable, or radix-regular) if it is $\base$-rational for
  some radix~$\base$. The sequence of coefficients~$(u_n)$ is termed
  in the same manner.
\end{definition}
\citet{AlSh92} named such sequences of coefficients "$k$-regular
sequences" ($k$ is the radix). The adjective "regular" is reminiscent
of "regular language" for a computer scientist but it is meaningless
for a mathematician.  The expression "recognizable sequence" has the
same flaw, even if the preceding definition would be more satisfying
with recognizable in place of rational for a computer scientist. To
the contrary "radix-rational sequence" make sense for both computer
scientist and mathematician and we adopt this terminology. The name
{\em radix-rational series\/} has the merit to remind that these
series are a generalization of classical rational functions. These one
admit a representation with only one square matrix, while
$\base$-rational series admits a representation with $\base$ square
matrices, as we shall see just below.

Radix-recognizable series are a metamorphosis of recognizable formal
series. Indeed, let $e^1$, $\ldots$, $e^d$ be a family of sequences
which generates a module left stable under the action of~$\cB^*$ and
containing the sequence of coefficients~$u$ of a radix-rational
series~$U(z)$. Such a family will be called a generating family for
the sequence~$u$ in the sequel. Each of the sequences $e^j.r$, $0\leq
r<\base$, $1\leq j\leq d$, writes as a linear combination
$e^j.r=\sum_i a_{r,i,j}e^i$, not necessarily in a unique manner. The
sequence~$u$ writes $u=\sum_i c_ie^i$, again not necessarily in a
unique manner.
\begin{definition}
  The row vector $L=(e^j(0))$, the matrices $A_r=(a_{r,i,j})$, $0\leq
  r<\base$, and the column vector $C=(c_i)$, all together, is a linear
  representation of the series~$U(z)$ or of the sequence~$u$.
\end{definition}
\noindent
Besides the linear representation defines a recognizable series
over~$\cB^*$ with coefficients in~$\bK$ in the sense of formal
language theory. A formal series, usually written
\[
S=\sum_{w\in\cB^*}(S,w)\, w,
\]
is an application from~$\cB^*$ into~$\bK$, which associates to a
word~$w$ the coefficient $(S,w)$ \citep{BeRe88,Sakarovitch03}. Here we
consider the formal series~$S$ defined by
\[
(S,w)=LA_{w_1}\dotsb A_{w_{K}}C
\]
for a word $w=w_1\ldots w_{K}$. The sequence~$u$ is obtained by
restricting the formal series to the radix~$\base$ expansions of
integers.

\subsection{Notations and hypotheses}
The data is a linear representation $L$, $A=(A_x)_{x\in{\cX}}$,
$C$ of dimension~$d$ with a finite alphabet~$\cX$ of cardinality
greater or equal to~$2$. The coefficients of the matrices are taken
from the field of complex numbers in the most general case. Up to a
permutation, $\cX$ may be taken equal to $\{0,\,1,\,\ldots\,,\base -1\}$
for some integer $\base \geq 2$. The linear representation defines first a
rational formal series~$S$ and second a radix-rational series $U$, by
restriction to the radix~$\base $ expansions of the integers. To abbreviate,
we write $A_w=A_{w_1}\dotsb A_{w_{K}}$ for $w=w_1\ldots w_K$, hence
the writing $(S,w)=LA_wC$. We denote by~$u$ the sequence of
coefficients of~$U$, which means that we have $u(n)=(S,w)$ for the
integer $n=\Bint{n_{K-1}\ldots n_0}$ whose radix~$\base$ expansion
is the word $w=n_{K-1}\ldots n_0$.

For such a given $\bC$-rational formal series~$S$ we consider the
running sum over words~$w$ of length~$K$ submitted to the condition
that the real number with radix~$\base $ expansion $\Bexp w$ is not
greater than~$x$ in~$\left[0,1\right]$
\[
\boldS_K(x)=\sum_{\begin{subarray}{c}\length w = K \\ \Bexp w \leq x \end{subarray}}A_wC
\]
for all nonnegative integer~$K$ and we want to estimate its asymptotic
behaviour when~$K$ tends towards~$+\infty$. With
$Q=A_0+A_1+\dotsb+A_{\base-1}$, it writes
\begin{multline}\label{Du08vers09:eq:SKexpression}
\boldS_K(x)=
\sum_{r_1<x_1}A_{r_1}Q^{K-1}C+
\sum_{r_2<x_2}A_{x_1}A_{r_2}Q^{K-2}C+
\sum_{r_3<x_3}A_{x_1}A_{x_2}A_{r_3}Q^{K-3}C\\
\mbox{}+
\dotsb+
\sum_{r_K\leq x_K}A_{x_1}A_{x_2}\dotsb A_{r_K}C
\end{multline}
if~$x$ admits the radix~$\base$ expansion $x=(0.x_1x_2\dotsb)_\base $. The
formula renders evident the following lemma.
\begin{lemma}\label{Du08vers09:lemma:genrecrelS}
With $Q=A_0+A_1+\dotsb+A_{\base-1}$, the sequence of running
sums~$(\boldS_K)$ satisfies the recursion
\[
\boldS_{K+1}(x)=\sum_{r_1<x_1}A_{r_1}Q^{K}C+A_{x_1}\boldS_K(\base x-x_1),
\]
where~$x_1$ is the first figure in the radix~$\base $ expansion of~$x$
in~$\left[0,1\right)$, with $\boldS_0(x)=C$.
\end{lemma}

The matrix~$Q$ is the essential component which governs the mean
asymptotic behaviour of~$S$ first, and of the sequence~$u$
next. Lemma~\ref{Du08vers09:lemma:genrecrelS} leads us to consider the
powers~$Q^K$ of the matrix~$Q$. More precisely, we look at the
dominant term (from the asymptotic point of view) in the vector~$Q^KC$
and we will use the following hypothesis (named
\ref{Du08vers09:hypo:asymptbehav} for asymptotic dominant term).
\begin{hypothesis}{adt}\label{Du08vers09:hypo:asymptbehav}
  The sequence~$(Q^KC)$ admits the asymptotic expansion
\begin{equation}\label{Du08vers09:eq:asymptbehavQKC}
  Q^KC\mathop{=}_{K\to+\infty}R(K)V+O(R'(K))
\end{equation}
  with
\begin{multline}\label{Du08vers09:eq:asymptbehavR}
\frac{R(K+1)}{R(K)}\mathop{=}_{K\to+\infty}\rho\omega\left(1+O\left(\frac{1}{K}\right)\right),\qquad \rho>0,\qquad |\omega|=1,
\\\text{and}\qquad
R'(K)\mathop{=}_{K\to+\infty}R(K)O\left(\frac{1}{K}\right).
\end{multline}
\end{hypothesis}

\noindent
Necessarily the vector~$V$ is an eigenvector of~$Q$ for the
eigenvalue~$\rho\omega$.

In Section~\ref{Du08vers09:sec:fullexpansion}, we will expand the
vector~$C$ over a Jordan basis, and this leads us to consider the
following hypothesis (named
\ref{Du08vers09:hypo:generalizedeigenvector} for generalized
eigenvector). This hypothesis leads to a more particular but more
precise expression of~$Q^KC$ than
Hypothesis~\ref{Du08vers09:hypo:asymptbehav} does.
\begin{hypothesis}{gev}\label{Du08vers09:hypo:generalizedeigenvector}
The family~$(V^{(j)})_{0\leq j<\nu}$ is full rank and satisfies 
$QV^{(0)}=\rho\omega V^{(0)}$ and
$QV^{(j)}=\rho\omega V^{(j)}+V^{(j-1)}$ for $j>0$, with $\rho\geq 0$
and $|\omega|=1$.
\end{hypothesis}
\noindent
At occasion, we will say that~$\nu$ is the height of the generalized
eigenvector~$V^{(\nu-1)}$.

Formula~\eqref{Du08vers09:eq:SKexpression} use products of square
matrices~$A_r$, $0\leq r<\base$. A possible property to bound these
products is the following (named \ref{Du08vers09:hypo:roughspectralradius} for rough spectral radius).
\begin{hypothesis}{rsr}\label{Du08vers09:hypo:roughspectralradius}
  There exists an induced norm~$\norm{\ }$, and a constant~$\lambda$
  with $0<\lambda<\rho$ such that all matrices~$A_r$,
  $0\leq r<\base $, satisfy $\norm{A_r}\leq \lambda$.
\end{hypothesis}
\noindent
We will tacitly use a norm on~$\bC^d$ and the induced norms on square
matrices in order to guarantee the previous hypothesis. As a
consequence we have $\norm{A_{r_1}\dotsb A_{r_K}}\leq \lambda^K$ for
every integer~$K$.

Hypothesis~\ref{Du08vers09:hypo:roughspectralradius} will prove to be
useful, but it is not sufficiently well designed. So we refine it as
follows. We consider all the products~$A_w$ for words~$w$ of a given
length~$T$ and their norms. With the notation
\[
\lambda_T=\max_{\length w = T}\norm{A_w}^{1/T},
\]
the joint spectral radius of the set $A_r$, $0\leq r<\base$, is the
number~\citep{RoSt60,BlKaPrWi08}
\[
\lambdajsr=\lim_{T\to+\infty}\lambda_T.
\]
It is known that the joint spectral radius is not greater than any of
the numbers~$\lambda_T$. Moreover~$\lambdajsr$ is independent of the
used induced norm. The following hypothesis (named \ref{Du08vers09:hypo:jointspectralradius} for joint
spectral radius) is made to replace Hypothesis~\ref{Du08vers09:hypo:roughspectralradius}.
\begin{hypothesis}{jsr}\label{Du08vers09:hypo:jointspectralradius}
  The joint spectral radius~$\lambdajsr$ of the family of matrices
  $(A_r)_{0\leq r<\base}$ is smaller than the number~$\rho$, that is
  $\lambdajsr<\rho$.
\end{hypothesis}
\noindent
It is known that~$\lambdajsr$ is difficult to compute~\citep{TsBl97}.
For our purpose it is sufficient to find an induced norm and a~$T$ such
that $\lambda_T<\rho$. To the sake of clarity we will use a superior
index to show the used norm if necessary, like $\lambda_T^{(1)}$ if we
use the absolute maximum column norm induced by the norm of index~$1$
of~$\bC^d$. In the sequel, we will say that~$\lambdajsr$ is attained
if there exists an induced norm and an integer~$T$ such that
$\lambdajsr=\lambda_T$.

Besides, the joint spectral radius depends on the linear
representation. To each representation is associated a finite
dimensional vector subspace of the space of formal series left stable
by the operators $S\mapsto S.r^{-1}=\sum_w(S,wr)$. The representation
is reduced if the subspace is as small as
possible~\citep{BeRe88}. Evidently the smaller is the subspace the
smaller is the joint spectral radius and all reduced representations
provide the same~$\lambdajsr$ because they are isomorphic. Hence it is
better to always use a reduced representation and the joint spectral
radius associated to a reduced representation depends only on the
formal series; it is intrinsic. But it may be easier to compute the
joint spectral radius for a non reduced representation, at the risk of
obtaining a too large value.

\begin{example}[Rudin-Shapiro sequence]\label{Du08vers09:ex:RudinShapiro}
  The Rudin-Shapiro sequence may be defined as $u(n)=(-1)^{e_{2\,;
      11}(n)}$ where $e_{2\,; 11}(n)$ is the number of (possibly
  overlapping) occurrences of the pattern~$11$ in the binary expansion
  of the integer~$n$ \citep{BrCa70}. This sequence was defined
  independently by \citet{Shapiro51} and~\citet{Rudin59} to solve a
  problem of optimality about the $L^{\infty}$ norm of trigonometric
  polynomials with coefficients~$\pm1$. It is $2$-rational: it admits
  the generating family ~$(u(n),u(2n+1))$ and the reduced linear
  representation
\[
L=\left(\begin{array}{cc}
1 & 1
	\end{array}\right),\qquad
A_0=\left(\begin{array}{cc}
1 & 1 \\ 0 & 0
	\end{array}\right),\qquad
A_1=\left(\begin{array}{cc}
0 & 0 \\ 1 & -1
	\end{array}\right),\qquad
C=\left(\begin{array}{c}
1 \\ 0
	\end{array}\right).
\]
The matrix
\[
Q=A_0+A_1=\left(\begin{array}{cc}
1 & 1 \\ 1 & -1
	\end{array}\right)
\]
has two eigenvalues~$\pm\sqrt 2$ which have the same absolute
value. The vector $Q^KC$ has expression
\[
Q^KC=
\frac{2^{K/2}}{4}
\left(\begin{array}{cc}
(2+2^{1/2})+(-1)^K(2-2^{1/2})\\
2^{1/2}(1+(-1)^K)
	\end{array}\right)
\]
and Hypothesis~\ref{Du08vers09:hypo:asymptbehav} is not satisfied. The
sequence is $4$-rational too with representation~\citep{AlSh03},
relative to the generating family $(u(n),u(4n+2))$,
\begin{multline*}
L=\left(\begin{array}{cc}
1 & 1
	\end{array}\right),\qquad
A^{(2)}_0=A_0^2=\left(\begin{array}{cc}
1 & 1 \\ 0 & 0
	\end{array}\right),\qquad
A^{(2)}_1=A_0A_1=\left(\begin{array}{cc}
1 & -1 \\ 0 & 0 
	\end{array}\right),\\
A^{(2)}_2=A_1A_0=\left(\begin{array}{cc}
0 & 0 \\ 1 & 1
	\end{array}\right),\qquad
A^{(2)}_3=A_1^2=\left(\begin{array}{cc}
0 & 0 \\ -1 & 1
	\end{array}\right),\qquad
C=\left(\begin{array}{c}
1 \\ 0
	\end{array}\right).
\end{multline*}
The matrix $Q^{(2)}=A_0^2+A_0A_1+A_1A_0+A_1^2=(A_0+A_1)^2=Q^2$ has a
dominant eigenvalue~$\rho=2$ and it is evident that
Hypothesis~\ref{Du08vers09:hypo:asymptbehav} is satisfied for this
linear representation. Because~$A_0$ and~$A_1$ maps each vector of the
canonical basis onto a vector of the canonical basis or its negative,
we have $\lambda_T=1$ for every $T\geq 1$ if we use a norm which gives
the same value for both vectors of the canonical basis. Hence
Hypothesis~\ref{Du08vers09:hypo:roughspectralradius} is satisfied for the
radix~$4$ representation above and
Hypothesis~\ref{Du08vers09:hypo:jointspectralradius} too because
$\lambdajsr=1$.
\end{example}

\subsection{Self-similarity}
For any positive integer~$K$, let us introduce the
function~$\boldF_K$ from~$[0,1]$ into~$\bC^d$ defined by
\[
\boldF_K(x)=\frac{1}{R(K)}\boldS_K(x).
\]
Lemma~\ref{Du08vers09:lemma:genrecrelS} translates into the equation
\begin{equation}\label{Du08vers09:eq:genrecrelF}
\boldF_{K+1}(x)=\frac{1}{R(K+1)}\sum_{r_1<x_1}A_{r_1}Q^{K}C+\frac{R(K)}{R(K+1)}A_{x_1}\boldF_K(\base x-x_1).
\end{equation}
We may consider the operator $\cL_K$ of the space of continuous
functions from~$[0,1]$ into $\bC^d$ defined by
\[
\cL_K \Phi(x)=\frac{1}{R(K+1)}\sum_{r_1<x_1}A_{r_1}Q^{K}C+\frac{R(K)}{R(K+1)}A_{x_1}\Phi(\base x-x_1).
\]
(In all the paper $x_1$ is the first digit of~$x\in\left[0,1\right)$.) 
Equation~\eqref{Du08vers09:eq:genrecrelF} rewrites
\[
\boldF_{K+1}(x)=\cL_K \boldF_K(x).
\]
According to Hypothesis~\ref{Du08vers09:hypo:asymptbehav}, the
sequence of operators~$\cL_K$ converges weakly towards the
operator~$\cL$ defined by
\begin{equation}\label{Du08vers09:def:operatorL}
\cL \Phi(x)=\frac{1}{\rho\omega}\sum_{r_1<x_1}A_{r_1}V+\frac{1}{\rho\omega}A_{x_1}\Phi(\base x-x_1)
\end{equation}
and we will first study the equation $\cL \Phi=\Phi$. 

In the sequel, the following system of equations, whose unknown is a
function from the segment~$[0,1]$ into the space~$\bC^d$,
\begin{itemize}
\item[--] $\Phi(0)=0$, $\Phi(1)=V$, 
\item[--] for every figure~$r$ of the radix~$\base$ system 
  and for~$x$ in $\left[r/\base ,(r+1)/\base \right)$,
\begin{equation}\label{Du08vers09:eq:dilationequation}
  \Phi(x)=\frac{1}{\rho\omega}\sum_{r_1<r}A_{r_1}V+\frac{1}{\rho\omega}A_{r}\Phi(\base x-r)
\end{equation}
\end{itemize}
will be named {\em the basic dilation equation}.
\begin{proposition}\label{Du08vers09:prop:basicprop}
  Let $L$, $(A_r)_{0\le r<\base}$, $C$ be a linear representation of
  dimension~$d$ for the radix~$\base$. Under
  Hypotheses~\ref{Du08vers09:hypo:asymptbehav}
  and~\ref{Du08vers09:hypo:roughspectralradius} the problem
  \begin{itemize} 
\item[--] $\Phi$ is a continuous function from the
  segment~$[0,1]$ into the space~$\bC^d$, 
\item[--] $\Phi$ is a solution of the basic dilation equation,
  \end{itemize}
  has a unique solution~$\boldF$.
\end{proposition}

Before we go further we have to do a simple remark. If we multiply the
square matrices~$A_r$, $0\leq r<\base$, by a nonzero scalar~$\alpha$,
the eigenvalue~$\rho\omega$ is multiplied by~$\alpha$, the running
sum~$\boldS_K(x)$ and the dominant asymptotic behaviour~$R(K)$ are
multiplied by~$\alpha^K$. It follows that~$\boldF_K$ and the
operator~$\cL_K$ are unchanged. In the same manner the operator~$\cL$
is not modified and the solution of the fixed point problem described
in the previous proposition remains the same. This permits us to consider
that the modulus~$1$ number~$\omega$ is equal to~$1$ in order to prove
the assertions. Even we may assume that $\rho\omega$ is equal to~$1$
but we prefer to keep in mind the modulus~$\rho$ of the
eigenvalue~$\rho\omega$. (Nevertheless see the next subsection.)

\begin{proof}
  According to the previous remark, we may assume $\omega=1$ by
  multiplying all the square matrices ~$A_r$, $0\leq r<\base$, by the
  conjugate number~$\overline \omega$.

  The space of continuous functions from $[0,1]$ into the
  space~$\bC^d$, equipped with the norm of the maximum $\norm \Phi
  _{\infty}=\max_x \norm{\Phi(x)}$, is a complete normed
  space. (Recall that Hypothesis~\ref{Du08vers09:hypo:roughspectralradius}
  assumes that we have chosen a norm on~$\bC^d$.) The continuous
  functions~$\Phi$ which satisfy $\Phi(0)=0$ and $\Phi(1)=V$ are the
  elements of a closed, hence complete, subspace~$\cC$ of this
  complete space. The equation of the problem appears as a fixed point
  equation $\Phi= {\mathcal L}\Phi$. It is sufficient to see that the
  subspace~$\cC$ is left invariant by~$\mathcal L$ and that~$\mathcal
  L$ is a contraction to prove the assertion.

  We must verify that for a continuous~$\Phi$ in~$\cC$ the function
  $\Psi=\cL \Phi$ is a member of~$\cC$. According to the piecewise
  definition of~$\Psi$, we have to consider the left and right limits
  of~$\Psi$ at the points~$r/\base$ for $0\leq r\leq \base$. The
  definition of~$\Psi$ and the continuity of~$\Phi$ give immediately
\begin{multline*}
\Psi(0)=\frac{1}{\rho}A_0\Phi(0)=0,\\
\Psi(1-0)=\frac{1}{\rho}\sum_{0\leq
  s<{\base'}}A_sV+\frac{1}{\rho}A_{\base'}V=\frac{1}{\rho}QV=V,\quad
\Psi(1)=V,
\end{multline*}
and for $0<r<\base$
\begin{multline*}
\Psi(\frac{r}{\base}+0)=\Psi(\frac{r}{\base})=\frac{1}{\rho}\sum_{0\leq s<r}A_sV,\\
\Psi(\frac{r}{\base}-0)=\frac{1}{\rho}\sum_{0\leq
  s<r-1}A_sV+\frac{1}{\rho}A_s\Phi(1)=\frac{1}{\rho}\sum_{0\leq
  s<r}A_sV.
\end{multline*}
  The constraints are satisfied and the subspace~$\cC$ is stable. 

  If we have two functions~$\Phi_1$ and~$\Phi_2$ in~$\cC$,
  let~$\Psi_1$ and~$\Psi_2$ be their images by~$\cL$. From
  $\Psi_1(x)-\Psi_2(x)=(1/\rho)A_r(\Phi_1(\base x-r)-\Phi_2(\base
  x-r))$ follows the inequality $\norm{\Psi_1-\Psi_2}\leq
  (\lambda/\rho)\norm{\Phi_1-\Phi_2}$ because the norm of each
  matrix~$A_r$ is bounded by~$\lambda$. As a consequence of
  Hypothesis~\ref{Du08vers09:hypo:roughspectralradius} the
  operator~$\mathcal L$ is a contraction.
\end{proof}

In fact Hypotheses~\ref{Du08vers09:hypo:asymptbehav}
and~\ref{Du08vers09:hypo:roughspectralradius} are not necessary to
conclude that the solution of the basic dilation equation is
unique. A drop of regularity is sufficient.
\begin{lemma}\label{Du08vers09:lemma:functequnicity}
The basic dilation equation may have only one solution under the
sole hypothesis $\rho>0$, $|\omega|=1$ if we ask for a function
continuous on the right.
\end{lemma}
\begin{proof}
  For a $\base$-adic number, that is a number~$x$ whose radix~$\base$
  expansion is finite, the dilation equation gives the value~$F(x)$
  because the recursion formula provided by the dilation equation
  terminates on the basic case $x=0$. But the set of $\base$-adic
  numbers is dense in~$[0,1]$ and the function is assumed to be right
  continuous at every point. Hence it is completely determined.
\end{proof}

\begin{example}[Regular self-similar functions]\label{Du08vers09:ex:regularF}
  Often the functional equations
  like~\eqref{Du08vers09:eq:dilationequation} are presented as leading
  necessarily to chaotic functions, but some very regular functions
  may satisfy such a self-similarity property. As an example consider
  the sequence with takes value~$1$ on all the multiple of~$3$ and~$0$
  for the others integers. It is rational in the classical sense and
  its generating function is rational with poles which are all roots
  of the unity. As a consequence~\citep[Th.~16.4.3, p.~446]{AlSh03}
  the sequence is rational with respect to every radix. Here is a
  binary linear representation~\citep[p.~6]{Sakarovitch03},
\[
L=\left(\begin{array}{ccc}1&0&0\end{array}\right),\quad
A_0=\left(\begin{array}{ccc}1&0&0\\0&0&1\\0&1&0\end{array}\right),\quad
A_1=\left(\begin{array}{ccc}0&1&0\\1&0&0\\0&0&1\end{array}\right),\quad
C=\left(\begin{array}{c}1\\0\\0\end{array}\right).
  \]
  The system satisfied by~$\boldF$ is
\[
\left\{
\begin{array}{l}
  F_1(x)=\displaystyle\frac{1}{2}F_1(2x),\\[1.5ex]
  F_2(x)=\displaystyle\frac{1}{2}F_3(2x),\\[1.5ex]
  F_3(x)=\displaystyle\frac{1}{2}F_2(x),
\end{array}\right.
\;\text{for $0\leq x<\frac{1}{2}$;}\quad
\left\{
\begin{array}{l}
  F_1(x)=\displaystyle\frac{1}{2}+\frac{1}{2}F_2(2x-1),\\[1.5ex]
  F_2(x)=\displaystyle\frac{1}{2}+\frac{1}{2}F_1(2x-1),\\[1.5ex]
  F_3(x)=\displaystyle\frac{1}{2}+\frac{1}{2}F_3(2x-1).
\end{array}\right.
\; \text{for $0\leq x<1$.}
\]
  We find immediately the solution
  $\boldF(x)=\left(\begin{array}{ccc}x&x&x\end{array}\right)^{\trpse}$.
\end{example}
\citet{DiHa99} have studied the distributional solutions with a 
bounded support of a dilation equation. The idea is to use an
antiderivative of sufficiently high order~$n$. This makes contracting
the operator behind the equation. 

\subsection{Wavelets and refinement schemes}
The basic dilation equation enters into the domain of what is called
a {\em two-scale difference equation}, namely
\[
\varphi(x)=\sum_{n\in{\bZ}}c_n\varphi(2x-n).
\]
(For the sake of simplicity, we limit ourselves to radix $\base=2$ in
this subsection.)  These equations have been heavily studied because
they appear in the theory of wavelets to define a {\em scale function\/}
and in the theory of refinement schemes of computer graphics to define
a {\em refinement function\/}. \citet{DaLa91} provide an expository of
their occurrences and a bibliography. See also~\citep{HeCo96}.

For compactly supported wavelets the previous sum is finite and the
equation may be rewritten in a way which looks like our dilation
equation. We follow~\citep[p.1036]{DaLa92}
or~\citep[p.~235]{Daubechies92}. The equation 
\[
\varphi(x)=\sum_{n=0}^N c_n\varphi(2x-n)
\]
is translated into an equation
\[
\Phi(x)=\left\{
\begin{array}{ll}
T_0\Phi(2x) & \text{if $0\leq x\leq 1/2$,}\\
T_1\Phi(2x-1)& \text{if $1/2\leq x\leq 1$,}
\end{array}
\right.
\] 
with a vector-valued unknown function
\[
\Phi(x)=\left(\begin{array}{ccccc}\varphi(x)&\varphi(x+1)&
    \varphi(x+2) & \ldots & \varphi(x+N-1)
  \end{array}\right)^{\trpse}
\]
and square matrices $T_0=(c_{2i-j-1})_{1\leq i,j\leq N}$,
$T_1=(c_{2i-j})_{1\leq i,j\leq N}$ (it is assumed $c_n=0$ for~$n$
outside $[0,N]$).  Scalars~$c_n$ are constrained by
$\sum_mc_{2n}=\sum_nc_{2n+1}=1$ and the number~$1$ is an eigenvalue
for both~$T_0$ and~$T_1$, and for $M=(c_{2i-j})_{1\leq i,j\leq N-1}$.
We consider a right eigenvector~$W=(w_i)_{1\leq i\leq N-1}$ for~$M$
(an additional condition imposes that~$1$ is a simple eigenvalue, so
there is essentially only one possibility for~$W$). We extend~$W$ by
$w_0=0$, $w_N=0$ and define~$V=(v_i)_{1\leq i\leq N}$ by $v_i=w_{i-1}$
for $1\leq i\leq N$, so that~$W$ is an eigenvector for~$T_0$ and~$V$
is an eigenvector for~$T_1$ relative to the eigenvalue~$1$. Boundary
conditions are added to the equation, namely $\Phi(0)=W$, $\Phi(1)=V$,
which guarantee the well definition of the equation and the continuity
of the solution~$\varphi$ (under some conditions on the~$c_n$'s). 

Besides our dilation equation writes
\[
\boldF(x)=\left\{
\begin{array}{ll}
T_0\boldF(2x) & \text{if $0\leq x\leq 1/2$,}\\
T_0V+T_1\boldF(2x-1)& \text{if $1/2\leq x\leq 1$,}
\end{array}
\right.
\]
where matrices~$T_0$ and~$T_1$ are defined by
\[
T_0=\frac{1}{\rho\omega}A_0,\qquad T_1=\frac{1}{\rho\omega}A_1,
\]
and~$V$ is an eigenvector for $T_0+T_1$  relative to the
eigenvalue~$1$. Here the boundary conditions are $\boldF(0)=0$ and
$\boldF(1)=V$. 

Evidently both equations are very akin. The two-scale difference
equation for the scale function of wavelets is homogeneous while our
dilation equation is not, but their linear parts are the
same. Inhomogeneous dilation equations have been
studied~\citep{StZh98}, because they are useful in the construction of
wavelets on a finite interval and of multiwavelets.  Nevertheless this
contrast is an illusion: if we extend the function~$\boldF$ as a
continuous function over the whole real line by making it constant on
$\left(-\infty,0\right]$ with value~$0$ and constant
on~$\left[1,+\infty\right)$ with value~$V$, the dilation equation
rewrites
\[
\boldF(x)=T_0\boldF(2x)+T_1\boldF(2x-1)
\]
for~$x$ real. The writing of the dilation equation by cases is more
concrete, but the homogeneous version above is more compact and more
practical for proofs. So we will made use of the following convention
(named \ref{Du08vers09:hypo:homogequconv} for homogeneous equation
convention) at occasion.
\begin{convention}{hec}\label{Du08vers09:hypo:homogequconv}
  The solution of a basic dilation equation is extended to the whole
  real line as a continuous function constant on the left of~$0$ and
  on the right of~$1$.
\end{convention}

The eigenvalue~$1$ appears in both cases. The boundary conditions are
not of the same form and the matrices~$T_0$ and~$T_1$ for the wavelets
have a very special structure, while the matrices~$A_0$ and~$A_1$ of a
linear representation are not constrained in our study. As a result,
even if the computations are not exactly the same, the ideas which
work for wavelets work too for rational series. For example, the basic
idea of the {\em cascade algorithm}~\citep[\S~6.5,
p.~206]{Daubechies92} or of the refinement schemes~\citep{DyLe02},
which computes the value of a scale function for dyadic numbers,
applies here. We have yet used it in
Lemma~\ref{Du08vers09:lemma:functequnicity}, and evidently to draw the
pictures in the paper.  In the same manner, the key point for the
existence and uniqueness of the solution of these two-scale difference
equations is the occurrence of a contracting
operator~\cite[Sec.~4]{DaLa91}. The same idea have appeared before in
\citep{hutchinson81:_fract_self_simil} which describes a construction
of self-similar parameterized curves (and the construction of the
sequence~$(\boldG_K)$ in the proof of
Lemma~\ref{Du08vers09:thm:Holderbasic} below is of the type described
in its~\S~3.5).

\section{Basic limit theorem}\label{Du08vers09:sec:basiclimitthm}
\subsection{Uniform convergence}\label{Du08vers09:subsec:unifconvergence}
We are now in position to prove a first result about the asymptotic
behaviour of the running sums $\boldS_K$.
\begin{proposition}\label{Du08vers09:thm:basicproposition}
Under Hypotheses~\ref{Du08vers09:hypo:asymptbehav}
and~\ref{Du08vers09:hypo:roughspectralradius}, the sequence~$(\boldF_K)$
converges uniformly towards the function~$\boldF$ defined in
Proposition~\ref{Du08vers09:prop:basicprop}.
\end{proposition}

\begin{proof}
To obtain a uniform convergence, we take care that all the big oh with
respect to~$K$ are uniform with respect to~$x$. Also we may assume
$\omega=1$ as in the proof of
Proposition~\ref{Du08vers09:prop:basicprop}. 

We first note that the sequence~$(\boldF_K(x))$ is uniformly
bounded. Actually let~$\varepsilon>0$ chosen such that 
\[
k=\frac{\lambda}{\rho}\frac{1}{1-\varepsilon}<1.
\]
There exists a~$K_0$ such that for $K\geq K_0$ we have the inequality
\[
\left|\frac{1}{\rho}\frac{R(K+1)}{R(K)}-1\right|<\varepsilon.
\]
The triangular inequality applied to the right member
of Formula~\eqref{Du08vers09:eq:genrecrelF} provides an inequality
\[
\norm{\boldF_{K+1}}_{\infty}\leq 
k\left(\gamma+\norm{\boldF_K}_{\infty}\right),
\]
where~$\gamma$ is a constant. By induction we obtain
\[
\norm{\boldF_K}_{\infty}\leq \frac{k\gamma}{1-k}+\norm{\boldF_{K_0}}_{\infty}
\]
for~$K\geq K_0$. Hence the sequence~$(\norm{\boldF_K}_{\infty})$ is
bounded.

Formula~\eqref{Du08vers09:eq:genrecrelF} and
Formula~\eqref{Du08vers09:eq:dilationequation} (applied
to~$\boldF(x)$) provide 
\begin{multline}\label{Du08vers09:eq:ineq0}
\boldF_{K+1}(x)-\boldF(x)=
\sum_{r_1<x_1}A_{r_1}\left(
\frac{1}{R(K+1)}Q^KC-\frac{1}{\rho}V
\right)\\
+
A_{x_1}\left(\frac{R(K)}{R(K+1)}\boldF_K(\base x-x_1)-\frac{1}{\rho}\boldF(\base x-x_1)\right).
\end{multline}
The writing $Q^KC=R(K)V+R(K)O(1/K)$ gives
\begin{multline*}
\boldF_{K+1}(x)-\boldF(x)=
\sum_{r_1<x_1}A_{r_1} O\left(\frac{1}{K}\right)+
A_{x_1}\boldF_K(\base x-x_1)O\left(\frac{1}{K}\right)\\+
\frac{1}{\rho}A_{x_1}\left(\boldF_K(\base x-x_1)-\boldF(\base x-x_1)\right)
\end{multline*}
and next, with the uniformly bounded character of~$(\boldF_K(x))$,
\[ 
\norm{\boldF_{K+1}-\boldF}_{\infty}\leq 
O\left(\frac{1}{K}\right)+\frac{\lambda}{\rho}\norm{\boldF_K-\boldF}_{\infty}.
\] 
By induction, we find 
$\norm{\boldF_K-\boldF}_{\infty}=O\left({1}/{K}\right)$.
If the constant implied in the big oh of the last but one formula
is~$\gamma$, we may take $2\gamma/(1-\lambda/\rho)$ for the constant
in the big oh of the last formula.
\end{proof}

We may be more precise about the speed of convergence, but we content
ourselves with a particular case (which will prove to be basic).
\begin{corollary}\label{Du08vers09:cor:eigenvector}
  Under Hypotheses~\ref{Du08vers09:hypo:asymptbehav},
  \ref{Du08vers09:hypo:roughspectralradius}, if~$C$ is an eigenvector
  associated to the eigenvalue~$\rho\omega$, the speed of convergence
  is $\norm{\boldF_K-\boldF}_{\infty}=O((\lambda/\rho)^K)$.
\end{corollary}
\begin{proof}
  We do not assume here $|\omega|=1$ hence $\rho$ must be replaced by
  $\rho\omega$. 
  By substituting $Q^KC=R(K)V+O(R'(K))$ into
  Formula~\eqref{Du08vers09:eq:ineq0}, we obtain 
\begin{equation}\label{Du08vers09:eq:speed}
\norm{\boldF_{K+1}-\boldF}_{\infty}\leq 
O\left(\frac{R(K)}{R(K+1)}-\frac{1}{\rho\omega}\right)
+
O\left(\frac{R'(K)}{R(K+1)}\right)
+
\frac{\lambda}{\rho}\norm{\boldF_K-\boldF}_{\infty}.
\end{equation}
  But if~$C$ is an eigenvector,  we have
  $R(K+1)/R(K)=\rho\omega$ and $R'(K)=0$ and 
  Formula~\eqref{Du08vers09:eq:speed} becomes
  $\norm{\boldF_{K+1}-\boldF}_{\infty}\leq
  (\lambda/\rho)\norm{\boldF_K-\boldF}_{\infty}$, hence the conclusion.
\end{proof}

\subsection{Basic theorem}
Let us assume that we group the letters into pairs, which means that
we consider only words of even length. We obtain a new formal
series~$S_2$. It is rational and admits a linear representation whose
square matrices are the~$\base^2$ products $A_rA_s$ with $0\leq
r,s<\base$. The associated matrix~$Q$ becomes $Q_2=Q^2$. (See
Ex.~\ref{Du08vers09:ex:RudinShapiro}.) The sequence of running
sums~$(\boldS_K)$ is changed into its subsequence~$(\boldS_{2K})$ and
it is the same thing for the sequence~$(\boldF_K)$. In the same manner
we may group the letters by~$T$ for a given~$T$. This leads us to
consider the subsequence~$(\boldF_{KT})$ and the power~$Q^T$.

In order that Proposition~\ref{Du08vers09:thm:basicproposition}
applies to a power of matrix~$Q$, we have to consider all
products~$A_w$ for words~$w$ of a given length and their
norms. Hypothesis~\ref{Du08vers09:hypo:roughspectralradius} will be
satisfied for some power~$Q^T$ if we impose
Hypothesis~\ref{Du08vers09:hypo:jointspectralradius}. If~$\varepsilon>0$
is chosen such that $\lambdajsr(1+\varepsilon)<\rho$ and~$T$ is
sufficiently large, we thus have $\norm{A_w} \leq (\lambda_T)^T\leq
\lambdajsr^T(1+\varepsilon)^T<\rho^T$ for all words~$w$ of
length~$T$. Hence the subsequence~$(\boldF_{KT})$ is convergent. We
will show that not only this subsequence is convergent but the
sequence~$(\boldF_K)$ is convergent.
\begin{theorem}\label{Du08vers09:thm:existenceF}
Under Hypotheses~\ref{Du08vers09:hypo:asymptbehav}
and~\ref{Du08vers09:hypo:jointspectralradius}, the
sequence~$(\boldF_K)$ converges uniformly towards the unique
continuous solution~$\boldF$ of the basic dilation
equation~\eqref{Du08vers09:eq:dilationequation}.
\end{theorem}
\begin{proof}
  As we have explained just before the wording of the theorem, the
  sequence $(\boldF_{KT})$ converges uniformly to a
  function~$\boldF^0$ for some~$T$. Let us consider the
  sequence~$(\boldF_{KT+1})$. According to
  Lemma~\ref{Du08vers09:lemma:genrecrelS}, we have
\[
\frac{1}{R(KT+1)}\boldS_{KT+1}(x)=
\frac{1}{R(KT+1)}\sum_{r_1<x_1}A_{r_1}Q^{KT}C
+\frac{1}{R(KT+1)}\boldS_{KT}(\base x-x_1),
\]
hence
\begin{multline*}
\boldF_{KT+1}(x)=
\frac{1}{R(KT+1)}\sum_{r_1<x_1}A_{r_1}\left(R(KT)V+O(R'(KT))\right)\\+
\frac{R(KT)}{R(KT+1)}\boldF_{KT}(\base x-x_1).
\end{multline*}
Because~$(\boldF_{KT})$ converges uniformly towards ~$\boldF^0$, we
see that~$(\boldF_{KT+1})$ converges uniformly towards $\boldF^1=\cL
\boldF^0$. Repeating the argument, we conclude that each
subsequence~$(\boldF_{KT+s})$ converges uniformly towards
$\boldF^s=\cL^s \boldF^0$ for $s\geq 1$. But each of these functions
satisfies $\boldF^s=\cL^T \boldF^s$, because~$\boldF^0$ is a fixed
point of~$\cL^T$. Since the unique solution of this equation is
$\boldF^0$, all these functions are equal. Henceforth the sequence
$(\boldF_K)$ is uniformly convergent towards $\boldF^0$. Because of
the equality $\boldF^0=\boldF^1$ the function $\boldF^0$ is the unique
solution of the equation $\Phi =\cL\Phi$. (The uniqueness of the solution
for the last equation is guaranteed by
Lemma~\ref{Du08vers09:lemma:functequnicity}.)
\end{proof}

The following case will prove to be useful in
Section~\ref{Du08vers09:sec:fullexpansion}, where it will be extended.
\begin{corollary}\label{Du08vers09:cor:eigenvectorwitherror}
Under Hypotheses~\ref{Du08vers09:hypo:asymptbehav}
and~\ref{Du08vers09:hypo:jointspectralradius}, if~$C$ is an
eigenvector associated to the eigenvalue~$\rho\omega$, the speed of
convergence is
$\norm{\boldF_K-\boldF}_{\infty}=O((\lambda/\rho)^K)$
for every $\lambda>\lambdajsr$. If~$\lambdajsr$ is attained we may
replace~$\lambda$ by~$\lambdajsr$.
\end{corollary}
\begin{proof}
  Corollary~\ref{Du08vers09:cor:eigenvector} gives the speed of
  convergence for the subsequence~$(\boldF_{KT})$ and we obtain
  $\norm{\boldF_{KT}-\boldF}_{\infty}=O(r_T^K)$ with
  $r_T=\lambda_T^T/\rho^T$ because we use the power~$Q^T$ in place
  of~$Q$. This gives $\norm{\boldF_{KT+s}-\boldF}_{\infty}=O(r_T^K)$
  for $0\leq s<T$ because the operator~$\cL$ is continuous and because
  there is a finite number of subsequences to consider. We obtain
  $\norm{\boldF_K-\boldF}_{\infty}=O(r_T^{K/T})$ and the formula above
  because~$T$ was chosen such that $\lambda_T\leq
  \lambdajsr(1+\varepsilon)$ for an arbitrary $\varepsilon >0$. In
  case $\lambdajsr$ is attained the positive number~$\varepsilon$ is
  useless.
\end{proof}

\begin{example}[Worst mergesort sequence]
  Mergesort is a comparison based sorting algorithm which uses the
  recursive divide-and-conquer strategy. The list to be sorted is
  split into two lists of almost equal size; both are sorted by
  mergesort (there is nothing to do for a list with only one item);
  both sorted lists are merged. Taking into account the number of
  comparisons, the cost of mergesort for a list of $n$ items is
\[
u_n=u_{\lfloor n/2\rfloor}+u_{\lceil n/2\rceil}+m_n,\qquad u_1=0,
\]
where $m_n$ is the cost of merging two sorted lists with $\lfloor
n/2\rfloor$ and $\lceil n/2\rceil$ elements~\citep{FlGo94}.  The cost
of mergesort in the worst case (that is $m_n=n-1$) writes
$u_n=\sum_{k=1}^n\lceil \log_2k\rceil$ and the sequence $(v_n)$
defined by $v_0=0$ and $v_n=\lceil \log_2n\rceil$ for $n\geq1$ is a
$2$-rational sequence which admits the representation, relative to the
basis $(v(n)$, $v(2n)$, $v(2n+1)$, $v(4n+2))$,
\[
L=\left(\begin{array}{c}
0\\\noalign{\medskip}0\\\noalign{\medskip}0\\\noalign{\medskip}1
	\end{array}\right)^\trpse,\;\;
A_0=\left(\begin{array}{cccc}
0&-1&-1&-1\\\noalign{\medskip}1&2&1&1
\\\noalign{\medskip}0&0&1&0\\\noalign{\medskip}0&0&0&1
	\end{array}\right),\;\;
A_1=\left(\begin{array}{cccc}
0&0&1&1\\\noalign{\medskip}0&0&-1&-1
\\\noalign{\medskip}1&0&-1&-2\\\noalign{\medskip}0&1&2&3
	\end{array}\right),\;\;
C=\left(\begin{array}{c}
1\\\noalign{\medskip}0\\\noalign{\medskip}0\\\noalign{\medskip}0
	\end{array}\right)
\]
The matrix $Q=A_0+A_1$ has two eigenvalues, namely~$2$ and~$1$,
which are double. A computation, based on the
Jordan reduced form of~$Q$, gives
\[
Q^KC=2^KK
\left(\begin{array}{c}
0\\
0\\
-1\\
1
	   \end{array}\right)
+
2^K
\left(\begin{array}{c}
0\\
0\\
2\\
-1
	   \end{array}\right)
+
K
\left(\begin{array}{c}
-1\\
1\\
1\\
-1
	   \end{array}\right)
+
\left(\begin{array}{c}
1\\
0\\
-2\\
1
	   \end{array}\right)
\]
and Hypothesis~\ref{Du08vers09:hypo:asymptbehav} works with
\[
\rho=2,\qquad\omega=1,\qquad
R(K)=2^KK,\qquad
V=\left(\begin{array}{ccccc}0&0&-1&1\end{array}\right)^{\trpse},\qquad
R'(K)=2^K.
\]
With the absolute maximum column norm, 
we find that the maximal norm of the matrices~$A_w$ with~$w$ of
length~$4$ is $9$. This gives $\lambda^{(1)}_4\simeq 1.73$ and
Hypothesis~\ref{Du08vers09:hypo:jointspectralradius} is satisfied. As
a consequence $(\boldF_{K})$ converges uniformly towards~$\boldF$
which is nothing but
\[
\boldF(x)=\left(\begin{array}{cccc}0&0&-x&x\end{array}\right)^{\trpse}.
\]
\end{example}

\subsection{H\"older property}\label{Du08vers09:subsec:Holder}
Let us recall that a function~$f$ from the real line into a normed
space is H\"older with exponent~$\alpha$ if it satisfies
$\norm{f(y)-f(x)}\leq c |y-x|^{\alpha}$ for some constant~$c$.  Under
the hypotheses of Theorem~\ref{Du08vers09:thm:existenceF} the
function~$\boldF$ is not only continuous but H\"older. Such a result
is classical in the study of dilation equations and the proof below is
a replay of point~$(7)$ in the proof of Th.~2.2
from~\citep{DaLa92}. See also~\citep{Rioul92}.
\begin{lemma}\label{Du08vers09:thm:Holderbasic}
  Under Hypotheses~\ref{Du08vers09:hypo:asymptbehav} and~\ref{Du08vers09:hypo:roughspectralradius} the function~$\boldF$ is
  H\"older with exponent~$\log_{\base}(\rho/\lambda)$.
\end{lemma}
The inequality $\lambda<\rho$ is assumed, but we have also $\rho\leq
\norm{Q}\leq \base\lambda$ by triangular inequality. This gives bounds
$0<\log_{\base}(\rho/\lambda)\leq 1$ for the exponent, as expected.
\begin{proof}
  Let us introduce the sequence~$\boldG_K$ defined by $\boldG_0(x)=xV$
  and $\boldG_{K+1}=\cL\boldG_K$, where~$\cL$ is the operator defined
  by Eq.~\eqref{Du08vers09:def:operatorL} and used in the basic
  dilation equation~\eqref{Du08vers09:eq:dilationequation}. The
  function~$\boldG_0$ is linear and thanks to the properties of the
  operator~$\cL$, all the functions~$\boldG_K$ are piecewise linear
  and continuous. More precisely~$\boldG_K$ is linear in each of the
  interval $[k/\base^K,(k+1)/\base^K]$, $0\leq k<\base^K$. 

  Let~$s$ and~$t$ be two real numbers which are situated in one of
  these intervals $\left[k/\base^K,(k+1)/\base^K\right)$, $0\leq
  k<\base^K$, and satisfy $1/\base^{K+1}\leq |t-s|<1/\base^K$. Their
  radix~$\base $ expansions have the same first $K$ figures, but the
  $(K+1)$th figures are different. We may write $s=\Bexp {u_1\ldots
    u_Ks_1\ldots}$, $t=\Bexp {u_1\ldots u_Kt_1\ldots}$ and $s'=\Bexp
  {s_1\ldots}$, $t'=\Bexp {t_1\ldots}$. According to the definition
  of~$\cL$, we have
\[
\boldG_K(t)-\boldG_K(s)=\frac{1}{(\rho\omega)^K}A_{u_1}\dotsb A_{u_K}(t'-s')V.
\]
  Consequently we obtain
\[
\norm{\boldG_K(t)-\boldG_K(s)}\leq
\left(\frac{\lambda}{\rho}\right)^K\norm{V}\leq
\base \norm{V} \left(\frac{\base\lambda}{\rho}\right)^K|t-s|,
\]
using the hypothesis $1/\base^{K+1}\leq |t-s|$.  Because $\boldG_K$ is
piecewise linear this inequality is valid for all pairs $(s,t)$ taken
from~$[0,1]$. Let us insist on that point. The function
$\Delta:\,(s,t)\mapsto \norm{\boldG_K(t)-\boldG_K(s)}/|t-s|$ is
defined on the square $[0,1]^2$ minus its diagonal. It gives the mean
speed of the parameterized curve $\boldG_K$ on the interval whose ends
are~$s$ and~$t$. According to the previous computation the quantity
$\base \norm{V}(\base\lambda/\rho)^K$ is an upper bound for the mean
speed~$\Delta$ on each square $\left[k/\base^K,(k+1)/\base^K\right)^2$
(minus their diagonals), $0\leq k<\base^K$, because $\boldG_K$ is
piecewise linear and~$\Delta$ is constant on such a square. Again
because $\boldG_K$ is piecewise linear these particular squares give
the larger value of~$\Delta$. Henceforth the previous upper bound is
valid on the entire square~$[0,1]^2$ (minus its diagonal).

Let~$x$ and~$y$ be two real numbers from~$[0,1]$ which satisfy
$1/\base^{\nu+1}< |y-x|\leq 1/\base^{\nu}$ for some integer~$\nu$. The
sequence~$(\boldG_K)$ converges uniformly towards~$\boldF$ and we may
guarantee $\norm{\boldG_{\nu}-\boldF}_{\infty}\leq
\gamma(\lambda/\rho)^{\nu}$ for some constant~$\gamma$ because the
operator~$\cL$ is a contraction with ratio $\lambda/\rho$. Using the
triangular inequality we obtain
\[
\norm{\boldF(y)-\boldF(x)}\leq
2\gamma\left(\frac{\lambda}{\rho}\right)^{\nu}+
\base \norm{V} \left(\frac{\base\lambda}{\rho}\right)^{\nu}|y-x|\leq
\left(2\gamma+\base \norm{V}\right)
\left(\frac{\lambda}{\rho}\right)^{\nu}
,
\]
the last inequality coming from the hypothesis $ |y-x|\leq
1/\base^{\nu}$. But the inequality $1/\base^{\nu}<\base |y-x|$
provides an inequality of the desired form with the help of the
formula $\lambda/\rho=(1/\base)^{\log_{\base}(\rho/\lambda)}$. 
\end{proof}

\begin{proposition}\label{Du08vers09:thm:Holderwithjsr}
  Under Hypotheses~\ref{Du08vers09:hypo:asymptbehav} and~\ref{Du08vers09:hypo:jointspectralradius} the function~$\boldF$ is
  H\"older with exponent~$\log_{\base}(\rho/\lambda)$ for every
  $\lambda>\lambdajsr$. Moreover $\lambda=\lambdajsr$ may be used if\/
  $\lambdajsr$ is attained.
\end{proposition}

\begin{proof}
  If $\lambda>\lambdajsr$ there is a~$T$ such that $\lambda\geq
  \lambda_T$ and we may assume $\lambda=\lambda_T$ for the proof
  because the smaller is~$\lambda$, the stronger is the constraint
  imposed by the H\"older property (the
  exponent~$\log_{\base}(\rho/\lambda)$ is a decreasing function
  of~$\lambda$). We replace the radix~$\base$ by~$\base^T$. This
  changes $\rho$ into $\rho^T$ and $\lambda$ into $\lambda^T$, but
  $\log_{\base^T}(\rho^T/\lambda^T)=\log_{\base}(\rho/\lambda)$
  remains the same and the previous lemma gives the conclusion.
\end{proof}

\begin{example}[Billingsley's distribution function]\label{Du08vers09:ex:Billingsley}
  \citet[Ex.~31.1, p.~407]{Billingsley95} studied the random variable
  $X=\sum_{n\geq0}X_n/2^n$ where~$X_n$ is the result of a coin tossing
  with probabilities~$p_0$ and~$p_1$ for $X_n=0$ and $X_n=1$
  respectively. This defines a rational series with dimension~$1$,
  radix~$2$ and a linear representation \[
  L=\left(\begin{array}{c}1\end{array}\right),\quad
  A_0=\left(\begin{array}{c}p_0\end{array}\right),\quad
  A_1=\left(\begin{array}{c}p_1\end{array}\right),\quad
  C=\left(\begin{array}{c}1\end{array}\right) \] with $0<p_0,p_1<1$,
  $p_0+p_1=1$. We have $Q=\left(\begin{array}{c}1\end{array}\right)$,
  $\rho=1$, $V=\left(\begin{array}{c}1\end{array}\right)$, $R(K)=1$,
  $R'(K)=0$, $\lambdajsr=\max(p_0,p_1)$. The distribution function~$F$
  is the limit function of Theorem~\ref{Du08vers09:thm:existenceF} and
  it is H\"older with exponent $\log_2(1/\lambdajsr)$. We illustrate
  the example (Fig.~\ref{Du08vers09:fig:Billingsley}, left-hand side)
  with $p_0=1/4$, $p_1=3/4$ as in~\cite[p.~408]{Billingsley95} and the
  exponent is~$\log_2(4/3)\simeq 0.415$, and
  (Fig.~\ref{Du08vers09:fig:Billingsley}, right-hand side) with
  $p_0=1/5$, $p_1=4/5$ and the exponent is $\log_2(4/5)\simeq
  0.322$. Following the same way as in \citep[\S~4.3]{DuLiWa07}, it is
  possible to show that, assuming $p_0\leq p_1$, at every dyadic point
  the best H\"older exponent is $\log_2(1/p_0)$ on the right-hand side
  and $\log_2(1/\lambdajsr)$ on the left-hand side. Except in the case
  $p_0=p_1=1/2$, this gives $\log_2(1/p_0)>1$ and this explains the
  horizontal tangents on the right-hand side that we see on the
  pictures. But we will not elaborate upon this point because the
  argument is very simple and assumes a linear representation with
  nonnegative coefficients while \citep{DaLa92} has provided a general
  approach to this subject.

\begin{figure}
\begin{center}
  \includegraphics[width=0.45\linewidth]{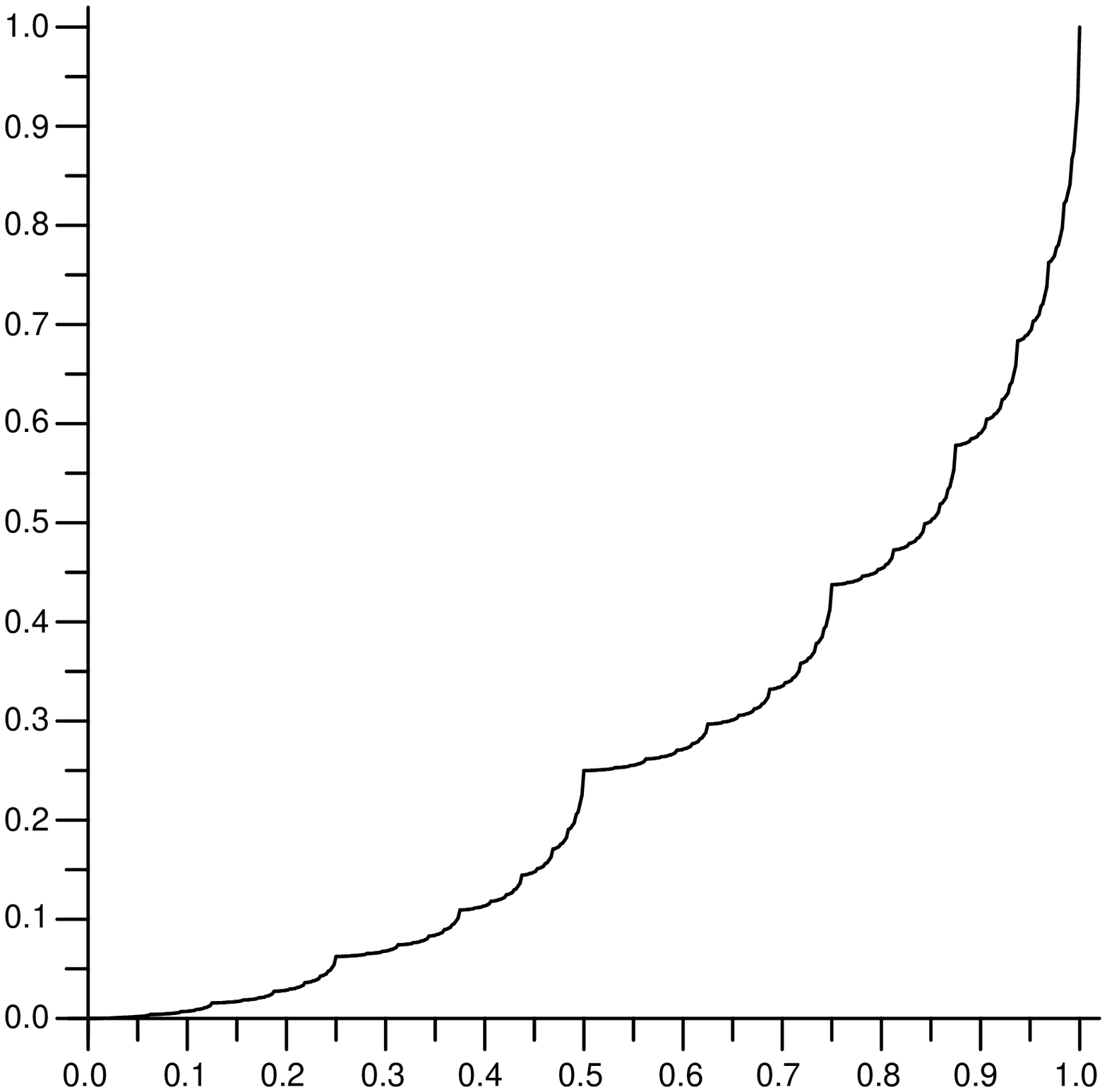}
\hfil
  \includegraphics[width=0.45\linewidth]{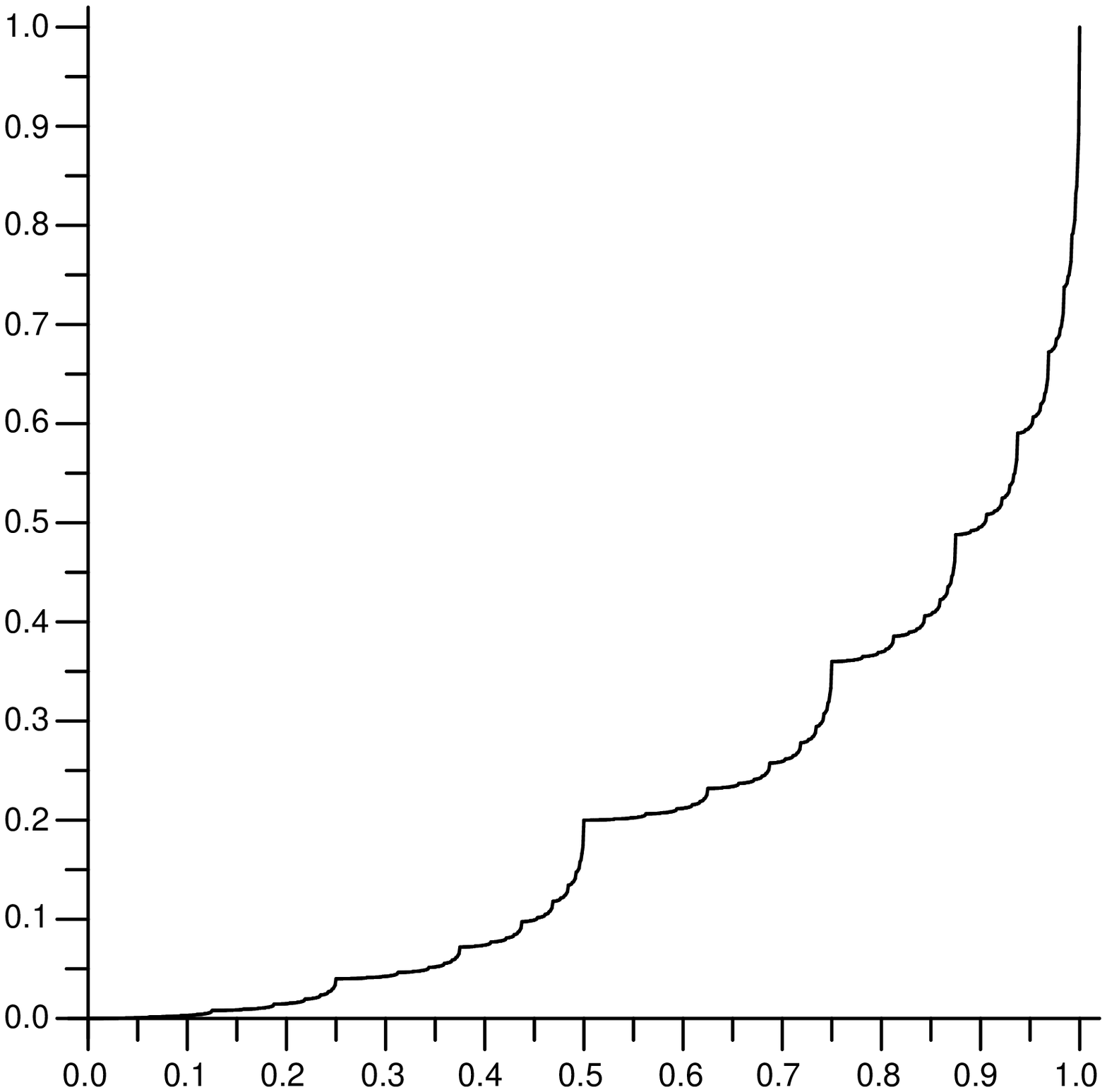}
\end{center}
\caption{\label{Du08vers09:fig:Billingsley}Limit probability
  distributions which come from a Bernoulli process described in
  Ex.~\ref{Du08vers09:ex:Billingsley}.}
\end{figure}
\end{example}

\section{Asymptotic expansion}\label{Du08vers09:sec:fullexpansion}
To a linear representation $L$, $(A_r)_{0\leq r<\base}$, $C$, we
associate a vector-valued sequence of running sums
\[
\boldS_K(x)=\sum_{\begin{subarray}{c}
  \length w =K \\ \Bexp w \leq x		  
		  \end{subarray}}A_wC.
\]
We want to derive an asymptotic expansion for the sequence
$(\boldS_K)$ and as a by-product an asymptotic expansion for the
sequence $(S_K)=(L\boldS_K)$. Theorem~\ref{Du08vers09:thm:existenceF}
and its corollaries provide a one term expansion
\[
\boldS_K(x)
\mathop{=}_{K\to+\infty}
R(K)\boldF(x)+o(R(K)).
\]
In order to obtain a more precise expansion, we generalize the result
of Corollary~\ref{Du08vers09:cor:eigenvectorwitherror}, which deals
with eigenvectors of~$Q$. We find a basis which reduces the
matrix~$Q=\sum_{0\leq r<\base}A_r$ to its Jordan normal form and we
decompose the column vector~$C$ on this basis. The sequence
$(\boldS_K)$ appears as a sum of sequences associated to the vectors
of the basis. We have to discuss according to the eigenvalue of~$Q$
relative to each vector. If the modulus of an eigenvalue is larger
than~$\lambdajsr$ we have to consider a generalized eigenvector or a
Jordan vector and this will be made in the next section. After that it
remains only to consider the case where the modulus is less or equal
to the joint spectral radius~$\lambdajsr$. This case produces a noise
which will enter in the error term of the asymptotic expansion.

\subsection{Jordan vector}\label{Du08vers09:subsec:Jordan}
We use a family of linearly independent vectors $(V^{(j)})_{0\leq
j<\nu}$ which satisfies $QV^{(j)}=\rho\omega V^{(j)}+V^{(j-1)}$ for
$j>0$ and $QV^{(0)}=\rho\omega V^{(0)}$, with $\rho>\lambdajsr$ and
$|\omega|=1$, that is we assume
Hypothesis~\ref{Du08vers09:hypo:generalizedeigenvector} with
Hypothesis~\ref{Du08vers09:hypo:jointspectralradius}. As a
consequence~$Q$ induces on the vector space generated by~$V^{(0)}$,
$V^{(1)}$, $\ldots$, $V^{(\nu-1)}$ the usual Jordan block of
size~$\nu$,
\begin{equation}\label{Du08vers09:eq:Jordanblock}
J_{\rho\omega}=\left(\begin{array}{ccccc}
	 \rho\omega &1 & & & \\
              &\rho\omega&1 && \\
              &    & \ddots&\ddots & \\
              &    &       & & 1\\
              &    & & & \rho\omega
	 \end{array}\right).
\end{equation}
This gives immediately
\begin{multline}\label{Du08vers09:eq:Q^KV}
Q^KV^{(\nu-1)}=\binom K {\nu-1}(\rho\omega)^{K-\nu+1}V^{(0)}+
\binom K {\nu-2}(\rho\omega)^{K-\nu+2}V^{(1)}+
\dotsb\\+
\binom K 1 (\rho\omega)^{K-1}V^{(\nu-2)}+
(\rho\omega)^KV^{(\nu-1)}.
\end{multline}
We claim that the running sum associated to the vector~$V^{(\nu-1)}$
\[
\boldS_K(x)=\sum_{\begin{subarray}{c}\length w = K\\ \Bexp w \leq x\end{subarray}}A_wV^{(\nu-1)}
\]
admits an asymptotic expansion of the form
\begin{multline*}
\boldS_K(x)\mathop{=}_{K\to+\infty}
\binom K {\nu-1}(\rho\omega)^{K-\nu+1}\boldF^{(0)}(x)+
\binom K {\nu-2}(\rho\omega)^{K-\nu+2}\boldF^{(1)}(x)+
\dotsb\\+
\binom K 1 (\rho\omega)^{K-1}\boldF^{(\nu-2)}(x)+
(\rho\omega)^K\boldF^{(\nu-1)}(x)+
\operatorname{error}_K(x),
\end{multline*}
where the~$\boldF^{(j)}$'s are continuous functions from~$[0,1]$
into~$\bC^d$ and the error term is $O(\lambda^K)$ for every
$\lambda>\lambdajsr$. 

The polynomial function $K\mapsto (\rho\omega)^{-K}Q^KV^{(\nu-1)}$ has a
unique writing on the basis $\binom K j$, $0\leq j<\nu$, and we have
necessarily $\boldF^{(j)}(1)=V^{(j)}$ for $0\leq j<\nu$. In the same manner
the uniqueness of asymptotic expansions shows that the family
$(\boldF^{(j)})_{0\leq j<\nu}$ must satisfy the system of dilation
equations,
\begin{equation}\label{systemdilationequation}
\left\{\begin{array}{rcl}
\rho\omega \boldF^{(0)}(x) &=& 
\displaystyle
\sum_{r_1<x_1}A_{r_1}V^{(0)}+A_{x_1}\boldF^{(0)}(\base x-x_1),\\[3.0ex]
\boldF^{(0)}(x)+\rho\omega\boldF^{(1)}(x) &=& 
\displaystyle
\sum_{r_1<x_1}A_{r_1}V^{(1)}+A_{x_1}\boldF^{(1)}(\base x-x_1),\\
&\vdots&\\
\boldF^{(\nu-2)}(x)+\rho\omega\boldF^{(\nu-1)}(x) &=& 
\displaystyle
\sum_{r_1<x_1}A_{r_1}V^{(\nu-1)}+A_{x_1}\boldF^{(\nu-1)}(\base x-x_1).
\end{array}
\right.
\end{equation}
We obtain these formul\ae\ by substituting the asymptotic expansion
into the functional equation of
Lemma~\ref{Du08vers09:lemma:genrecrelS}. Evidently the first equation
of the system is nothing but the basic dilation
equation~\eqref{Du08vers09:eq:dilationequation}. To obtain the
proposition we have in mind, we will follow the same path as in
Sections~\ref{Du08vers09:subsec:unifconvergence}--\ref{Du08vers09:subsec:Holder}. For
the sake of clarity we cut the proof into lemmas.
\begin{lemma}\label{Du08vers09:lemma:systemdilationequation}
  Under Hypotheses~\ref{Du08vers09:hypo:jointspectralradius}
  and~\ref{Du08vers09:hypo:generalizedeigenvector},
  System~\eqref{systemdilationequation} has a unique
  solution~$(\boldF^{(j)})_{0\leq j<\nu}$ in the space $\Pi_{0\leq
    j<\nu}\cC^{(j)}$ where~$\cC^{(j)}$ is the space of continuous
  functions from~$[0,1]$ into~$\bC^d$ which satisfy
  $\boldF^{(j)}(0)=0$, $\boldF^{(j)}(1)=V^{(j)}$ for $0\leq j<\nu$.
\end{lemma}
\begin{proof}
It is possible to consider one equation at a time, but it is more
enlightening to deal globally with the system. Let us introduce the
matrix $\bF(x)$ of type $d\times \left[0,\nu\right)$ whose columns
are the column vectors $\boldF^{(0)}(x)$, $\boldF^{(1)}(x)$, $\ldots$,
$\boldF^{(\nu-1)}(x)$. The system writes
\begin{equation}\label{Du08vers09:eq:matrixsystemdilationequation}
\bF(x)J_{\rho\omega}=\sum_{r<x_1}A_rV+A_{x_1}\bF(\base x-x_1),
\end{equation}
where~$V$ is the matrix whose columns are the column
vectors~$V^{(0)}$, $V^{(1)}$, $\ldots$, $V^{(\nu-1)}$, and $J_{\rho\omega}$
is the Jordan block defined by
Eq.~\eqref{Du08vers09:eq:Jordanblock}. Because of the positivity
of~$\rho$, the matrix~$J_{\rho\omega}$ is invertible and the system turns
out to be a fixed point equation $\bF=\cL\bF$, namely 
\[
\bF(x)=\sum_{r<x_1}A_rVJ_{\rho\omega}^{-1}+A_{x_1}\bF(\base x-x_1)J_{\rho\omega}^{-1}
\]
for some operator~$\cL$.  Let~$\mu$ lie between~$\lambdajsr$
and~$\rho$, that is $\lambdajsr<\mu<\rho$. Because the spectral radius
of~$ J_{\rho\omega}^{-1}$ is $1/\rho$, there exists an induced norm
which provides $1/\rho\leq \| J_{\rho\omega}^{-1}\| <1/\mu$. Next
there exists an integer~$T$ which gives $\lambdajsr \leq
\lambda_T<\mu$. We consider for a while that the radix
is~$\base^T$. From $
\norm{A_{x_1}} 
\| J_{\rho\omega}^{-1} \|
< \mu\times 1/\mu=1$, we see that the right-hand side of the
previous equation turns out to be the image of~$\bF$ by a
contracting operator. We conclude as in the proof of
Proposition~\ref{Du08vers09:prop:basicprop} that the equation
$\bF=\cL^T\bF$ has a unique continuous solution~$\bF^{(0)}$
from~$[0,1]$ into the space of matrices $\operatorname{M}_{d\times
  \left[0,\nu\right)}({\mathbb C})$ which satisfies the constraints
$\bF^{(0)}(0)=0$, $\bF^{(0)}(1)=V$. Next we prove
that~$\bF^{(0)}$ is the unique solution of the problem, in the same
way we have finished the proof of
Theorem~\ref{Du08vers09:thm:existenceF}.
\end{proof}

\begin{lemma}
Let $\boldA_K$ be the function from~$[0,1]$ into~$\bC^d$ defined by
\begin{multline*}
\boldA_K(x)=
\binom K {\nu-1}(\rho\omega)^{K-\nu+1}\boldF^{(0)}(x)+
\binom K {\nu-2}(\rho\omega)^{K-\nu+2}\boldF^{(1)}(x)+
\dotsb\\+
\binom K 1 (\rho\omega)^{K-1}\boldF^{(\nu-2)}(x)+
(\rho\omega)^K\boldF^{(\nu-1)}(x)
\end{multline*}
for each nonnegative integer~$K$, where $\boldF^{(0)}$, $\ldots$, $
\boldF^{(\nu-1)}$ are the functions whose existence and uniqueness are
asserted by Lemma~\ref{Du08vers09:lemma:systemdilationequation}. Under
Hypotheses~\ref{Du08vers09:hypo:jointspectralradius}
and~\ref{Du08vers09:hypo:generalizedeigenvector}, the running sum
$\boldS_K$ associated to~$V^{(\nu-1)}$ writes
$\boldS_K(x)=\boldA_K(x)+O(\lambda^K)$ for every $\lambda>\lambdajsr$
and the big oh is uniform with respect to~$x$ in~$[0,1]$. Moreover if
$\lambdajsr$ is attained we may replace~$\lambda$ by~$\lambdajsr$.
\end{lemma}
\begin{proof}
With $\operatorname{error}_K(x)=\boldS_K(x)-\boldA_K(x)$ it is not
difficult to obtain, by substitution, the recursion 
\[
\operatorname{error}_{K+1}(x)=A_{x_1}\operatorname{error}_K(\base x-x_1).
\]
More precisely, we use the recursion of
Lemma~\ref{Du08vers09:lemma:genrecrelS},
System~\eqref{systemdilationequation}, and the fundamental formula of
binomial coefficients.  For $\lambda>\lambdajsr$ we choose an
integer~$T$ such that $\lambdajsr\leq \lambda_T<\lambda$. We readily
obtain
$\norm{\operatorname{error}_{KT}}_{\infty}=O(\lambda_T^{KT})$. Next
the recursion shows that we have
$\norm{\operatorname{error}_{KT+s}}_{\infty}=O(\lambda_T^{KT+s})$ for
$0\leq s<T$, hence the result. If $\lambdajsr$ is some~$\lambda_T$, it
is useless to consider a $\lambda>\lambdajsr$.
\end{proof}

As in Lemma~\ref{Du08vers09:thm:Holderbasic}, we show that
the functions~$\boldF^{(j)}$'s are somewhat regular.
\begin{lemma}\label{Du08vers09:lemma:HolderJordan}
Under Hypotheses~\ref{Du08vers09:hypo:jointspectralradius} and~\ref{Du08vers09:hypo:generalizedeigenvector}, each of the functions
$\boldF^{(j)}$, $0\leq j<\nu$, defined in
Lemma~\ref{Du08vers09:lemma:systemdilationequation} is H\"older with
exponent $\log_{\base}(\rho/\lambda)$ for every
$\lambda>\lambdajsr$. If $\lambdajsr$ is attained, they are H\"older
with exponent $\log_{\base}(\rho/\lambdajsr)$.
\end{lemma}
\begin{proof}
  We define a sequence of parameterized curves in the space of
  matrices of type $d\times\left[0,\nu\right)$. The first one is
  $\varGamma_0=0$ and the others are defined by the recursion
  $\varGamma_{K+1}=\cL\varGamma_K$, where~$\cL$ is the operator we used
  in the proof of Lemma~\ref{Du08vers09:lemma:systemdilationequation},
  namely
\[
\varGamma_{K+1}(x)=\sum_{r<x_1}A_rVJ_{\rho\omega}^{-1}+A_{x_1}\varGamma_K(\base x-x_1)J_{\rho\omega}^{-1}.
\]
Each $\varGamma_K$ is piecewise linear. 

We choose an induced norm and a~$T$ such that $\lambdajsr\leq
\lambda_T<\rho$, and we consider that we use the radix~$\base^T$. As
in the proof of Proposition~\ref{Du08vers09:thm:Holderbasic} we
observe that the sequence~$(\varGamma_{KT})$ converges uniformly
towards the solution~$\bF=(\boldF^{(j)})_{0\leq j<\nu}$ described
in Lemma~\ref{Du08vers09:lemma:systemdilationequation}. This permits
us to show that~$\bF$ is H\"older with exponent
$\log_{\base}(\rho/\lambda_T)$. The argument remains the same: the
maximal speed of a piecewise linear function is the greatest uniform
speed on each piece and to compute a uniform speed it is sufficient to
test two distinct points of the interval under consideration.
\end{proof}

We summarize the obtained result.
\begin{proposition}\label{Du08vers09:thm:Jordanasymptexp}
  Under Hypotheses~\ref{Du08vers09:hypo:jointspectralradius}
  and~\ref{Du08vers09:hypo:generalizedeigenvector},
  System~\eqref{systemdilationequation} has a unique
  solution~$(\boldF^{(j)})_{0\leq j<\nu}$ in the space of continuous
  functions from~$[0,1]$ into~$\bC^{d\times \left[0,\nu\right)}$ such
  that $\boldF^{(j)}(0)=0$, $\boldF^{(j)}(1)=V^{(j)}$ for $0\leq
  j<\nu$. Each of the functions $\boldF^{(j)}$, $0\leq j<\nu$, is
  H\"older with exponent $\log_{\base}(\rho/\lambda)$ for every
  $\lambda>\lambdajsr$. The running sum $\boldS_K$ associated
  to~$V^{(\nu-1)}$ admits the asymptotic expansion
\begin{multline*}
\boldS_K(x)=
\binom K {\nu-1}(\rho\omega)^{K-\nu+1}\boldF^{(0)}(x)+
\binom K {\nu-2}(\rho\omega)^{K-\nu+2}\boldF^{(1)}(x)+
\dotsb\\+
\binom K 1 (\rho\omega)^{K-1}\boldF^{(\nu-2)}(x)+
(\rho\omega)^K\boldF^{(\nu-1)}(x) +O(\lambda^K)
\end{multline*}
for every $\lambda>\lambdajsr$ and the big oh is uniform with respect
to~$x$ in~$[0,1]$. Moreover if $\lambdajsr$ is attained, the  
functions~$\boldF^{(j)}$'s, $0\leq j<\nu$, are H\"older with exponent
$\log_{\base}(\rho/\lambdajsr)$, and we may replace~$\lambda$
by~$\lambdajsr$ in the asymptotic expansion.
\end{proposition}
Despite of its technical aspect, the obtained result is very simple:
each term of the expansion of $Q^KV^{(\nu-1)}$ is shaken by a somewhat
regular vector-valued function which goes from~$0$ to the vector which
appears at the same place in the expansion.

\subsection{Noise}
So far we did not take into account the case $\rho\leq
\lambdajsr$. Some experiments lead us to expect a chaotic behaviour
for the running sums and we content ourselves with exhibiting a
natural bound.

First we consider an eigenvector~$V$ of~$Q$ associated to an
eigenvalue~$\rho\omega$, with $\rho\geq 0$ and $|\omega|=1$. Note that
this time the modulus~$\rho$ may be~$0$. 
\begin{lemma}\label{Du08vers09:lemma:noiseeigenvector}
Let~$V$ be an eigenvector for~$Q$ with eigenvalue of modulus
$\rho\geq 0$ which satisfies $\rho\leq \lambdajsr$. The sequence
of running sums associated to~$V$ satisfies
$\norm{\boldS_K}_{\infty}=O(\lambda^K)$ for every
$\lambda>\lambdajsr$. If $\lambdajsr$ is attained, we may replace~$O(\lambda^K)$
 by~$O(\lambdajsr^K)$ in case $\rho<\lambdajsr$ and
 by~$O(K\lambdajsr^K)$ in case $\rho=\lambdajsr$.  
\end{lemma}
\begin{proof}
  The functional equation satisfied by the running sum associated to
  the eigenvector~$V$ is
\[
\boldS_{K+1}(x)=\sum_{r_1<x_1}A_{r_1}Q^KV+A_{x_1}\boldS_K(\base x-x_1),
\]
that is 
\[
\boldS_{K+1}(x)=(\rho\omega)^K\sum_{r_1<x_1}A_{r_1}V
+A_{x_1}\boldS_K(\base x-x_1).
\]
Let us assume first that we have a $\lambda$ such that $\rho\leq
\lambda$ and $\norm{A_r}\leq \lambda$ for $0\leq r<\base$ and for some
induced norm, that is let us assume in a certain sense the converse of
Hypothesis~\ref{Du08vers09:hypo:roughspectralradius}. We obtain
\[
\norm{\boldS_{K+1}(x)}\leq \rho^K\sum_{r_1<x_1}\lambda\norm{V}+\lambda\norm{\boldS_K(\base x-x_1)},
\]
hence
\[
\norm{\boldS_{K+1}}_{\infty}\leq \rho^K\base \lambda\norm{V}+\lambda\norm{\boldS_K}_{\infty}.
\]
With $\norm{\boldS_K}_{\infty}=\lambda^K u_K$, we have $u_{K+1}\leq
(\rho/\lambda)^K\base \norm{V}+u_K$. In case $\rho<\lambda$, we have
$u_K\leq u_0+\base\norm{V}/(1-\rho/\lambda)$ and we conclude
$\norm{\boldS_K}_{\infty}=O(\lambda^K)$. In case $\rho=\lambda$, we
have $u_K\leq u_0+K\base \norm{V}$ and we conclude
$\norm{\boldS_K}_{\infty}=O(K\lambda^K)$.

Second we consider the general case, which divides into two sub-cases.
Let us assume that the joint spectral radius is smaller that all
$\lambda_T$, for every induced norm. For any $\lambda>\lambdajsr$, we
may find, for a given induced norm, a~$T$ such that
$\lambdajsr<\lambda_T\leq \lambda$. We apply the previous argument to
the sequence $(\boldS_{KT})$. This gives
$\norm{\boldS_{KT}}_{\infty}=O(\lambda_T^{KT})$ because we have
$\rho\leq \lambdajsr<\lambda_T$. Next,
Lemma~\ref{Du08vers09:lemma:genrecrelS} gives
$\norm{\boldS_{KT+s}}_{\infty}=O(\lambda_T^{KT+s})$ for $0\leq s<T$,
because there is a finite number of~$s$ between~$0$ and~$T$. In this
way, we obtain $\norm{\boldS_K}_{\infty}=O(\lambda_T^K)$, and we
conclude $\norm{\boldS_K}_{\infty}=O(\lambda^K)$ for every
$\lambda>\lambdajsr$.

Let us assume that~$\lambdajsr$ is some~$\lambda_T$ for some induced
norm. In case $0\leq\rho<\lambdajsr=\lambda_T$ the same reasoning as
above, with~$\lambdajsr$ in place of~$\lambda$, gives
$\norm{\boldS_K}_{\infty}=O(\lambdajsr^K)$. In case $\rho=\lambdajsr$
we find $\norm{\boldS_{KT}}_{\infty}=O(K\lambdajsr^{KT})$. As above
and again because there is a finite number of integers between 0
and~$T$, we obtain $\norm{\boldS_K}_{\infty}=O(K\lambdajsr^K)$.
\end{proof}
The next example shows that the upper bound we have found in case
$\rho=\lambdajsr$ has the right order of growth.
\begin{example}[Triangular tiling]\label{Du08vers09:ex:geometricex}
For a real~$\vartheta$, we consider the rotation matrix
\[
R_{\vartheta}=\left(\begin{array}{cc}\cos{\vartheta} & -\sin{\vartheta}\\\sin{\vartheta}&\cos{\vartheta}\end{array}\right).
\]
In the following we use $V_0=E_1$ the first vector of the canonical
basis and $V_{\vartheta}=R_{\vartheta}V_0$. The example is based on the
linear representation
\[
L=\left(\begin{array}{cc}
1&0
	\end{array}\right),\qquad
A_0=R_{-\vartheta},\qquad
A_1=R_{\vartheta},\qquad
C=V_0.
\]
Because we use orthogonal matrices, the joint spectral radius is
$\lambdajsr=1$. We specialize the case to $\vartheta=\pi/3$, so that
we have $\rho=\lambdajsr$. 

We consider the words $u=11$, $v=01$, and the rational
numbers~$x_{2K+2}$ and $y_{2K}$ whose binary expansions are
$(0.uv^K)_2$ and $(0.v^K)_2$ respectively. According to the functional
equation satisfied by~$(\boldS_K)$, we have
\begin{multline*}
\boldS_{2K+2}(x_{2K+2})=R_{-{\pi/3}}V_0+R_{{\pi/3}}\left(R_{-{\pi/3}}V_0+R_{{\pi/3}}S_{2K}(y_{2K})\right)\\=
V_{-{\pi/3}}+V_0+R_{2{\pi/3}}S_{2K}(y_{2K}),
\end{multline*}
\[
\boldS_{2K}(y_{2K})=R_{-{\pi/3}}\left(R_{-{\pi/3}}V_0+R_{{\pi/3}}\boldS_{2K-2}(y_{2K-2})\right)=
V_{-2{\pi/3}}+\boldS_{2K-2}(y_{2K-2}).
\]
With $\boldS_0(\varepsilon)=V_0$ ($\varepsilon$ is the empty word), we obtain by induction  $\boldS_{2K}(y_{2K})=KV_{-2{\pi/3}}+V_0$ and
\[
\boldS_{2K+2}(x_{2K+2})=(K+1)V_0+V_{-{\pi/3}}+V_{2{\pi/3}}=(K+1)V_0.
\]
Hence the~$O(K\lambdajsr^K)$ we have found in
Lemma~\ref{Du08vers09:lemma:noiseeigenvector} is satisfying.

The orbit of the parameterized curve~$\boldS_K$ is illustrated in
Fig.~\ref{Du08vers09:fig:triangular_cells}. In this example the
recurrences which define the sequence~$(\boldS_K)$ and the
sequence~$(\boldF_K)$ are exactly the same. Moreover computing the
values $\boldS_K(x)$ for a rational~$x$ whose binary expansion has
length at most~$K$ is almost the same as computing the values
$\boldF(x)$ for a hypothetical solution of the basic dilation
equation. This is not really true because $A_0V_0$ is not~$V_0$, and
the less significant bits with value~$0$ must be taken into
account. (Look at Fig.~\ref{Du08vers09:fig:triangular_cells}. From the
first to the second picture a half-turn has been applied to the little
pattern of the leftmost picture, because $A_0^3=-\operatorname{I}_2$.)
But $A_0^6=\operatorname{I}_2$ and if we use radix $2^6=64$ in place
of~$2$ this problem disappears. The computation we have made shows
that the basic dilation equation has no solution: if~$\boldF$ was a
solution, it would satisfy $\boldF(y)=V_{-2{\pi/3}}+\boldF(y)$ for
$y=1/3=(0.010101\ldots)_2$.

\begin{figure}
\begin{center}
\includegraphics[width=0.3\linewidth]{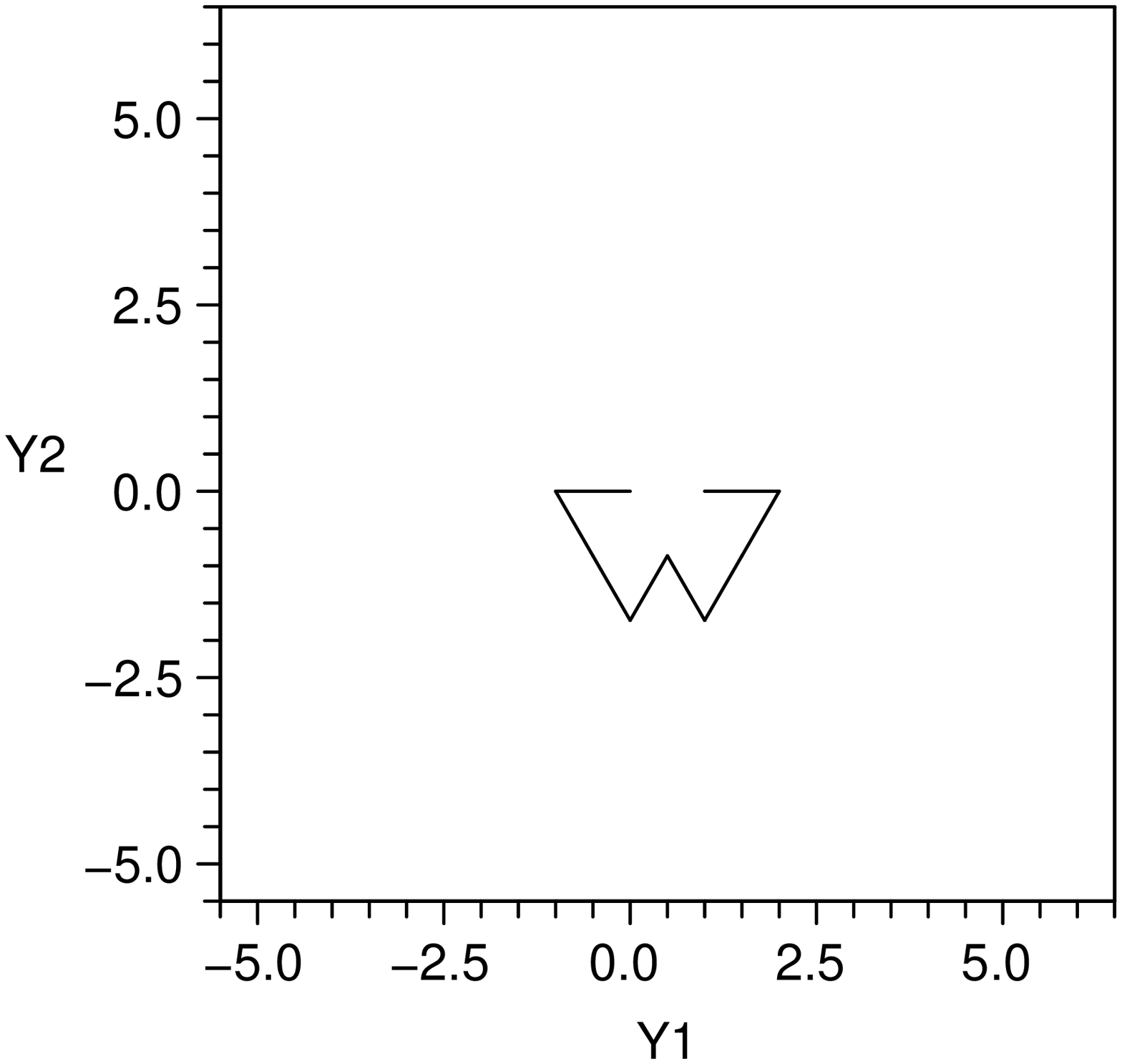}
\hfil
\includegraphics[width=0.3\linewidth]{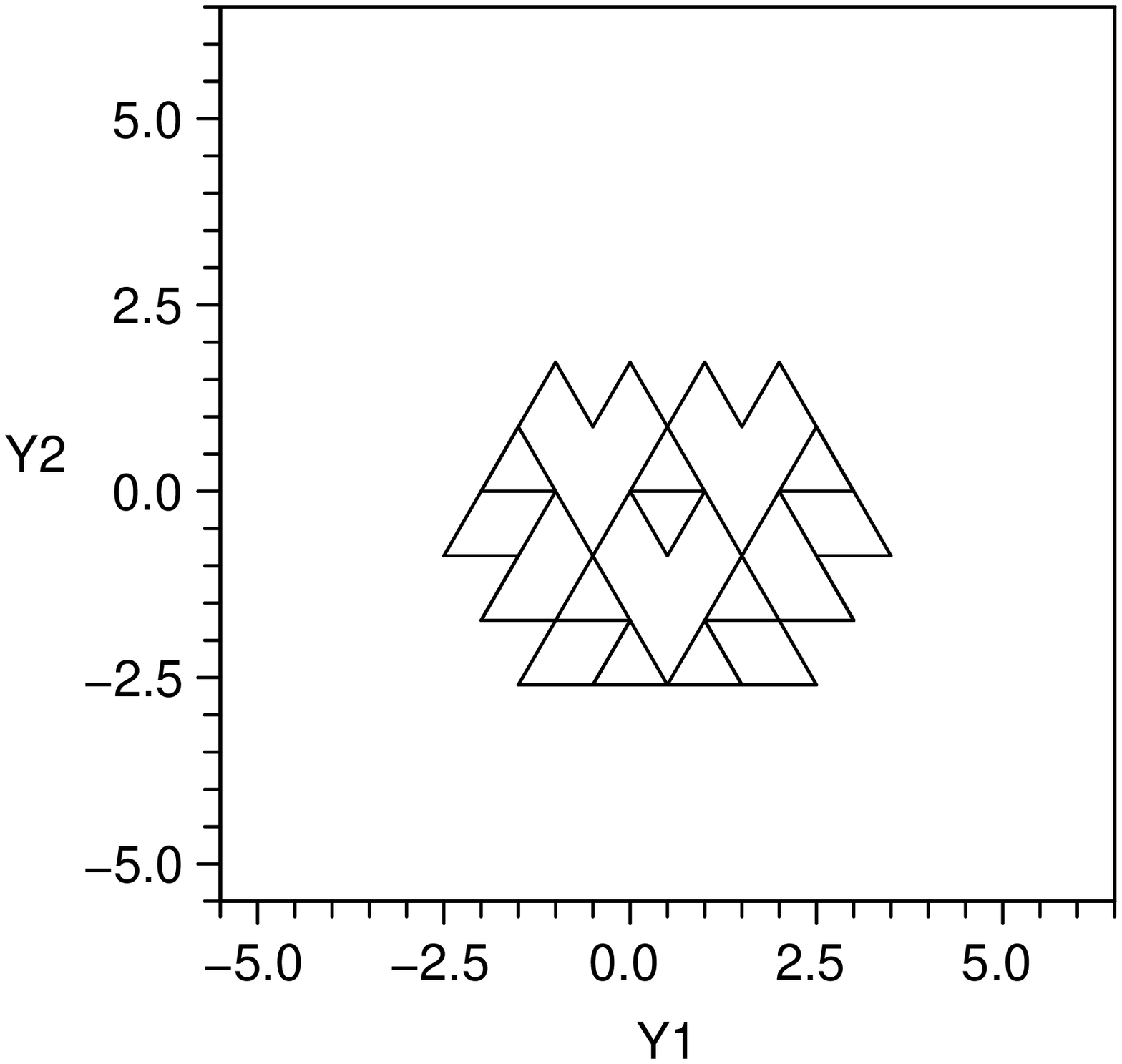}
\hfil
\includegraphics[width=0.3\linewidth]{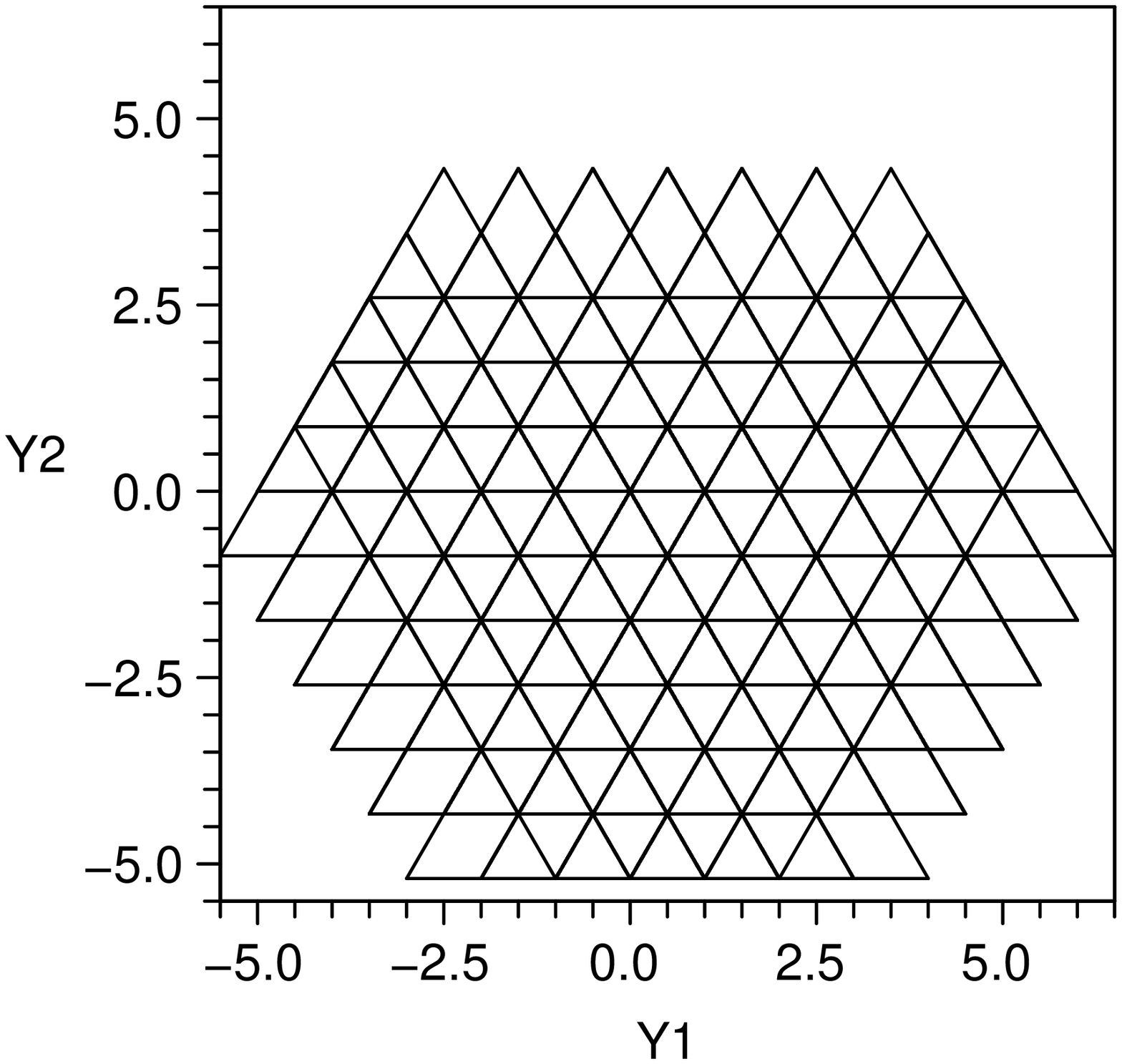}
\end{center}
\caption{\label{Du08vers09:fig:triangular_cells}The image of the parameterized curve~$\boldS_K$ from Ex.~\ref{Du08vers09:ex:geometricex} grows linearly with~$K$. Here we see $\boldS_3$, $\boldS_6$, and $\boldS_{12}$.
}
\end{figure}
\end{example}

We generalize slightly the previous lemma into the following assertion.
\begin{proposition}\label{Du08vers09:prop:noiseforgeneigenvector}
  Under Hypothesis~\ref{Du08vers09:hypo:generalizedeigenvector} with
  $0\leq \rho\leq \lambdajsr$, the sequence of running
  sums~$(\boldS_K)$ associated to~$V^{(\nu-1)}$ satisfies
  $\norm{\boldS_K}_{\infty}=O(\lambda^K)$ for every
  $\lambda>\lambdajsr$. Moreover if~$\lambdajsr$ is attained, we may
  replace~$O(\lambda^K)$ by~$O(\lambdajsr^K)$ in case
  $\rho<\lambdajsr$ and by $O(\lambdajsr^KK^{\nu})$ otherwise.
\end{proposition}
\begin{proof}
The generalized eigenvector $V^{(\nu-1)}$ satisfies
$Q^KV^{(\nu-1)}=O(\rho^KK^{\nu-1})$. As above we conclude
$\norm{\boldS_K}_{\infty}=O(\lambda^K)$ for every
$\lambda>\lambdajsr$. We may refine the computation in case
$\lambdajsr$ is some~$\lambda_T$~: in that case, if $\rho<\lambdajsr$
we have $\norm{\boldS_K}_{\infty}=O(\lambdajsr^K)$, and if
$\rho=\lambdajsr$ we obtain
$\norm{\boldS_K}_{\infty}=O(\lambdajsr^KK^{\nu})$.
\end{proof}

\begin{example}[Powers of~$2$ characteristic sequence]\label{Du08vers09:ex:poersof2}
The sequence~$(u(n))$ which takes value~$1$ on the powers of~$2$
and~$0$ on the others integers admits the generating family
$(u(n),u(2n+1))$ and the linear representation
\[
L=\left(\begin{array}{cc}
0 & 1 
	\end{array}\right),\qquad
A_0=\left(\begin{array}{cc}
1 & 0 \\ 0 & 1 
	\end{array}\right),\qquad
A_1=\left(\begin{array}{cc}
0 & 0 \\ 1 & 0
	\end{array}\right),\qquad
C=\left(\begin{array}{c}
1\\0
	\end{array}\right).
\]
Matrix $Q=A_0+A_1$ has~$1$ as a unique eigenvalue and the
plane~$\bR^2$ is the characteristic subspace
$\operatorname{Ker}(Q-\operatorname{I}_2)^2$.  
Hypothesis~\ref{Du08vers09:hypo:asymptbehav} is satisfied with
\[
\rho=1,\qquad R(K)=K,\qquad
V=\left(\begin{array}{cc}0&1\end{array}\right)^{\trpse}. 
\]
The joint spectral radius is $\lambdajsr=1$
(consider~$A_0$). The fixed point equation is
\[
\left\{\begin{array}{l}
F_1(x)=F_1(2x),\\
F_2(x)=F_2(2x),
       \end{array}\right.
\quad\text{for $0\leq x<\frac{1}{2}$\,;}
\qquad
\left\{\begin{array}{l}
F_1(x)=0,\\
F_2(x)=1+F_1(2x-1),
       \end{array}\right.
\quad\text{for $\frac{1}{2}\leq x<1$.}
\]
It admits the only solution (consider~$F_1$ and first the
interval~$\left[1/2,1\right[$, next the
interval~$\left[1/4,1/2\right[$, and so on)
\[
\boldF(x)=\left\{\begin{array}{ll}
(\begin{array}{cc}0&1\end{array})^{\trpse} & \text{if $0<x\leq 1$,}\\
(\begin{array}{cc}0&0\end{array})^{\trpse} & \text{if $x=0$.}
		 \end{array}\right.
\]
(We do not forget the conditions $\boldF(0)=0$, $\boldF(1)=V$.) The
sequence $(\boldF_K)$ satisfies
\[
\boldF_K(x)=\left\{\begin{array}{ll}
(\begin{array}{cc}0&0\end{array})^{\trpse} & \text{if $0\leq x<1/2^K$,}\\
(\begin{array}{cc}1/K&k/K\end{array})^{\trpse} & \text{if $2^{k-1}/2^K\leq x <2^k/2^K$ for $1\leq k\leq K$,}\\
(\begin{array}{cc}1/K&1\end{array})^{\trpse} & \text{if $x= 1.$}
		   \end{array}\right.
\]
It turns out that $(\boldF_K)$ converges simply towards~$\boldF$. This
implies that we have $\boldS_K(x)=K\boldF(x)+O(1)$ and
$\norm{\boldS_K}_{\infty}\sim K$.
\end{example}
Previous Examples~\ref{Du08vers09:ex:geometricex}
and~\ref{Du08vers09:ex:poersof2} show that under the condition
$\rho=\lambdajsr$ the sequence of running sums~$(\boldS_K)$ can or
cannot define a limit function~$\boldF$ according to the case.

\subsection{General case}
The general case is obtained by expanding the vector~$C$ over a Jordan
basis for~$Q$. Let us recall that the height of a generalized
eigenvector is the dimension of the subspace it spans under the action
of~$Q$. Gathering Prop.~\ref{Du08vers09:thm:Jordanasymptexp}
and~\ref{Du08vers09:prop:noiseforgeneigenvector}, we arrive at the
following assertion.
\begin{theorem}\label{Du08vers09:thm:formalasymptoticexpansion}
  Let $L$, $(A_r)_{0\leq r<\base}$, $C$ be a linear representation of
  a rational formal power series. The sequence of running sums 
  \[
\boldS_K(x)=\sum_{\begin{subarray}{c}\length w = K\\ \Bexp w \leq x\end{subarray}} A_wC
  \]
  admits an asymptotic expansion with error term~$O(\lambda^K)$ for
  every $\lambda>\lambdajsr$, where~$\lambdajsr$ is the joint spectral
  radius of the family~$(A_r)_{0\leq r<\base}$. The used asymptotic
  scale is the family of sequences $\rho^K\binom K \ell$, $\rho>0$,
  $\ell\in\mathbb N_{\geq 0}$. The error term is uniform with respect
  to~$x\in[0,1]$. 

  If~$C$ is a generalized eigenvector~$V$ for~$Q$ of
  height~$\nu$ associated to an eigenvalue~$\rho\omega$ of
  modulus~$\rho>\lambdajsr$ and if~$(V^{(j)})_{0\leq j<\nu}$ is the
  associated family of vectors such that $V^{(\nu-1)}=V$,
  $QV^{(j)}=\rho\omega V^{(j)}+V^{(j-1)}$ for $j>0$ and
  $QV^{(0)}=\rho\omega V^{(0)}$ , from the expansion 
\begin{multline*}
  Q^KV=\binom K {\nu-1}(\rho\omega)^{K-\nu+1}V^{(0)}+
\binom K {\nu-2}(\rho\omega)^{K-\nu+2}V^{(1)}+
\dotsb\\+
\binom K 1 (\rho\omega)^{K-1}V^{(\nu-2)}+
(\rho\omega)^KV^{(\nu-1)}.
\end{multline*}
  is deduced the asymptotic expansion 
\begin{multline*}
\sum_{\begin{subarray}{c}\length w = K\\ \Bexp w \leq x\end{subarray}}A_wV
\mathop{=}_{K\to+\infty}
\binom K {\nu-1}(\rho\omega)^{K-\nu+1}\boldF^{(0)}(x)+
\binom K {\nu-2}(\rho\omega)^{K-\nu+2}\boldF^{(1)}(x)\\+
\dotsb+
\binom K 1 (\rho\omega)^{K-1}\boldF^{(\nu-2)}(x)+
(\rho\omega)^K\boldF^{(\nu-1)}(x)+
O(\lambda^K)
\end{multline*}
where the family $(\boldF^{(j)})_{0\leq j<\nu}$ is the unique solution
to the system of dilation equations
\[
\left\{\begin{array}{rcl}
\rho\omega \boldF^{(0)}(x) &=& 
\displaystyle
\sum_{r_1<x_1}A_{r_1}V^{(0)}+A_{x_1}\boldF^{(0)}(\base x-x_1),\\[3.0ex]
\boldF^{(0)}(x)+\rho\omega\boldF^{(1)}(x) &=& 
\displaystyle
\sum_{r_1<x_1}A_{r_1}V^{(1)}+A_{x_1}\boldF^{(1)}(\base x-x_1),\\
&\vdots&\\
\boldF^{(\nu-2)}(x)+\rho\omega\boldF^{(\nu-1)}(x) &=& 
\displaystyle
\sum_{r_1<x_1}A_{r_1}V^{(\nu-1)}+A_{x_1}\boldF^{(\nu-1)}(\base x-x_1).
\end{array}
\right.
\]
Moreover if the joint spectral radius~$\lambdajsr$ is attained, the
error term may be improved in the following way. Let~$\cV$ be a
Jordan basis for ~$Q$. Let~$m$ be the maximal height of the
generalized eigenvectors~$V$ from~$\cV$ associated to an eigenvalue
of modulus~$\lambdajsr$ such that the vector~$C$ has a nonzero
component over~$V$, with $m=0$ if there does not exist such
generalized eigenvector. The error term may be taken as
$O(\lambdajsr^KK^m)$.
\end{theorem}

The asymptotic expansion is concretely obtained by the algorithm of
page~\pageref{Du08vers09:algoforrationalformalseries} named {\tt
LRtoAE1} (for linear representation to asymptotic expansion). Clearly
this gives an asymptotic expansion for the number-valued sum
$S_K(x)=L\boldS_K(x)$ too. It is noteworthy that two important
hypotheses are made in the algorithm. First it is assumed that the
joint spectral radius~$\lambdajsr$ is known and even more it is
assumed that we know if the joint spectral radius is attained. This is
not an obstacle to the use of the algorithm. If we have at our
disposal an upper bound $\lambda_+>\lambdajsr$, we may
replace~$\lambdajsr$ by~$\lambda_+$, obtaining of course a less
precise expansion. Second it is assumed that we can solve dilation
equations according to
Lemma~\ref{Du08vers09:lemma:systemdilationequation}. This is evidently
wrong in full generality, but the solution is mathematically well
defined and may be computed exactly over a dense subset by the cascade
algorithm.


\begin{algorithm}
\caption{\label{Du08vers09:algoforrationalformalseries} {\tt LRtoAE1}}
\SetKwInOut{KwInput}{Input}
\SetKwInOut{KwOutput}{Output}
\SetLine
\SetArgSty{rm}
  \KwInput{A linear representation $L$, $(A_r)_{0\leq r<\base}$, $C$,
  and its joint spectral radius~$\lambdajsr$.}
  \KwOutput{An asymptotic expansion of the running sum $\boldS_K(x)$
  associated to the vector~$C$ when~$K$ goes to infinity at the
  precision $\Ostar(\lambdajsr^K)$ if~$\lambdajsr$ is attained and
  $O(\lambda^K)$ for a $\lambda$ slightly greater than $\lambdajsr$ if
  not, with respect to the scale $\rho^K\binom K \ell$, $\rho>0$,
  $\ell\in\mathbb N_{\geq 0}$.}
  $Q:=\sum_{0\leq r<\base}A_r$\;
  \eIf{$\lambdajsr$ is attained}
      {$\lambda:=\lambdajsr$}
      {$\lambda:=$ any number between $\lambdajsr$ and the infimum of
      the modulus of eigenvalues of~$Q$ greater than~$\lambdajsr$}
  compute a Jordan basis $\cV$ for the matrix $Q$\;
  expand the column vector~$C$ of the linear representation over the
  Jordan basis, as $C=\sum_{V\in\cV}\gamma_VV$\;
  $\cV_{>}:=$ the set of generalized eigenvectors in~$\cV$ such that $\gamma_V\neq0$ and the associated eigenvalue~$\rho\omega$ has a modulus $\rho>\lambda$\;
  \For{each vector~$V$ in~$\cV_{>}$} 
      {$\rho\omega:=$ the eigenvalue associated to~$V$\;
       $\nu:=$ the height of~$V$\;
       an asymptotic expansion 
\begin{multline*}
\boldA_{V;K}(x):=
\binom K {\nu-1}(\rho\omega)^{K-\nu+1}\boldF^{(0)}(x)+
\binom K {\nu-2}(\rho\omega)^{K-\nu+2}\boldF^{(1)}(x)+\mbox{}\\
\dotsb+
\binom K 1 (\rho\omega)^{K-1}\boldF^{(\nu-2)}(x)+
(\rho\omega)^K\boldF^{(\nu-1)}(x)\qquad
\end{multline*}
       is deduced from Proposition~\ref{Du08vers09:thm:Jordanasymptexp}}
  $\boldA_K(x):=\sum_{V\in\cV_{>}}\gamma_{V}\boldA_{V;K}(x)$\;
  $\operatorname{error~term}_K(x):=O(\lambda^K)$\;
  \If{$\lambdajsr$ is attained}
     {$\cV_=:=$ the set of generalized eigenvectors in~$\cV$ such that
     $\gamma_V\neq0$ and the associated eigenvalue~$\rho\omega$ has a
     modulus $\rho=\lambdajsr$\;
     $m:=0$\;
    \For{each~$V$ in~$\cV_{=}$}
        {$\nu:=$ the height of~$V$\;
         $m:=\max(m,\nu)$}
     $\operatorname{error~term}_K(x):=O(\lambdajsr^KK^m)$}
  \Return{
     \[\displaystyle \boldS_K(x)\mathop{=}_{K\to+\infty}\boldA_K(x)+\operatorname{error~term}_K(x)\]
  }
\end{algorithm}

\begin{example}[Lipmaa--Wallen's formal power series]\label{Du08vers09:ex:LipmaaWallen}
For a nonnegative integer~$K$, the ring ${\mathbb Z}/2^K{\mathbb Z}$
may be equipped with two additions.  The first one is the quotient
addition~$+$ which comes form the addition of the ring of
integers~$\mathbb Z$. The second one uses the identification of the
set ${\mathbb Z}/2^K{\mathbb Z}$ with the set $({\mathbb Z}/2{\mathbb
Z})^K$ through the correspondence which maps a $K$-tuple of bits
$(\varepsilon_k)_{0\leq k<K}$ onto the integer
$\varepsilon_0+2\varepsilon_1+\dotsb+2^{K-1}\varepsilon_{K-1}$. It is
the bitwise addition~$\oplus$, or exclusive-or. The cryptanalysts
Lipmaa and Wallen have studied the propagation of differences (in the
sense of~$+$) through~$\oplus$ and more precisely the additive
differential probability $\operatorname{adp}$ which maps a triple
$(\alpha,\beta,\gamma)\in ({\mathbb Z}/2^K{\mathbb Z})^3$ onto the
probability that the difference $((x+\alpha)\oplus (y+\beta))-x\oplus
y$ has value~$\gamma$ for $x$, $y$ taken from ${\mathbb Z}/2^K{\mathbb
Z}$ with equiprobability. By coding every triple
$(\alpha,\beta,\gamma)$ as a word~$w$ over the alphabet ${\mathbb
Z}/8{\mathbb Z}$ (with $w_k=4\alpha_k+2\beta_k+\gamma_k$ for $0\leq
k<K$), \citet{LiWaDu04} show that $\operatorname{adp}$ turns out to be
a rational series. With $(E_j)_{1\leq j\leq 8}$ the canonical basis of
$\mathbb Q^8$, it admits the $8$-dimensional linear representation
\[
L=\sum_{0\leq r<8}E_i^{\trpse},\qquad (A_r)_{0\leq r<8},\qquad C=E_1
\]
where~$A_0$ is defined as
  \[
  A_0 = \frac14\left(
    \begin{array}{cccccccc}
      4 & 0 & 0 & 1 & 0 & 1 & 1 & 0 \\
      0 & 0 & 0 & 1 & 0 & 1 & 0 & 0 \\
      0 & 0 & 0 & 1 & 0 & 0 & 1 & 0 \\
      0 & 0 & 0 & 1 & 0 & 0 & 0 & 0 \\
      0 & 0 & 0 & 0 & 0 & 1 & 1 & 0 \\
      0 & 0 & 0 & 0 & 0 & 1 & 0 & 0 \\
      0 & 0 & 0 & 0 & 0 & 0 & 1 & 0 \\
      0 & 0 & 0 & 0 & 0 & 0 & 0 & 0
    \end{array}\right).
  \]
and the other matrices $A_r$ with $r\neq 0$ are obtained from~$A_0$ by
permuting row $i$ with row $i \oplus r$ and column $j$ with column $j
\oplus r$. Using the maximum absolute column sum norm, we find
$\lambdajsr\leq 1$ and because~$1$ is an eigenvalue for all
matrices~$A_r$, $0\leq r<8$, we have $\lambdajsr=1$. 
\begin{figure} 
\begin{center}
\includegraphics[width=0.45\textwidth]{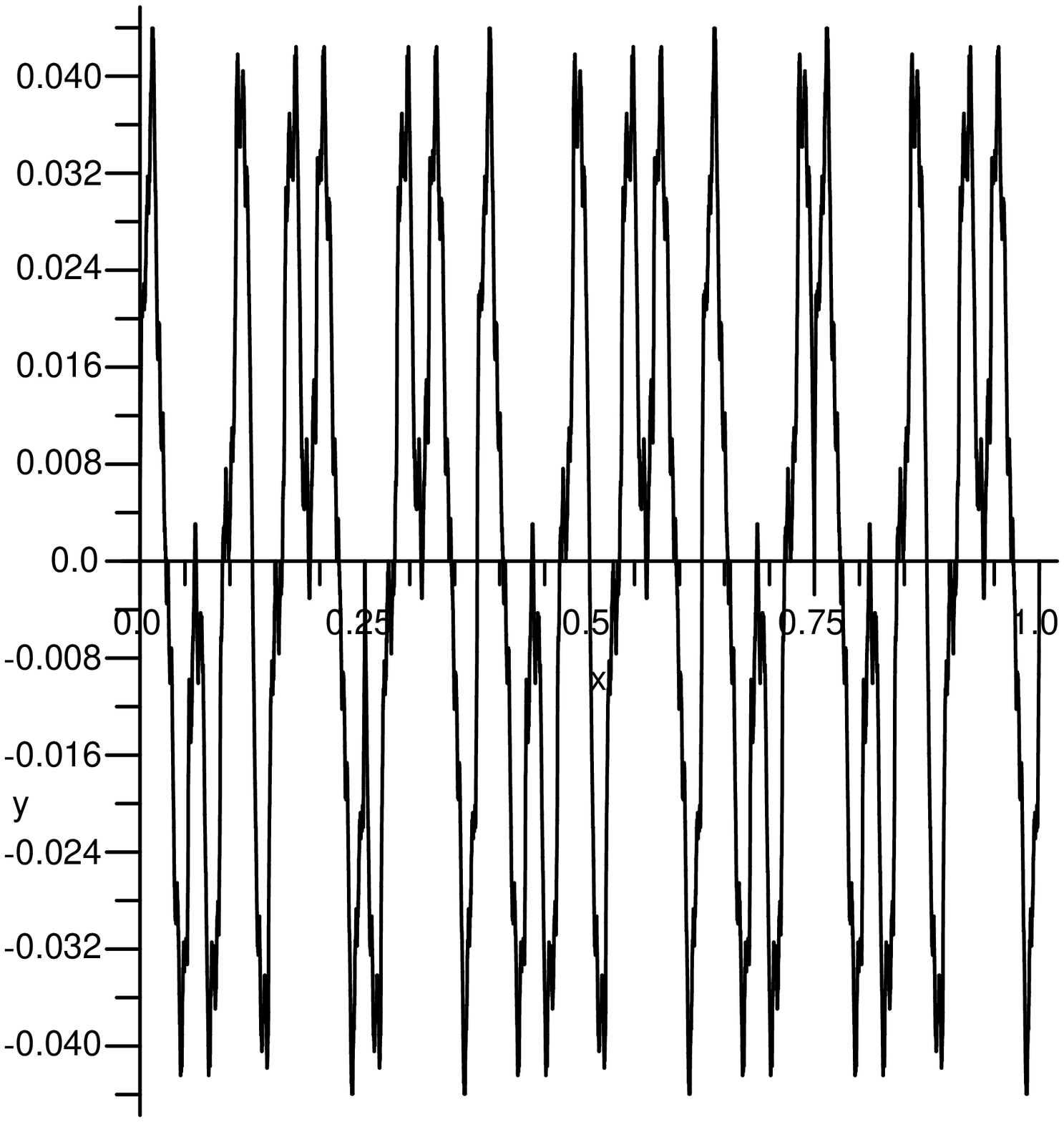}
\hfil
\includegraphics[width=0.45\textwidth]{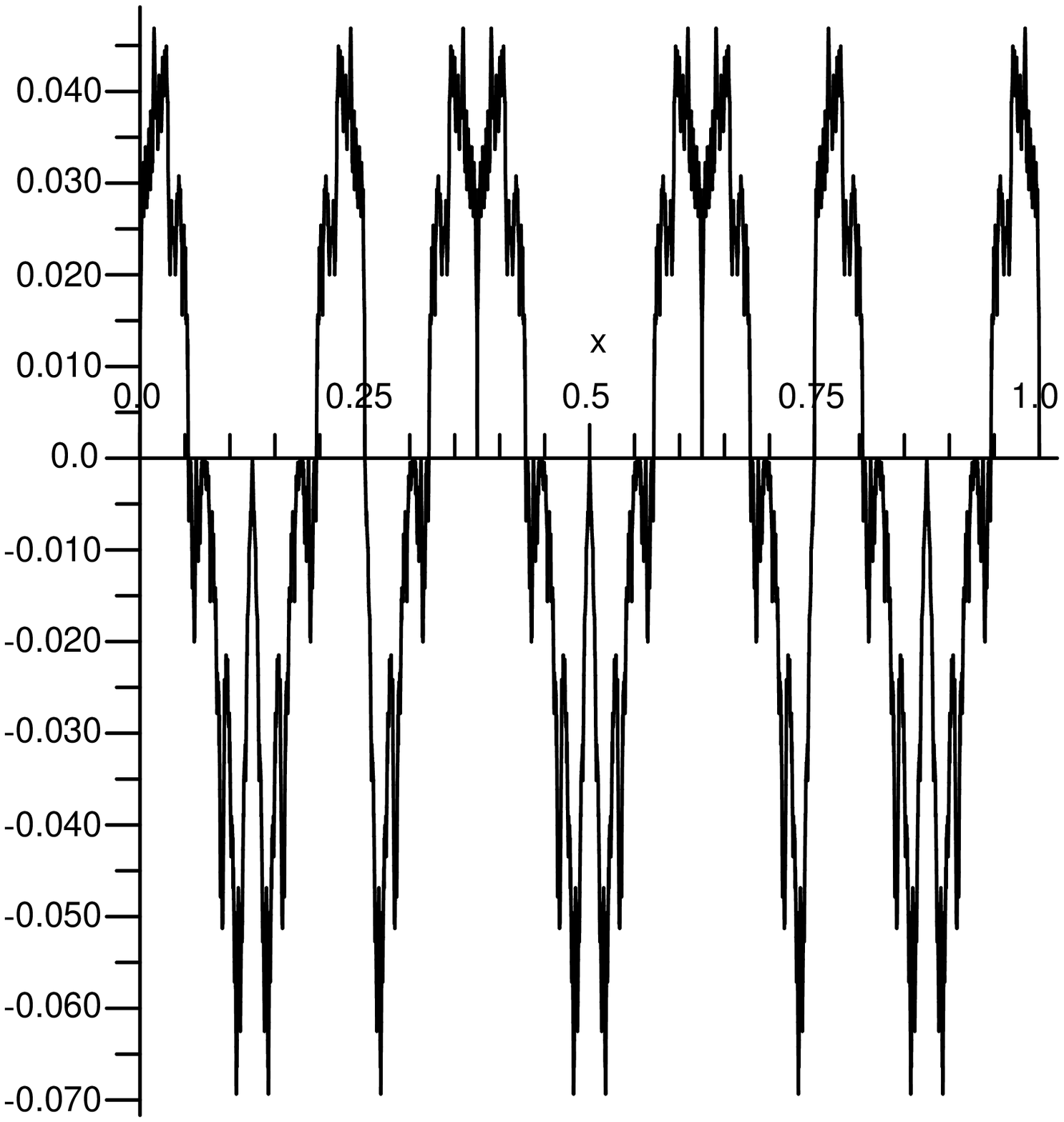}
    \end{center}
    \caption{\label{Du08vers09:fig:LipmaaWallenadp}Refinement
      functions which originate from cryptanalysis
      (Ex.~\ref{Du08vers09:ex:LipmaaWallen}).}
\end{figure}

The matrix 
\[
Q=A_0+A_1+\dotsb+A_7=
\frac{1}{4}
 \left( 
\begin{array}{cccccccc} 
7&2&2&1&2&1&1&0\\2&
7&1&2&1&2&0&1\\2&1&7&2&1&0&2&1\\1&
2&2&7&0&1&1&2\\2&1&1&0&7&2&2&1\\1&
2&0&1&2&7&1&2\\1&0&2&1&2&1&7&2\\0&
1&1&2&1&2&2&7
\end{array} \right) 
\]
is diagonalizable, with eigenvalues~$4$, $2$, $1$, and eigenspaces
respectively of dimension~$1$, $3$, $4$. The vector $C=E_1$ decomposes
as $C=V_4+V_2+V_1$, where
\[
V_4=\frac{1}{8}\left(\begin{array}{c}
1\\1\\1\\1\\1\\1\\1\\1
		     \end{array}\right),\qquad
V_2=\frac{1}{8}\left(\begin{array}{c}
3\\1\\1\\-1\\1\\-1\\-1\\-3
		     \end{array}\right),\qquad
V_1=\frac{1}{4}\left(\begin{array}{c}
2\\-1\\-1\\0\\-1\\0\\0\\1
		     \end{array}\right)
\]
are respectively eigenvectors for the eigenvalues~$4$, $2$,
$1$. According to
Theorem~\ref{Du08vers09:thm:formalasymptoticexpansion}, we obtain the
following asymptotic expansion (with $F_{\rho}=L\boldF_{\rho}$ for
$\rho=4$, $2$)
\[
\sum_{\begin{subarray}{c}\length w =K\\(0.w)_{{8}}\leq x
      \end{subarray}}\operatorname{adp}(w)\mathop{=}_{K\to+\infty}
4^K F_4(x)+2^K F_2(x)+O(K),
\]
where the functions~$\boldF_4$ and~$\boldF_2$ are determined by the
dilation equation (as usual~$x_1$ is the first octal digits after
the octal point of $x\in\left[0,1\right)$)
\[
\boldF_{\rho}(x)=\frac{1}{\rho}\sum_{r<x_1}A_rV_{\rho}+
\frac{1}{\rho}A_{x_1}\boldF_{\rho}(8x-x_1),\qquad
\text{with $\boldF_{\rho}(0)=0$ and $\boldF_{\rho}(1)=V_{\rho}$}
\]
for $\rho=4$, $2$. Functions~$\boldF_4$ and~$\boldF_2$ are
respectively H\"older with exponent $\log_84=2/3$ and $\log_82=1/3$.
Function $F_4=L\boldF_4$ is very near the identity and the graph of
$x\mapsto F_4(x)-x$ is drawn on the left of
Fig.~\ref{Du08vers09:fig:LipmaaWallenadp}. On the right is the graph
of $F_2=L\boldF_2$. Both functions have an amplitude of order about
$0.05$. The obtained result improves~\citep{DuLiWa07} which gives only
an error term~$O(2^K)$.
\end{example}

\section{Radix-rational sequences and periodic functions}\label{Du08vers09:sec:main:periodicfunction}
It is usual that sequences rational with respect to a radix introduce
periodic functions in logarithmic scale. These periodic functions
appear either by subtle elementary arguments, like explicit expression
of the sequence with fractional part of a logarithm in base~$\base $,
or by sophisticated argument which express partial sums by contour
integral of a meromorphic function whose poles are regularly spaced on
a vertical line. With the approach used in this paper, the
introduction of periodic functions is natural. Nevertheless these
functions are generally only pseudo-periodic functions and we will see
why they are periodic in the previous works on the subject.

\subsection{Asymptotic expansion}
To a rational series~$S$ defined by a linear representation~$L$,
$A_0$, $\ldots$, $A_{\base -1}$, $C$ is associated a sequence~$u$
defined in the following way: we write the integer~$n$ with respect to
the radix~$\base$, $n=\Bint w$, and the value of the sequence for~$n$
is the value of the rational series for the word~$w$, that is
$u(n)=(S,w)$.  We want to evaluate the partial sums $\sum_{0\leq n\leq
  N}u(n)$ of the sequence for~$N$ large. With $t\in\left[0,1\right)$
and~$K$ integer, we may write $N=\base ^{K+t}$, that is~$K$ is the
integer part of~$\log_{\base}N$ and~$t$ is its fractional part. The
partial sum up to~$N$ is the sum of the partial sum up to~$\base^K-1$
and a complementary term. The first appears as the sum of~$(S,w)$ over
all words~$w$ of length not greater than~$K$ minus the sum over all the
previous words which begin with a~$0$. The complementary term appears
as the running sum associated to~$S$ we have studied, but without the
words of length~$K+1$ which begin by a~$0$. This remark gives the
following formula
\begin{multline*}
  \sum_{n\leq \base^{K+t}}u(n)=
\sum_{0\leq k\leq K}\biggl(\sum_{\length w =k}LA_wC
-\sum_{\length{w'}=k-1}LA_0A_{w'}C\biggr)\\
+
\biggl(
\sum_{\begin{subarray}{c}\length w=K+1\\ \Bint w  \leq \base^{K+1}\base^{t-1}
  \end{subarray}}LA_wC
-\sum_{\length{w'}=K}LA_0A_{w'}C\biggr)
\end{multline*}
that is
\begin{equation}\label{Du08vers09:eq:runningsumu}
  \sum_{n\leq \base^{K+t}}u(n)=L(\operatorname{I}_d-A_0)
\sum_{0\leq k\leq K}Q^kC+
\sum_{\begin{subarray}{c}\length w=K+1\\ \Bint w \leq \base^{K+1}\base^{t-1}
  \end{subarray}}LA_wC.
\end{equation}
Just as we have considered the sums $\boldS_K(x)$, we introduce the
vector-valued sums
\begin{equation}\label{Du08vers09:eq:vectorvaluedrunningsumu}
\boldSigma_N=
\sum_{(w)_\base\leq N}A_wC,
\end{equation}
that is 
\begin{equation}\label{Du08vers09:eq:vectorvaluedrunningsumubis}
\boldSigma_N=
(\operatorname{I}_d-A_0)
\sum_{0\leq k\leq K}Q^kC+
\sum_{\begin{subarray}{c}\length w=K+1\\ \Bint w \leq \base^{K+1}\base^{t-1}
  \end{subarray}}A_wC
\end{equation}
with $N=\base^{K+t}$, $0\leq t<1$. As in the previous section, we
decompose the column vector~$C$ over a Jordan basis for the
matrix~$Q$. Hence we consider for the rest of the section a family of
linearly independent vectors $(V^{(j)})_{0\leq j<\nu}$ which satisfies
$QV^{(j)}=\rho\omega V^{(j)}+V^{(j-1)}$ for $j>0$ and
$QV^{(0)}=\rho\omega V^{(0)}$, with $\rho>\lambdajsr$ and
$|\omega|=1$. In other words, we assume
Hypotheses~\ref{Du08vers09:hypo:generalizedeigenvector}
and~\ref{Du08vers09:hypo:jointspectralradius}.  We want to determine
an asymptotic expansion for the running sums $\boldSigma_N$ associated
to $V=V^{(\nu-1)}$.

Formula~\eqref{Du08vers09:eq:vectorvaluedrunningsumubis} expresses the
running sums of the sequence as the sum of two terms.  For the first
term in the right member of
Formula~\eqref{Du08vers09:eq:vectorvaluedrunningsumubis}, we use
formul\ae\ like Eq.~\eqref{Du08vers09:eq:Q^KV} and the following
elementary summation formula
\begin{equation}\label{Du08vers09:eq:summatoryformulaQ^KV}
\sum_{k=\ell}^K\binom k \ell \alpha^{k-\ell}=
\left\{
\begin{array}{ll}
\displaystyle
\binom K {\ell+1} + \binom K \ell & \begin{array}{l}\text{if $\alpha=1$,}\end{array}\\[2.ex]
\begin{aligned}
\displaystyle
\binom K \ell \frac{\alpha^{K-\ell+1}}{\alpha-1}  +
\sum_{j=1}^{\ell}(-1)^j\binom K
{\ell-j}\frac{\alpha^{K-\ell+j}}{(\alpha-1)^{j+1}} &
\\
\displaystyle \mbox{} +
 (-1)^{\ell+1} \frac{1}{(\alpha-1)^{\ell+1}} &
\end{aligned}
& \begin{array}{l}\phantom{\displaystyle\sum_{j=1}^{\ell}}
\\\text{if $\alpha\neq 1$.}\end{array}
\end{array}
\right.
\end{equation}
As a consequence, the first term of
Formula~\eqref{Du08vers09:eq:vectorvaluedrunningsumubis} has an
asymptotic expansion in the asymptotic scale $\rho^K\binom K \ell$,
$\rho>0$, $\ell\in {\mathbb N}_{\geq 0}$.
\begin{lemma}\label{Du08vers09:lemma:firsttermasymptotic}
  Under Hypothesis~\ref{Du08vers09:hypo:generalizedeigenvector} with
  $C=V^{(\nu-1)}$ the first term in the right member of
  Formula~\eqref{Du08vers09:eq:vectorvaluedrunningsumubis} has the
  following expansion
\begin{multline}\label{Du08vers09:eq:Q^KVexpansioncaserhomeganeq1}
   (\operatorname{I}_d-A_0) \sum_{k=0}^KQ^kV^{(\nu-1)}=
\binom{K}{\nu-1}\frac{(\rho\omega)^{K-\nu+2}}{\rho\omega-1}(\operatorname{I}_d-A_0)V^{(0)}\\
+
\sum_{\ell=0}^{\nu-2}\binom{K}{\ell}(\rho\omega)^{K-\ell}\left[
\frac{\rho\omega}{\rho\omega-1}(\operatorname{I}_d-A_0)V^{(\nu-\ell-1)}
\phantom{\sum_{j=\ell+1}^{\nu-1}}
 \right.\\\left.
+
\sum_{j=\ell+1}^{\nu-1}\frac{(-1)^{j-\ell}}{(\rho\omega-1)^{j-\ell+1}}(\operatorname{I}_d-A_0)V^{(\nu-j-1)}
\right]\\
+
\sum_{\ell=0}^{\nu-1}(-1)^{\ell+1} \frac{1}{(\rho\omega-1)^{\ell+1}}(\operatorname{I}_d-A_0)V^{(\nu-\ell-1)}
\end{multline}
in case $\rho\omega\neq1$ and
\begin{multline}\label{Du08vers09:eq:Q^KVexpansioncaserhomega=1}
 (\operatorname{I}_d-A_0)\sum_{k=0}^K Q^kV^{(\nu-1)}=
\binom K\nu (\operatorname{I}_d-A_0) V^{(0)}\\ +
\sum_{\ell=1}^{\nu-1}\binom K\ell (\operatorname{I}_d-A_0)\left[V^{(\nu-\ell)}
+V^{(\nu-\ell-1}\right]+(\operatorname{I}_d-A_0)V^{(\nu-1)}
\end{multline}
in case $\rho\omega=1$.
\end{lemma}
\noindent
The previous expansion is not an asymptotic expansion but it may
easily converted into such an expansion in the scale $\rho^K\binom
K\ell$, $\rho>0$, $\ell\in\bN_{\geq0}$, with
$K=\lfloor\log_{\base}N\rfloor$. The error term depends on the
relative places of~$\lambdajsr$, $\rho$, and~$1$. We will return on
this point in a while.

The second term in the right member of
Formula~\eqref{Du08vers09:eq:vectorvaluedrunningsumubis} is nothing but
$\boldS_{K+1}(\base^{t-1})$.
Proposition~\ref{Du08vers09:thm:Jordanasymptexp} translates into the
following result.
\begin{lemma}
  Under Hypotheses~\ref{Du08vers09:hypo:generalizedeigenvector}
  and~\ref{Du08vers09:hypo:jointspectralradius}, the second term in
  the right member of
  Formula~\eqref{Du08vers09:eq:vectorvaluedrunningsumubis} for the
  running sum associated to~$V^{(\nu-1)}$ admits the following
  asymptotic expansion with respect to the scale $\rho^K\binom K\ell$,
  $\rho>0$, $\ell\in\bN_{\geq0}$, with
  $K=\lfloor\log_{\base}N\rfloor$, where the functions~$\boldF^{(\ell)}$'s
  are defined in Lemma~\ref{Du08vers09:thm:Jordanasymptexp},
\begin{multline}\label{Du08vers09:eq:expansionofthesecondterm}
\sum_{\begin{subarray}{c}\length w=K+1\\ \Bint w \leq \base^{K+1}\base^{t-1} \end{subarray}}
  A_wV^{(\nu-1)}\mathop{=}_{N\to+\infty}
\binom{K}{\nu-1}(\rho\omega)^{K-\nu+2}\boldF^{(0)}(\base^{t-1})\\
+
\sum_{\ell=0}^{\nu-2}\binom{K}{\ell}(\rho\omega)^{K-\ell}\left[
\rho\omega \boldF^{(\nu-\ell-1)}(\base^{t-1})+\boldF^{(\nu-\ell-2)}(\base^{t-1})
\right]+O(\lambda^K).
\end{multline}
for every $\lambda>\lambdajsr$ and the big oh is uniform with respect
to~$t$ in~$[0,1]$. Moreover if $\lambdajsr$ is attained we may
replace~$\lambda$ by~$\lambdajsr$ in the asymptotic expansion.
\end{lemma}
\begin{proof}
  Proposition~\ref{Du08vers09:thm:Jordanasymptexp} and the fundamental
  formula of binomial coefficients provide the expansion of
  $\boldS_{K+1}(\base^{t-1})$ by a mere substitution.
\end{proof}

Gathering the results of the previous discussion, we arrive at the
following qualitative theorem. 
\begin{theorem}\label{Du08vers09:thm:sequenceasymptoticexpansion}
  Let~$L$, $(A_r)_{0\leq r<\base}$, $C$ be a linear representation of
  a radix-rational sequence~$(u_n)$. The vector-valued running sum
  $\boldSigma_N$ defined by
  Equation~\eqref{Du08vers09:eq:vectorvaluedrunningsumu} admits an
  asymptotic expansion with error term~$O(N^{\log_{\base}\lambda})$
  for every $\lambda>\lambdajsr$, where~$\lambdajsr$ is the joint
  spectral radius of the family~$(A_r)_{0\leq r<\base}$. If
  $\lambdajsr$ is attained the error term may be replaced by
  $\Ostar(N^{\log_{\base}\lambdajsr})$, where the soft big oh hides
  a power of $\log N$. The used asymptotic scale is the family of
  sequences $N^{\alpha}\binom {\lfloor \log_{\base}N\rfloor} \ell$,
  $\alpha\in\bR$, $\ell\in\mathbb N_{\geq 0}$.
\end{theorem}
Theorem~\ref{Du08vers09:thm:sequenceasymptoticexpansion} translates
concretely into Algorithm {\tt LRtoAE2} of
pages~\pageref{Du08vers09:algoforrationalsequencespiece1},
\pageref{Du08vers09:algoforrationalsequencespiece2}. It is a simple
variation on Algorithm {\tt LRtoAE1}, but is a little more complicated
because it takes into account the relative positions of the
parameter~$\lambda$, which governs the error term, and of the
number~$1$. The right member of
Eq.~\eqref{Du08vers09:eq:expansionofthesecondterm} causes no problem
because it is readily seen as an asymptotic expansion in the scale
$N^{\alpha}\binom {\lfloor \log_{\base}N\rfloor} \ell$. The right
member of Eq.~\eqref{Du08vers09:eq:Q^KVexpansioncaserhomeganeq1}
and~\eqref{Du08vers09:eq:Q^KVexpansioncaserhomega=1} is more
annoying. In case $\lambda>1$, we take into account only the
eigenvalues with modulus $\rho>1$ and the case $\rho\omega=1$ of
Lemma~\ref{Du08vers09:lemma:firsttermasymptotic} cannot happen. For
the case $\rho>1$, we consider in the right member of
Eq.~\eqref{Du08vers09:eq:Q^KVexpansioncaserhomeganeq1} only the
quantity
\begin{multline*}
X=\binom{K}{\nu-1}\frac{(\rho\omega)^{K-\nu+2}}{\rho\omega-1}(\operatorname{I}_d-A_0)V^{(0)}\\
+
\sum_{\ell=0}^{\nu-2}\binom{K}{\ell}(\rho\omega)^{K-\ell}\left[
\frac{\rho\omega}{\rho\omega-1}(\operatorname{I}_d-A_0)V^{(\nu-\ell-1)}
\phantom{\sum_{j=\ell+1}^{\nu-1}}
 \right.\\\left.
+
\sum_{j=\ell+1}^{\nu-1}\frac{(-1)^{j-\ell}}{(\rho\omega-1)^{j-\ell+1}}(\operatorname{I}_d-A_0)V^{(\nu-j-1)}
\right]
\end{multline*}
because the quantity
\[
Y=\sum_{\ell=0}^{\nu-1}(-1)^{\ell+1} \frac{1}{(\rho\omega-1)^{\ell+1}}(\operatorname{I}_d-A_0)V^{(\nu-\ell-1)}
\]
is $O(1)$ and negligeable in that case. (See
lines~\ref{lrtoaeline3}--\ref{lrtoaeline1}
and~\ref{lrtoaeline8}--\ref{lrtoaeline2} of Algorithm {\tt LRtoAE2}.) 
In case $\lambda<1$, all terms must be taken into account but we have
to distinguish between the cases $\rho\omega\neq1$ and 
$\rho\omega=1$. (See lines~\ref{lrtoaeline3}--\ref{lrtoaeline6},
\ref{lrtoaeline4}--\ref{lrtoaeline7},
and~\ref{lrtoaeline9}--\ref{lrtoaeline10}.) Lastly in case
$\lambda=1$, the case $\rho\omega=1$ does not appear because it is
assumed $\rho>\lambda$, but it may have an effect on the error term if
$\lambdajsr$ is attained. (See
lines~\ref{lrtoaeline11}--\ref{lrtoaeline12} and
~\ref{lrtoaeline5}--\ref{lrtoaeline13}.) The case $\rho\omega\neq1$
provide the term~$X$. (See lines~\ref{lrtoaeline3}--\ref{lrtoaeline1}
and~\ref{lrtoaeline8}--\ref{lrtoaeline2}.) The term~$Y$, which
is~$O(1)$, comes automatically into the error term, which is
$O(\lambda^K)$ that is~$O(1)$.

\begin{algorithm}
\caption{\label{Du08vers09:algoforrationalsequencespiece1} {\tt LRtoAE2} (beginning)}
\SetKwInOut{KwInput}{Input}
\SetKwInOut{KwOutput}{Output}
\SetLine
\SetArgSty{rm}
  \KwInput{A linear representation $L$, $(A_r)_{0\leq r<\base}$, $C$,
  and its joint spectral radius~$\lambdajsr$.}
  \KwOutput{An asymptotic expansion of the running sum $\boldSigma_N$
  associated to the vector~$C$ when~$N$ goes to infinity at the
  precision $\Ostar(N^{\log_\base\lambdajsr})$ if~$\lambdajsr$ is
  attained and $O(N^{\log_\base\lambda})$ for a $\lambda$ slightly
  greater than $\lambdajsr$ if not, with respect to the scale
  $N^{\alpha}\binom {\lfloor\log_\base N\rfloor} \ell$, $\alpha\in\bR$,
  $\ell\in\mathbb N_{\geq 0}$.}
  $Q:=\sum_{0\leq r<\base}A_r$\;
  \eIf{$\lambdajsr$ is attained\nllabel{lrtoaeline11}}
      {$\lambda:=\lambdajsr$\nllabel{lrtoaeline12}}
      {$\lambda:=$ any number between $\lambdajsr$ and the infimum of
      the modulus of eigenvalues of~$Q$ greater than~$\lambdajsr$}
  compute a Jordan basis $\cV$ for the matrix $Q$\;
  expand the column vector~$C$ of the linear representation over the
  Jordan basis, as $C=\sum_{V\in\cV}\gamma_VV$\;
  $\cV_{>}:=$ the set of generalized eigenvectors in~$\cV$ such that $\gamma_V\neq0$ and the associated eigenvalue~$\rho\omega$ has a modulus $\rho>\lambda$\;
%
  \For{each vector~$V$ in~$\cV_{>}$} 
      {$\rho\omega:=$ the eigenvalue associated to~$V$\;
       $\nu:=$ the height of~$V$\;
       \eIf{$\rho\omega\neq 1$\nllabel{lrtoaeline3}}
            {
             using Eq.~\eqref{Du08vers09:eq:Q^KVexpansioncaserhomeganeq1}
             \nllabel{lrtoaeline1}
             \begin{multline*}\!\!\!\!\!\! 
	X:=\binom{K}{\nu-1}\frac{(\rho\omega)^{K-\nu+2}}{\rho\omega-1}(\operatorname{I}_d-A_0)V^{(0)}+\mbox{}\\
       \sum_{\ell=0}^{\nu-2}\binom{K}{\ell}(\rho\omega)^{K-\ell}\left[
\frac{\rho\omega}{\rho\omega-1}(\operatorname{I}_d-A_0)V^{(\nu-\ell-1)}
+
\sum_{j=\ell+1}^{\nu-1}\frac{(-1)^{j-\ell}}{(\rho\omega-1)^{j-\ell+1}}(\operatorname{I}_d-A_0)V^{(\nu-j-1)}
\right]; 
            \end{multline*}\\
         \raisebox{0.0ex}[1.0ex][0.0ex]{$\displaystyle Y:=\sum_{\ell=0}^{\nu-1}(-1)^{\ell+1} \frac{1}{(\rho\omega-1)^{\ell+1}}(\operatorname{I}_d-A_0)V^{(\nu-\ell-1)};$} 
         \\[2.0ex]
         $Z:=X+Y$\nllabel{lrtoaeline6}}
        {\nllabel{lrtoaeline4}using Eq.~\eqref{Du08vers09:eq:Q^KVexpansioncaserhomega=1}
        \[
        Z:=\binom K\nu (\operatorname{I}_d-A_0) V^{(0)} +
\sum_{\ell=1}^{\nu-1}\binom K\ell (\operatorname{I}_d-A_0)\left[V^{(\nu-\ell)}
+V^{(\nu-\ell-1}\right]+(\operatorname{I}_d-A_0)V^{(\nu-1)}
        \]
\nllabel{lrtoaeline7}}
    \eIf{$\lambda\geq1$\nllabel{lrtoaeline8}}
        {$\boldA_{V;N}:=X$\nllabel{lrtoaeline2}}
        {\nllabel{lrtoaeline9}$\boldA_{V;N}:=Z$\nllabel{lrtoaeline10}}
  according to Eq.~\eqref{Du08vers09:eq:expansionofthesecondterm}, a
  second term is added
  \begin{multline*}
  \boldA_{V;N}:=\boldA_{V;N}+
\binom{K}{\nu-1}(\rho\omega)^{K-\nu+2}\boldF^{(0)}(\base^{t-1})\\
+
\sum_{\ell=0}^{\nu-2}\binom{K}{\ell}(\rho\omega)^{K-\ell}\left[
\rho\omega \boldF^{(\nu-\ell-1)}(\base^{t-1})+\boldF^{(\nu-\ell-2)}(\base^{t-1})
\right]
  \end{multline*}
}
  $\boldA_N:=\sum_{V\in\cV_{>}}\gamma_{V}\boldA_{V;N}$\;
\setcounter{DuviiiLastNumberLineofAlgo}{\theAlgoLine}
\end{algorithm}
\addtocounter{algocf}{-1}
\begin{algorithm}
\setcounter{AlgoLine}{\theDuviiiLastNumberLineofAlgo}
\caption{\label{Du08vers09:algoforrationalsequencespiece2} {\tt LRtoAE2} (end)}
\SetLine
\SetArgSty{rm}
  $\operatorname{error~term}_N:=O(N^{\log_\base \lambda})$\;
  \If{$\lambdajsr$ is attained\nllabel{lrtoaeline5}}
     {$\cV_=:=$ the set of generalized eigenvectors in~$\cV$ such that
     $\gamma_V\neq0$ and the associated eigenvalue~$\rho\omega$ has a
     modulus $\rho=\lambdajsr$\;
     $m:=0$\;
    \For{each~$V$ in~$\cV_{=}$}
        {$\nu:=$ the height of~$V$\;
         $m:=\max(m,\nu)$}
     $\operatorname{error~term}_N:=O(N^{\log_\base \lambdajsr}\log^m_\base N)$\nllabel{lrtoaeline13}}
  \Return{
     \[\displaystyle \boldSigma_N\mathop{=}_{N\to+\infty}\boldA_N+\operatorname{error~term}_N\]
  }
\end{algorithm}

The asymptotic expansion has essentially the form
\[
\boldSigma_N\mathop{=}_{N\to+\infty}
\sum_{\rho>\lambdajsr,\,\ell\geq0}N^{\log_{\base}\rho}\binom{\lfloor \log_{\base}N\rfloor} \ell
\sum_{\omega}\omega^{\lfloor \log_{\base}N\rfloor}\Phi_{\rho,\ell,\omega}(\log_{\base}N)
+
\Ostar(N^{\log_{\base}\lambdajsr}),
\]
where~$\rho$ is the modulus of any eigenvalue of~$Q$, $\omega$ is a
complex number with modulus~$1$, and $\ell$ is an integer related to
the maximal size of the Jordan blocks associated to~$\rho\omega$. As
regards $\Phi$, we change the status of the variable~$t$. It was only
an abbreviation for $\{\log_{\base}N\}$. Now we see it as a real
variable and we substitute to it the logarithm base~$\base$ of~$N$ to
obtain the expression of the expansion. As a consequence
the functions~$\Phi$ appear as $1$-periodic functions.

The asymptotic expansion has variable
coefficients~\citep[Chapter~V]{Bourbaki76}. The coefficients are taken
from the vector space generated by sequences which write
$(\omega^{\lfloor\log_{\base}N\rfloor}\Phi(\log_{\base}N))$ with
$|\omega|=1$ and $\Phi$ a $1$-periodic function.

\begin{example}[Discrepancy of the van der Corput sequence]\label{Du08vers09:ex:vdCmeandiscrepancy}
\citet{MR0444600,MR520308} show that the discrepancy of the van der
Corput sequence satisfies the following recursion 
\[
D(1)=1,\qquad D(2n)=D(n),\qquad D(2n+1)=\frac{1}{2}(D(n)+D(n+1)+1)
\]
and we add $D(0)=0$. Let us recall first that the (binary) van der Corput
sequence is defined as follows: for an integer~$n$, we write its
binary expansion $(n_{K-1}\ldots n_1n_0)_2$; we reverse it and we place
it after the binary dot. The real number $(0.n_0n_1\ldots n_{K-1})_2$ is
the value~$x_n$ of the van der Corput sequence for the integer~$n$. Second
the discrepancy of the sequence~$(x_n)$ is 
\[
D(n)=\sup_{0\leq \alpha<\beta\leq 1}\left|\frac{\nu(n,\alpha,\beta)}{n}-(\beta-\alpha)\right|,
\]
where $\nu(n,\alpha,\beta)$ is the number of terms~$x_k$, $1\leq k\leq
n$, which fall in the interval~$\left[\alpha,\beta\right)$. It
measures the deviation from the uniform distribution for the
sequence~$(x_n)$ \citep{Niederreiter92}. We want to evaluate more
precisely the mean value of the sequence~$(D(n))$ given by
\citep[Th.~3]{MR520308}
\[
\frac{1}{N}\sum_{n=1}^ND(n)\mathop{=}_{N\to+\infty}\frac{1}{4}\log_2N+O(1).
\]
The sequence~$(D(n))$ is $2$-rational and admits the following linear
representation
\[
L=\left(\begin{array}{ccc}
0 & 1 & 1
	\end{array}\right),\quad
A_0=\left(\begin{array}{ccc}
1 & 1/2 & 0 \\
0 & 1/2 & 0 \\
0 & 1/2 & 1
	\end{array}\right),\quad
A_1=\left(\begin{array}{ccc}
1/2 & 0 & 0 \\
1/2 & 1 & 0 \\
1/2 & 0 & 1
	\end{array}\right),\quad
C=\left(\begin{array}{c}
1 \\ 0 \\ 0
	\end{array}\right).
\]
with respect to the generating family $(D(n),D(n+1),1)$.  Numerical
computations lead us to think that the joint spectral radius of the
representation is $\lambdajsr=1$. To prove this result we use the Lie
algebra $\frakg$ generated by~$A_0$ and~$A_1$ (with the usual bracket
product $[A,B]=AB-BA$). The derived series of a Lie algebra $\frakg$
is the filtration $\frakg=D^0\frakg\supset D^1\frakg\supset
\dotsb\supset D^k\frakg\supset\dotsb$ defined inductively by
$D^1\frakg=[\frakg,\frakg]$,
$D^k\frakg=[D^{k-1}\frakg,D^{k-1}\frakg]$. The Lie algebra~$\frakg$ is
defined to be solvable if $D^k\frakg=\{0\}$ for some nonnegative
integer~$k$. According to Lie's theorem, if~$\frakg$ is a solvable Lie
algebra over the algebraically closed field of complex numbers there
exists a change of basis matrix which makes triangular all matrices of
any representation of~$\frakg$~\citep{GoKh00}. As a consequence if~$\cA$
is a finite set of real matrices which generates a solvable Lie
algebra, the joint spectral radius of~$\cA$ is the maximum of the
spectral radii~$\rho(A)$ for~$A$ in~$\cA$. Moreover the joint spectral
radius is attained with a Euclidean norm which writes $x\mapsto
(x^{\trpse}Px)^{1/2}$ where~$P$ is a symmetric definite positive
matrix~\citep{BlNeTh05}.

To apply this result it remains to compute the derived series
of~$\frakg$. We use $B_0=2A_0$ and $B_1=2A_1$ to avoid
denominators. Let $\frakg$ be the Lie algebra generated by~$B_0$
and~$B_1$. We find at the same time a basis of this algebra and its
multiplication table, namely
\[
\begin{array}{c|c|c|c|c|}
[\ ,\ ] & B_0 & B_1 & B_2 & B_3 \\\hline
B_0 & 0 & -B_2 & -B_3 & -B_2 \\\hline
B_1 & B_2 & 0 & B_4 & -B_2 \\\hline
B_2 & B_3 & -B_4 & 0 & B_5 \\\hline
B_3 & B_2 & B_2 & -B_5 & 0 \\\hline  
\end{array}
\]
with 
\[
B_2=[B_1,B_0]=\left(\begin{array}{ccc}
-1 & 1 & 0 \\
1 & 1 & 0 \\
-1 & 1 & 0
  \end{array}\right),\qquad
B_3=[B_2,B_0]=\left(\begin{array}{ccc}
-1 & -1 & 0\\
1 & 1 & 0 \\
-1 & -3 & 0
  \end{array}\right)
\]
and $B_4=2B_0-2B_1+B_3$, $B_5=-2B_0+2B_1-2B_3$. Hence $\frakg$ admits
the basis $(B_0,B_1,B_2,B_3)$ and $D^1\frakg$ admits the basis
$(B_2,B_3,B_4)$. Next we find that $D^2\frakg$ has dimension~$1$ and
is generated by $B_3+B_4=2B_0-2B_1+2B_3$. As a consequence
$D^3\frakg=\{0\}$ and~$\frakg$ is solvable.  Further the joint
spectral radius is $\lambdajsr=1$ and it is attained. Because~$1$ is a
simple eigenvalue of $Q=A_0+A_1$, the error term of the asymptotic
expansion for the running sum is~$O(\log N)$.

Using the basis $(V_1,V_2^0,V_2^1)$ with
\[
V_1=\left(\begin{array}{c}
1/2 \\ -1/2 \\ 0
  \end{array}\right),\qquad
V_2^0=\left(\begin{array}{c}
0 \\ 0 \\ 1/2
  \end{array}\right),\qquad
V_2^1=\left(\begin{array}{c}
1/2 \\ 1/2 \\ 0
  \end{array}\right),
\]
the matrix~$Q$ takes the Jordan form 
\[
J=\left(\begin{array}{ccc}
1 & 0 & 0 \\
0 & 2 & 1 \\
0 & 0 & 2
  \end{array}\right).
\]
The vector~$C$ expands as $C=V_1+V_2^1$, and because~$V_1$ is related
to the eigenvalue~$1$ it may be neglected. We find 
\[
\boldS_K(x)=\frac{1}{2}2^K\,K\,\boldF^0(x)+2^K\,\boldF^1(x)
\]
where~$\boldF^0$ and~$\boldF^1$ are solutions of the dilation
equations
\begin{multline*}
\boldF^0(x)=\frac{1}{2}A_0\boldF^0(2x),\qquad \text{for $0\leq x<1/2$,}\\
\boldF^0(x)=\frac{1}{2}A_0V_2^0+\frac{1}{2}\boldF^0(2x-1),\qquad \text{for $1/2\leq x<1$;}
\end{multline*}
\begin{multline*}
\boldF^1(x)=-\frac{1}{2}\boldF^0(x)+
\frac{1}{2}A_0\boldF^0(2x),\qquad \text{for $0\leq x<1/2$,}\\
\boldF^1(x)=-\frac{1}{2}\boldF^0(x)+
\frac{1}{2}A_0V_2^1+\frac{1}{2}\boldF^1(2x-1),\qquad \text{for $1/2\leq x<1$}
\end{multline*}
with the boundary conditions $\boldF^0(0)=0$, $\boldF^0(1)=V_2^0$,
$\boldF^1(0)=0$, $\boldF^1(1)=V_2^1$. It is readily seen that
$\boldF^0(x)=xV_2^0$, but $\boldF^1$ is not
explicit. Eventually we arrive at the formula
\[
\frac{1}{N}\sum_{n=1}^ND(n)\mathop{=}_{N\to+\infty}
\frac{1}{4}\log_2N 
+
\frac{1}{4}\left(
1-\{t\}+2^{3-\{t\}}\left(F^1_2(2^{\{t\}-1})+F^1_3(2^{\{t\}-1}\right)
\right)
+
O\left(\frac{\log N}{N}\right).
\]
Figure~\ref{Du08vers09:fig:vdCmeandiscrepancy} show the empirical
periodic function and the comparison between the empirical and
theoretical periodic functions.

\begin{figure}
  \begin{center}
\includegraphics[width=0.30\linewidth]{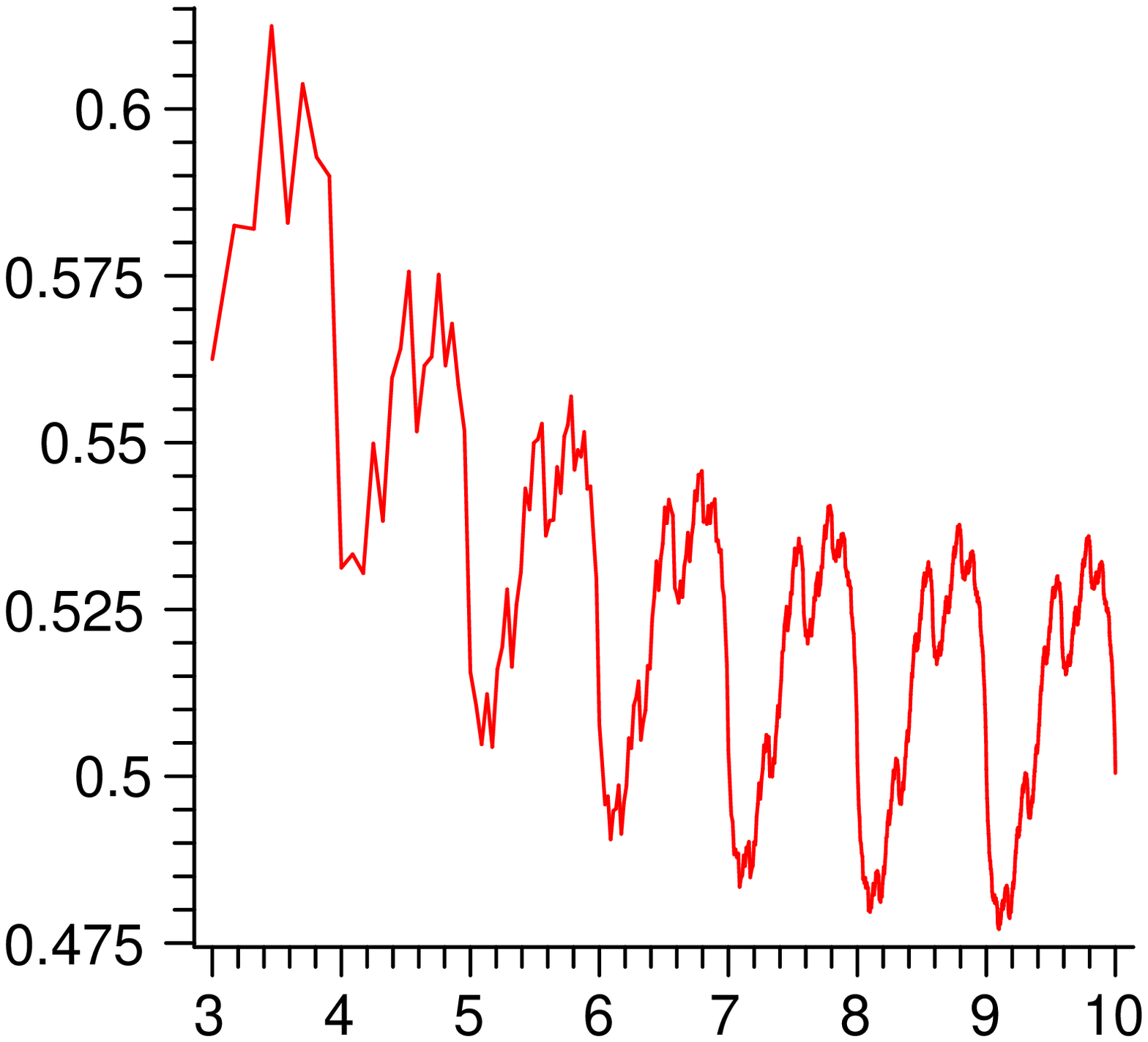}
\hfil
\includegraphics[width=0.56\linewidth]{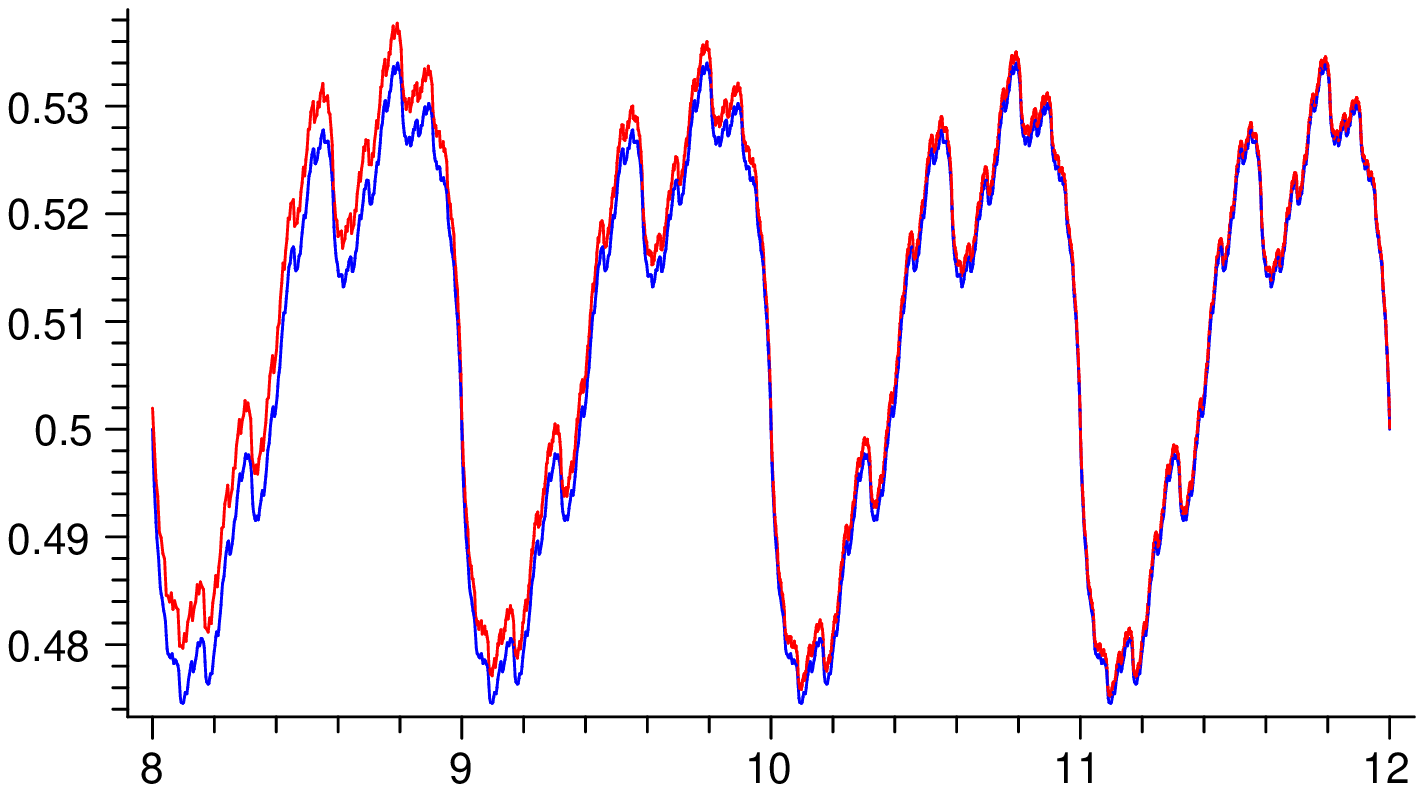}
  \end{center}
\caption{\label{Du08vers09:fig:vdCmeandiscrepancy}
On the left the empirical (red) periodic function of
Ex.~\ref{Du08vers09:ex:vdCmeandiscrepancy} (mean discrepancy of the
van der Corput sequence), and on the right the comparison between it
and the theoretical (blue) periodic function computed by the cascade
algorithm. }
\end{figure}
\end{example}

The coefficients of the asymptotic expansion of the running
sum~$\boldSigma_N$ inherit the properties of the coefficients of the
asymptotic expansion of $\boldS_K(x)$. To see this we need to be more
precise about the way we obtain the asymptotic expansion. Anew we use
the same formalism as in Section~\ref{Du08vers09:subsec:Jordan}.
\begin{corollary}
  Under Hypotheses~\ref{Du08vers09:hypo:generalizedeigenvector}
  and~\ref{Du08vers09:hypo:jointspectralradius} the asymptotic
  expansion of the running sum~$\boldSigma_N$ associated with
  $V^{(\nu-1)}$ has coefficients which are H\"older with exponent
  $\log_{\base}(\rho/\lambda)$ for $\lambda>\lambdajsr$ and~$\lambda$
  may be replaced by~$\lambdajsr$ if $\lambdajsr$ is attained.
\end{corollary}
\begin{proof}
  The only point which is to be verified is the continuity at
  integers, because up to this property the H\"olderian character is
  evident (the coefficients are expressed as combinations of solutions
  of dilation equations and functions which are smooth except perhaps
  to integers). As we have seen the expansion of~$\boldSigma_N$ comes from
  both terms in the right member of
  Formula~\eqref{Du08vers09:eq:vectorvaluedrunningsumubis}. The first is
  $(\operatorname{I}_d-A_0)\sum_{0\leq k\leq K}Q^kV^{(\nu-1)}$ and we
  do not want to expand it here, because the asymptotic expansion
  varies with the relative position of~$\rho$, $\lambdajsr$,
  and~$1$. Moreover this is not necessary since it depends on~$t$ only
  through~$K$. The expansion of the second term is given by
  Formula~\eqref{Du08vers09:eq:expansionofthesecondterm}. We note that
  this expression is merely the last column of the matrix
  $\bF(\base^{t-1})J_{\rho\omega}^{K+1}$ with the notations used in
  the proof of Lemma~\ref{Du08vers09:lemma:systemdilationequation}.
  Hence we are considering the matrix
\[
  \bA(t)=(\operatorname{I}_d-A_0)\sum_{0\leq k\leq K}Q^kV+
  \bF(\base^{t-1})J_{\rho\omega}^{K+1}, 
\] 
  where $V$ is the matrix whose columns are the column vectors
  $V^{(0)}$ , $V^{(1)}$, $\ldots$, $V^{(\nu-1)}$ as in the proof of
  Lemma~\ref{Du08vers09:lemma:systemdilationequation}. The last column
  of the matrix~$\bA(t)$ is the regular part of the expansion
  of~$\boldSigma_N$ (that is without the error term), and we want to
  verify that $t\mapsto \bA(t)$ is a continuous function at integers.
  Let~$K_0$ be an integer. When $t$ tends towards~$K_0$ from above, we
  have $K=\lfloor t\rfloor=K_0$ and
\[
\lim_{t\overset{>}{\rightarrow}K_0}\bA(t)=
(\operatorname{I}_d-A_0)\sum_{0\leq k\leq K_0}Q^kV+
\bF(1/\base)J_{\rho\omega}^{K_0+1},
\]
but according to the dilation
equation~\eqref{Du08vers09:eq:matrixsystemdilationequation},
$\bF(1/\base)=A_0VJ_{\rho\omega}^{-1}$ and we find
\[
\lim_{t\overset{>}{\rightarrow}K_0}\bA(t)=
(\operatorname{I}_d-A_0)\sum_{0\leq k\leq K_0}Q^kV+A_0VJ_{\rho\omega}^{K_0}.
\]
  On the other side when $t$ tends towards~$K_0$ from below, we have
  $K=\lfloor t\rfloor=K_0-1$ and
\[
\lim_{t\overset{<}{\rightarrow}K_0}\bA(t)=
(\operatorname{I}_d-A_0)\sum_{0\leq k\leq K_0-1}Q^kV+
\bF(1)J_{\rho\omega}^{K_0}=
(\operatorname{I}_d-A_0)\sum_{0\leq k\leq K_0-1}Q^kV+VJ_{\rho\omega}^{K_0}.
\]
  The difference is
\[
\lim_{t\overset{>}{\rightarrow}K_0}\bA(t)-
\lim_{t\overset{<}{\rightarrow}K_0}\bA(t)=
(\operatorname{I}_d-A_0)Q^{K_0}V+(A_0-\operatorname{I}_d)VJ_{\rho\omega}^{K_0}
\]
but Hypothesis~\ref{Du08vers09:hypo:generalizedeigenvector} writes
$QV=VJ_{\rho\omega}$ and we conclude that the difference is the null
matrix.
\end{proof}

Theorem~\ref{Du08vers09:thm:sequenceasymptoticexpansion} deserves some
comments. Functions which roughly write $\Phi(t)=\omega^K\rho^{1-t}
F(\base^{t-1})$ have not the self-similar character of the solution
~$F(x)$ of a dilation equation, because we have cut the piece of the
function between~$0$ and~$1/\base $. Moreover it is possible to use a
more ordinary scale by expanding $\binom{\lfloor
\log_{\base}N\rfloor}{\ell}=\binom{-t+\log_{\base}N}{\ell}$ as a
polynomial in $\log_{\base}N$ with coefficients in
$t=\{\log_{\base}N\}$. But doing this we hide completely the structure
of the asymptotic expansion and the dilation equations which are
behind it. This is certainly the reason why these equations have not
been perceived before, even if the word ``fractal'' is frequently
employed.
\begin{example}[Coquet sequence]
  The Coquet sequence~\citep{Coquet83} is defined as
  $u(n)=(-1)^{s_2(3n)}$. (We recall that $s_2(n)$ is the sum of the
  bits in the binary expansion of~$n$.)
\begin{figure}
\includegraphics[width=0.45\linewidth]{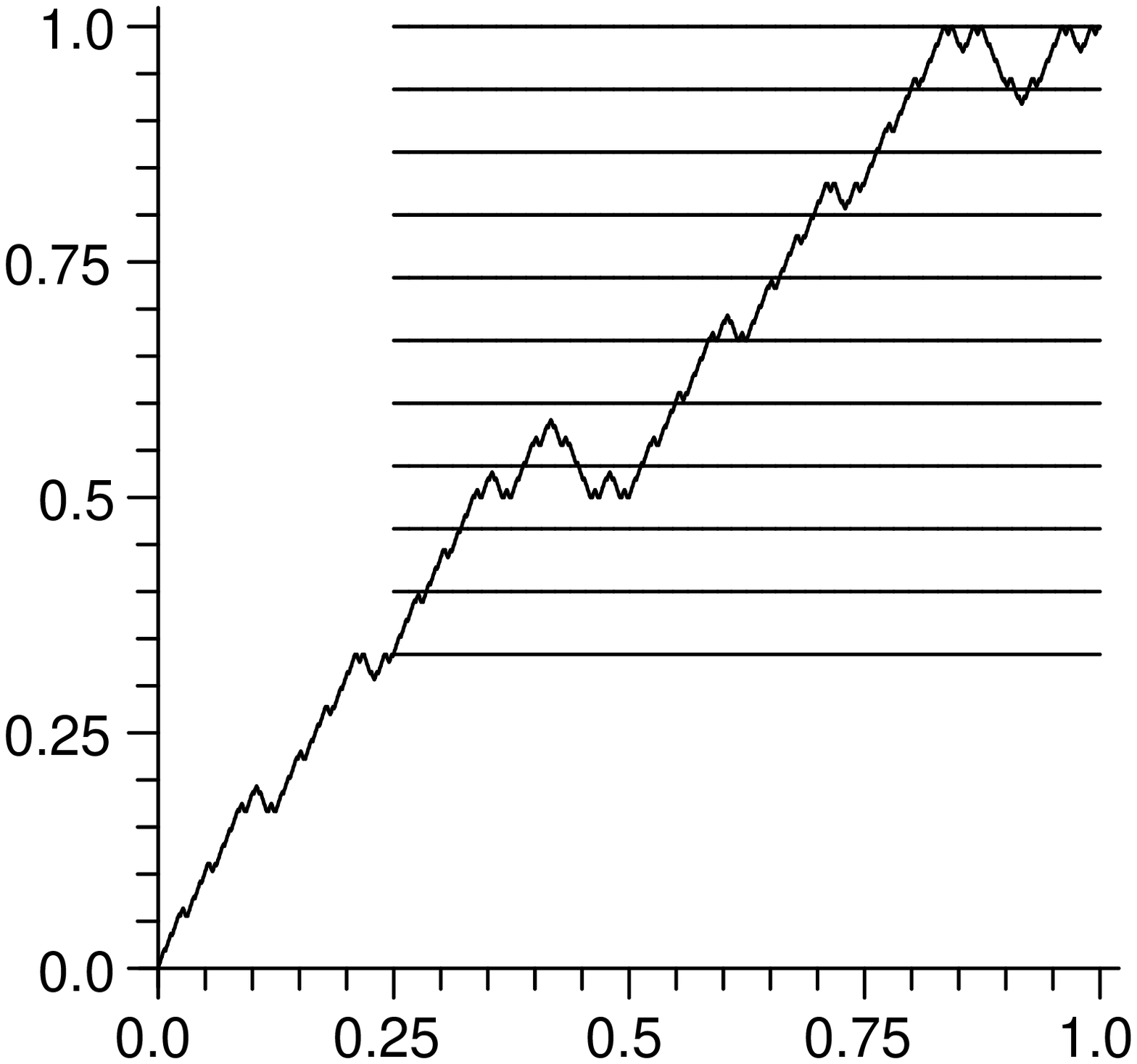}
\hfil
\includegraphics[width=0.45\linewidth]{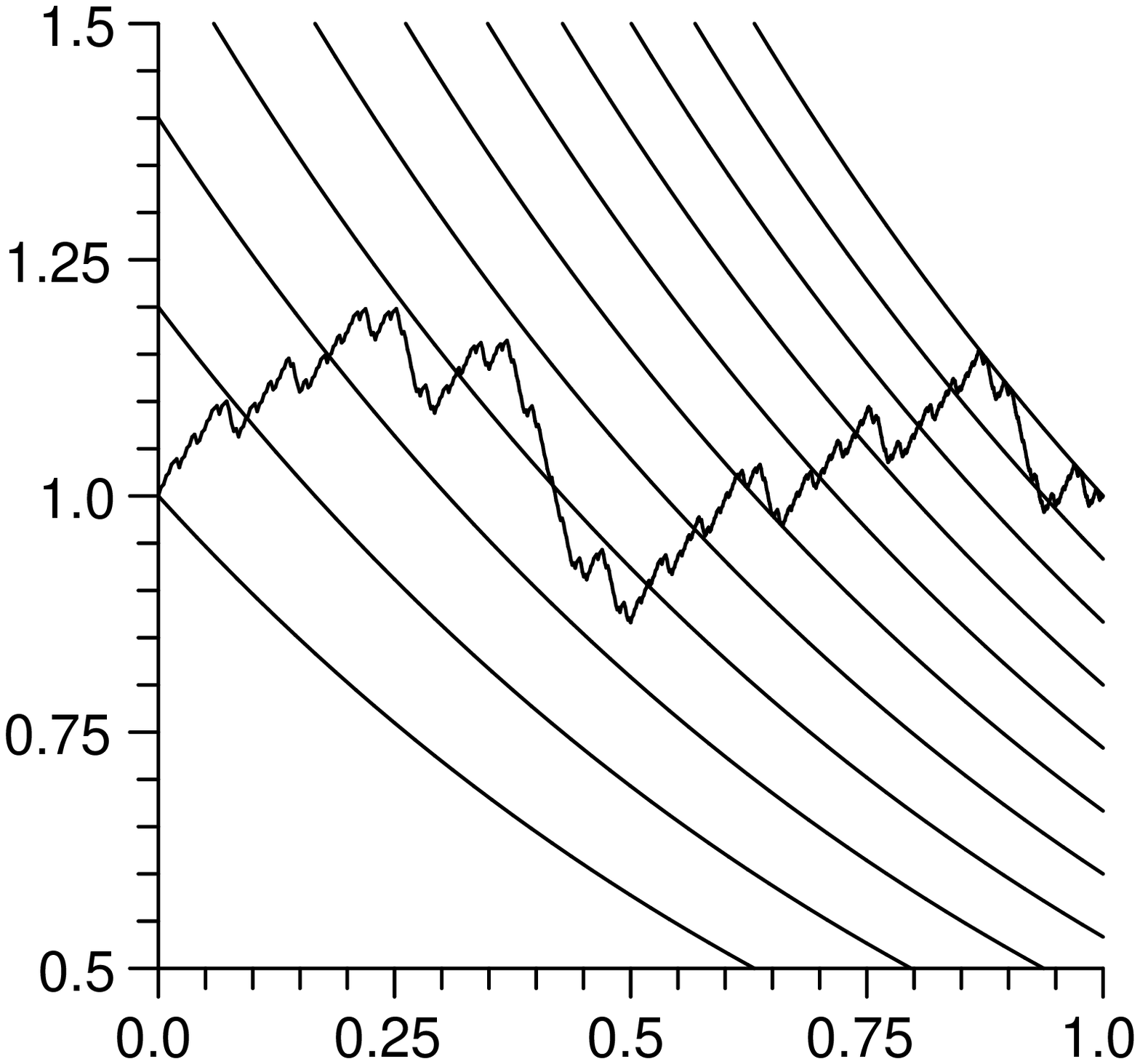}
\caption{\label{Du08vers09:fi:Coquet}
The change from the function~$F(x)$ (left) to the
periodic function~$\Phi(t)$ (right) illustrated for the Coquet
sequence.}
\end{figure}
It is 
$4$-rational and  it admits (with $(u(n),u(4n+2),u(4n+3))$ as a
generating family) the linear representation
\begin{multline*}
A_0=\left(\begin{array}{ccc}
1 & 1 & 1 \\ 0 & 0 & 0 \\  0 & 0 & 0 
	  \end{array}\right),\qquad
A_1=\left(\begin{array}{ccc}
1 & 0 & 0 \\ 0 & 1 & -1 \\  0 & 0 & 0 
	  \end{array}\right),\qquad
C=\left(\begin{array}{c} 1 \\ 0 \\ 0 \end{array}\right)
\\
L=\left(\begin{array}{ccc} 1 & 1 & 1\end{array}\right),\qquad
A_2=\left(\begin{array}{ccc}
 0 & 0 & 0 \\ 1 & 1 & 0 \\ 0 & 0 & 1
	  \end{array}\right),\qquad
A_3=\left(\begin{array}{ccc}
 0 & 0 & 0 \\  0 & 0 & 0 \\ 1 & -1 & 1
	  \end{array}\right).
\end{multline*}
The matrix
\[
Q=A_0+A_1+A_2+A_3=
\left(\begin{array}{ccc} 2 & 1 & 1 \\ 1 & 2 & -1 \\ 1 & -1 & 2\end{array}\right)
\]
has as eigenvalue~ $3$, which is double, and~$0$, which is simple. The
joint spectral radius of the family $(A_r)_{0\leq r<4}$ is
$\lambdajsr=1$, because with the maximum sum column norm we find
$\lambda_1^{(1)}=1$ and ~$1$ is an eigenvalue of each matrix~$A_r$,
$0\leq r<4$. The column vector~$C$ decomposes as $C=V_3+V_0$, with
\[
V_3=\left(\begin{array}{c}2/3 \\ 1/3 \\
1/3\end{array}\right),\qquad
V_0=\left(\begin{array}{c}1/3\\-1/3\\-1/3\end{array}\right)
\]
and~$V_{\rho}$ is an eigenvector for~$\rho=3$, $0$. 

Theorem~\ref{Du08vers09:thm:sequenceasymptoticexpansion} applies and
we obtain~\citep{Coquet83,FlGrKiPrTi94}
\[
\sum_{n\leq N}(-1)^{s_2(3n)}\mathop{=}_{N\to+\infty}
N^{\log_43}\,3^{1-\{t\}}F(4^{\{t\}-1})+O(1),
\]
with $t=\log_4 N$. The function~$F$ is the sum $F=F_1+F_2+F_3$,
where $\boldF=(F_1,F_2,F_3)$ is the unique solution of the dilation
equation, written with Convention~\ref{Du08vers09:hypo:homogequconv},
\[
\left\{
\begin{array}{l}
 F_1(x)=\displaystyle\frac{1}{3}F_1(4x)+\frac{1}{3}F_2(4x)+\frac{1}{3}F_3(4x)+\frac{1}{3}F_1(4x-1),\\[1.5ex]
F_2(x)=\displaystyle\frac{1}{3}F_2(4x-1)-\frac{1}{3}F_3(4x-1)+\frac{1}{3}F_1(4x-2)+\frac{1}{3}F_2(4x-2),\\[1.5ex]
F_3(x)=\displaystyle\frac{1}{3}F_3(4x-2)+\frac{1}{3}F_1(4x-3)-\frac{1}{3}F_2(4x-3)+\frac{1}{3}F_3(4x-3),
\end{array}\right.
\]
with the conditions $F_1(0)=F_2(0)=F_3(0)=0$, $F_1(1)=2/3$,
$F_2(1)=F_3(1)=1/3$. The function~$F(x)$ and the periodic function
$\Phi(t)=3^{1-\{t\}}F(4^{\{t\}-1})$ are illustrated in
Fig.~\ref{Du08vers09:fi:Coquet} respectively on the left and on the
right. The positive character of~$\Phi$ proves a subtle
phenomenon~\citep{Newman69}: the number of ones in the binary
expansion of the integers which are multiple of~$3$ is more often
even than odd.  Both functions are H\"older with exponent
$\log_43\simeq 0.795$. In the translation from~$F$ to~$\Phi$, the
first quarter of~$F$ is lost and~$\Phi$ has not the auto-similar
character of~$F$.
\end{example}

As we have seen if the previous sections, it may be interesting to
change the radix. Going from radix~$\base$ to radix~$\base^T$, we
change the matrix~$Q$ into its power~$Q^T$. Frequently the eigenvalues
of~$Q$ have arguments which are commensurable with~$\pi$ and a
suitable choice of~$T$ gives a matrix~$Q$ whose all eigenvalues are
real nonnegative. This trick is often used implicitly. For example
the Coquet sequence is viewed as a $4$-rational sequence (as we have
made in the previous example) in \citep{Coquet83,DuTh89,AlSh03} (but
this is not the case in~\citep{FlGrKiPrTi94}). However it may be of
interest to use two radices at a time. Even if both dilation equations
define as well the function under consideration, it may be practically
simpler to study the solution for one radix in terms of solutions for
the other radix. The next example illustrates this point.

\begin{example}[Rudin-Shapiro sequence continued]
  As an application of
  Theorem~\ref{Du08vers09:thm:sequenceasymptoticexpansion}, we obtain the
  asymptotic expansion for the Rudin-Shapiro sequence defined in
  Ex.~\ref{Du08vers09:ex:RudinShapiro} \citep{BrErMo83}
\[
\sum_{n\leq N}u_n\mathop{=}_{N\to+\infty}\sqrt{N}\Phi(\log_4N)+O(1)
\]
where~$\Phi$ is the $1$-periodic function defined by
$\Phi(t)=2^{1-\{t\}}F(4^{\{t\}-1})$ and~$F$ is defined through a
dilation equation for radix~$4$.  Functions~$F$ and~$\Phi$ are
illustrated in Fig.~\ref{Du08vers09:fig:RudinShapiro_symmetry}. Let us
denote $\pi_k$ the part of the graph of~$F$ which corresponds to the
interval $[k/8,(k+1)/8]$ for $0\leq k<8$. It is evident that the parts
of odd index on one side and the part of even index on the other side
reproduce the same pattern (with a piece upside down) and we want to
prove these facts.

\begin{figure}
\begin{center}
\includegraphics[width=0.45\linewidth]{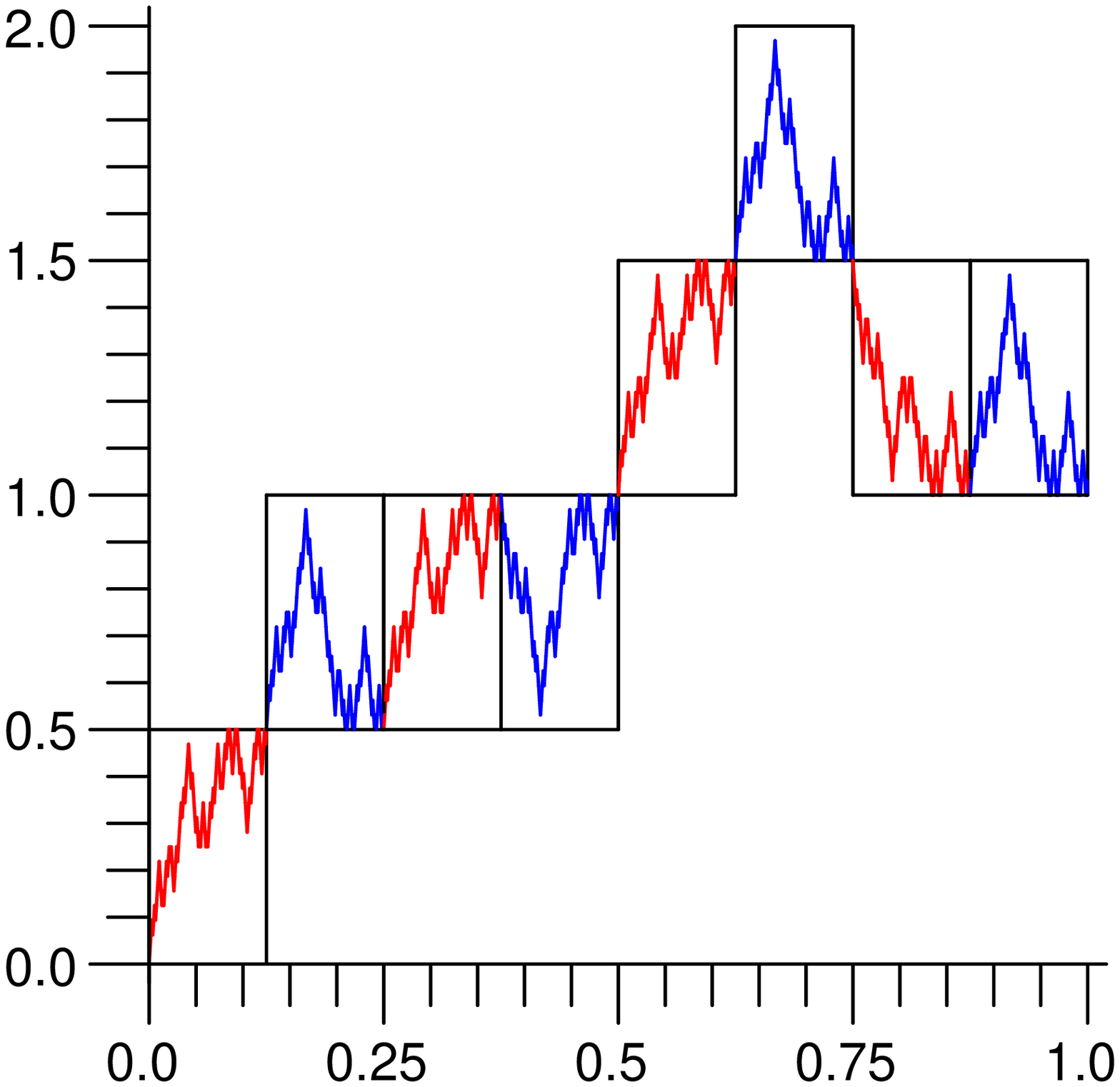}
\hfil
\includegraphics[width=0.45\linewidth]{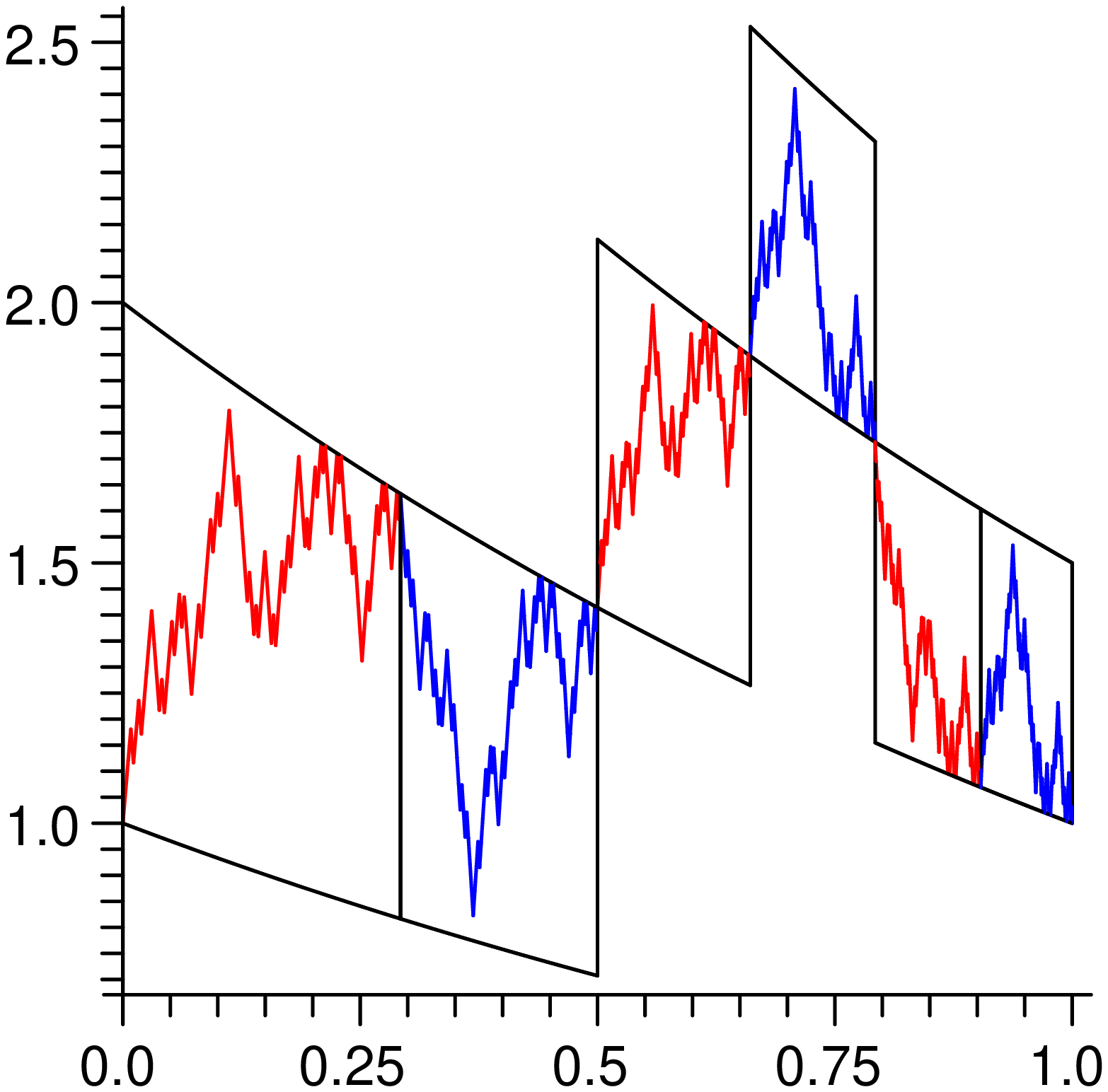}
\end{center} 
\caption{\label{Du08vers09:fig:RudinShapiro_symmetry} The solution of
  the dilation equation (left-hand side) and the periodic function
  (right-hand side) associated to the Rudin-Shapiro sequence show some
  symmetries.  }
\end{figure} 

Besides of radix~$4$, we will use radix~$2$. The matrix~$Q$ for
radix~$2$ (see Ex.~\ref{Du08vers09:ex:RudinShapiro}) has two
eigenvalues $\pm\sqrt{2}$ and associated eigenvectors
$V_{\pm}=(2\pm\sqrt{2},\pm\sqrt{2})/4$. The vector~$C$ of the
representation writes $C=V_++V_-$ and we have to consider two
vector-valued functions $\boldF_{\pm}$. Their components
are~$F_{\pm,1}$ and~$F_{\pm,2}$. The function under study is
$F=F_{+,1}+F_{+,2}+F_{-,1}+F_{-,2}$.

To abbreviate the computations and with the hope to make them clearer,
we introduce the following formalism. The group of affine transforms
of the real line acts on the vector space of functions from the real
line into itself. More precisely the group acts on the right by
substitution: for an affine transform~$a$ and a function~$f$, the
image $f.a$ is simply the composed map $f\circ a$. As affine transforms
it is natural in the actual context to consider $d_0:x\mapsto 2x$ and
$d_1:x\mapsto 2x-1$. We will use $t_{1/2}:x\mapsto x+1/2$ too. With
these notations the dilation equation for~$\boldF_\pm$ writes (We use
Convention~\ref{Du08vers09:hypo:homogequconv} to write the dilation
equations, that is we consider that functions are adequately extended to
the whole real line.)
\begin{equation}\label{Du08vers09:eq:abstractdileq}
F_{\pm,1}=\pm\frac{1}{\sqrt 2}F_{\pm,1}.d_0\pm\frac{1}{\sqrt
  2}F_{\pm,2}.d_0,\qquad
F_{\pm,2}=\pm\frac{1}{\sqrt 2}F_{+,1}.d_1\mp\frac{1}{\sqrt
  2}F_{\pm,2}.d_1,
\end{equation}
with the conditions $\boldF_\pm(0)=0$ and $\boldF_\pm(1)=V_\pm$ We
will collect some simple facts in order to achieve our goal. It
results from the dilation equations that $F_{\pm,2}$ is~$0$ on the
left of~$1/2$, while $F_{\pm,1}$ is constant on the right
of~$1/2$. The action of~$t_{1/2}$ on both members of the second
equation above with the equality $d_1\circ t_{1/2}=d_0$ provides
\[
F_{+,2}.t_{1/2}=\frac{1}{\sqrt 2}F_{+,1}.d_0-\frac{1}{\sqrt
  2}F_{+,2}.d_0.
\]
A subtraction and an addition give respectively
\begin{equation}\label{Du08vers09:eq:Fsym}
F_{+,1}-F_{+,2}.t_{1/2}=\sqrt{2}F_{+,2}.d_0,\qquad
F_{+,1}+F_{+,2}.t_{1/2}=\sqrt{2}F_{+,1}.d_0.
\end{equation}
From these preliminaries follow the following formul\ae
\begin{equation}\label{Du08vers09:eq:formulaforRudinShapiro}
\begin{array}{llll}
F\left(x+\frac{1}{2}\right)-F(x)=1,&  \text{for $0\leq x\leq
  \frac{1}{4}$},
&\qquad
F\left(x+\frac{1}{2}\right)+F(x)=2,&  \text{for $\frac{1}{4}\leq x\leq \frac{1}{2}$},
\\[2ex]
F\left(x+\frac{1}{4}\right)-F(x)=\frac{1}{2},&  \text{for
$\displaystyle 0\leq x\leq \frac{1}{8}$},
&\qquad
F\left(x+\frac{1}{4}\right)+F(x)=\frac{3}{2},&  \text{for
$\displaystyle\frac{1}{8}\leq x\leq \frac{1}{4}$}.
\end{array}
\end{equation}
We prove only the third one; the other proof are of the same type.
For $0\leq x\leq 1/8$, the number $x+1/4$ lies in $[1/4,3/8]$ and both
numbers are in $[0,1/2]$. The dilation equations~\eqref{Du08vers09:eq:abstractdileq} give
\[
F_{\pm,1}(x)=\pm\frac{1}{\sqrt 2}F_{\pm,1}(2x)\pm\frac{1}{\sqrt 2}F_{\pm,2}(2x),\qquad
F_{\pm,2}(x)=0,
\]
and 
\[
F_{\pm,1}\left(x+\frac{1}{4}\right)=\pm\frac{1}{\sqrt
  2}F_{\pm,1}\left((2x+\frac{1}{2}\right)\pm\frac{1}{\sqrt
  2}F_{\pm,2}\left(2x+\frac{1}{2}\right),\quad\displaystyle
F_{\pm,2}\left(x+\frac{1}{4}\right)=0.
\]
Since $F$ is $F_{+,1}+F_{+,2}+F_{-,1}+F_{-,2}$, we obtain 
\begin{multline*}
F\left(x+\frac{1}{4}\right)-F(x)=\frac{1}{\sqrt 2}\left[
F_{+,1}\left(2x+\frac{1}{2}\right)+F_{+,2}\left(2x+\frac{1}{2}\right)
\right.\\\left.
-
F_{-,1}\left(2x+\frac{1}{2}\right)-F_{-,2}\left(2x+\frac{1}{2}\right)
-F_{+,1}(2x)-F_{+,2}(2x)
+F_{-,1}(2x)+F_{-,2}(2x)
\right].
\end{multline*}
However $2x$ is in $[0,1/4]$ and $2x+1/2$ is in  $[1/2,3/4]$. The
functions $F_{\pm,1}$ are constant on the right of $1/2$, while the
functions $F_{\pm,2}$ are constant on the left of $1/2$. Making
explicit theses constant values, we obtain 
\[
F\left(x+\frac{1}{4}\right)-F(x)=\frac{1}{\sqrt 2}\left[
\frac{1}{2}\sqrt{2}
-F_{+,1}(2x)
+F_{-,1}(2x) 
+F_{+,2}\left(2x+\frac{1}{2}\right)
-F_{-,2}\left(2x+\frac{1}{2}\right)
\right].
\]
Besides Formula~\eqref{Du08vers09:eq:Fsym} gives 
\[
F_{\pm,1}(2x)=F_{\pm,2}\left(2x+\frac{1}{2}\right)
\]
because $x$ is in $[0,1/8]$. The expected formula is proved.

The previous computations ask some questions. We have established
formul\ae\, which write
\[
\sum_{f,a} \lambda_{f,a} f.a=c,
\]
where $f$ runs through the four functions $F_{+,1}$, $F_{+,2}$,
$F_{-,1}$, $F_{-,2}$; $a$ runs through the set of affine transforms of
the real line; the family $(\lambda_{f,a})$ has a finite support; $c$
is some constant. Is it possible to design a method in order to find
such relationships? For example, how to obtain mechanically the
formula
\[
\frac{4}{2+\sqrt 2}F_{\pm,1}(x)-\frac{4}{\sqrt 2}F_{\mp,2}(1-x)=1?
\]
Does exists a regular process which gives formul\ae\, of the type
\[
\sum_a\lambda_a F.a=C?
\]
Evidently we want a minimal generating set for such formul\ae. 

From the four formul\ae~\eqref{Du08vers09:eq:formulaforRudinShapiro}
above we deduce that all pieces of the graph of~$F$ are obtained
from~$\pi_0$ and~$\pi_1$ by some translations or glide
reflections. The symmetries of~$F$ are translated to~$\Phi$, but they
lose their graphical evidence. The pieces~$\pi_0$ and~$\pi_1$
disappear and the pieces~$\pi_k$, $2\leq k<8$ become the
pieces~$\pi_k'$ associated to the interval
$[\log_4(k/2),\log_4((k+1)/2)]$. The links between the pieces become
more intricate. For example the pieces $\pi_2'$ and $\pi_3'$ on one
side and $\pi_6'$ and $\pi_7'$ on the other side are linked by the
formula $2^s\Phi(s)+2^t\Phi(t)=4$ under the condition that both
numbers $s\in[0,1/2]$ and $t\in[\log_43,1]$ are related by
$4^t-4^s=2$.
\end{example}

\begin{example}[Rescaling]
  In the study of radix-rational sequences, it is an attractive idea
  at first sight to multiply all the matrices~$A_0$, $\ldots$,
  $A_{\base -1}$ of a linear representation by a scalar in order
  to control the order of growth of the sequence. But after a while
  the idea seems to be silly, because it introduces a weighting
  according to the length of the radix~$\base $ expansions of the
  integers. The previous theorem shows clearly the effect of such a
  rescaling. It appears that the modification does not change at all
  the functions~$F(x)$ defined by a dilation equation, and changes very
  simply the periodic functions~$\Phi(t)$. 

  Let us consider the sequence $u(n)=n$. It admits the representation
\[
L=\left(\begin{array}{cc}0&1\end{array}\right),\quad
A_0=\left(\begin{array}{cc}{\base}&0\\0&1\end{array}\right),\quad
A_r=\left(\begin{array}{cc}{\base}&0\\r&1\end{array}\right)\;\text{for
$r>0$,}\quad C=\left(\begin{array}{c}1\\0\end{array}\right).
\]
for radix~$\base$. There is a dominant eigenvalue $\rho={\base}^2$,
with associated functions $F(x)=x^2$ and $\Phi(t)=1/2$. We obtain the
obvious asymptotic expansion
\[
\sum_{n\leq N}n\mathop{=}_{N\to+\infty}\frac{1}{2}N^2+O(N).
\]
Now consider the sequence $u(n)=n/{\base}^{\lambda_{\base}(n)}$, where
$\lambda_{\base}(n)$ is the length of the radix~${\base}$ expansion
of~$n$. With ${\base}=2$ it begins with $0$, $0.1$, $0.10$, $0.11$,
$0.100$, $0.101$, $0.110$, $0.111$, $0.1000$, $\ldots$ It admits the
representation
\[
L=\left(\begin{array}{cc}0&1\end{array}\right),\quad
A_0=\left(\begin{array}{cc}1&0\\0&1/{\base}\end{array}\right),\quad
A_r=\left(\begin{array}{cc}1&0\\r/{\base}&1/{\base}\end{array}\right)\;\text{for $r>0$,}\quad
C=\left(\begin{array}{c}1\\0\end{array}\right)
\]
and again we have $F(x)=x^2$, but $\rho={\base}$. We obtain the less
obvious asymptotic expansion
\[
\sum_{n\leq N}\frac{n}{\base^{\lambda_\base(n)}}\mathop{=}_{N\to+\infty}\frac{1}{2}N\Phi(\log_{\base}N)+O(1)
\]
with
\[
\Phi(t)=\frac{{\base}^{-t}+{\base}^{t-1}}{2}=\frac{1}{{\base}^{1/2}}\cosh
\left((t-\frac{1}{2})\ln {\base}\right),\qquad\text{for $0\leq t<1$.}
\]
The convergence towards this function is illustrated in
Fig.~\ref{Du08vers09:fig:rescaling} where a (piece of) catenary
appears in an evident way.
\begin{figure}
  \centering
\includegraphics[width=0.8\linewidth]{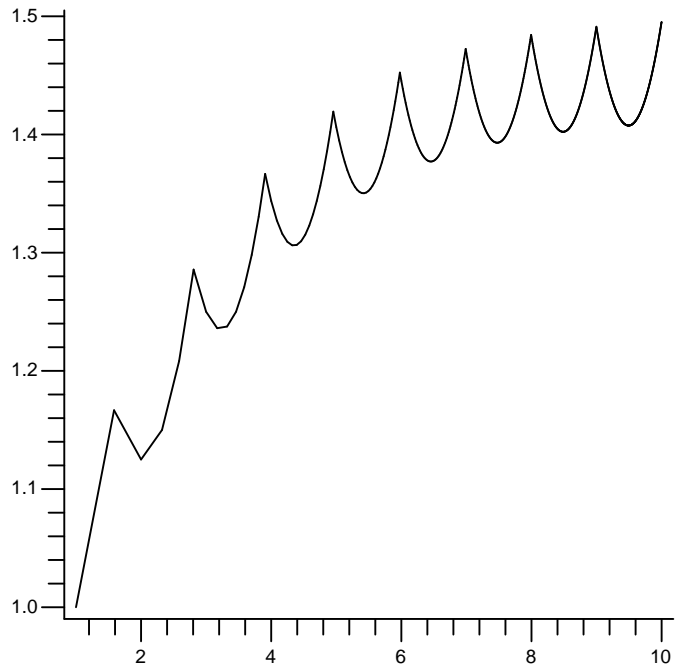}
\caption{\label{Du08vers09:fig:rescaling}The empirical periodic
  function for sequence $u(n)=n/2^{\lambda_2(n)}$
  (where $\lambda_{2}(n)=1+\lfloor \log_2 n\rfloor$ is the length of the
  binary expansion of~$n$) as a function of~$\log_2n$. A catenary
  appears.}
\end{figure}
\end{example}

\subsection{Substitutions and automata}
In the very beginning of the introduction we have written that nobody
has asserted a general theorem about the asymptotic expansion of
radix-rational sequences. This is not quite the truth because
\citet[Sec.~3.5]{AlSh03} provide a theorem. But their statement
suffers from a severe restriction because they assume for example the
hypothesis $A_w=O(K)$ for $\length w=K$, that is $\lambdajsr\leq 1$,
and some other technical hypotheses. It is noteworthy that the result
they give has no error term and for this it appears as a direct
generalization of~\citep{Delange75}.

Other works which deserve attention are related to the name of Dumont 
\citep{DuTh89,Dumont90,DuSiTh99}.  There is a link between these works
and the present one, but it needs some explanations to be
emphasized. We consider a finite alphabet~$\cX$. A
substitution~$\sigma$ on~$\cX$ is a map which associates to each
letter a non empty word. It extends as a morphism of~$\cX^*$ by
concatenation. If~$a$ is a letter which turns out to be a strict
prefix of the word~$\sigma(a)$, the sequence~$(\sigma^k(a))$, obtained
by iteration of~$\sigma$ from~$a$, converges to an infinite
word~$\omega=(\omega_n)_{n\geq0}$, which is named a fixed
point of~$\sigma$~\citep{ArBeFeItMaMoPeSiTaWe01}.

We may associate to the substitution~$\sigma$ and to the letter~$a$ an
automaton~$\cA_{\sigma}$. The state of the automaton are the letters
of the alphabet. There is a unique initial state which is the
letter~$a$. All states are final states. The transitions are labeled
by nonnegative integers and obtained in the following way. For a
letter~$x$, the non empty word~$\sigma(x)$ has length~$\ell\geq 1$ and
writes
$\sigma(x)=\sigma_0(x)\sigma_1(x)\ldots\sigma_{\ell-1}(x)$. There is a
transition from state~$x$ to state~$\sigma_r(x)$ labeled
by~$r$. Conversely the automaton determines the substitution and the
fixed point of the substitution which begins by~$a$. If~$\base$ is the
maximal length of the words $\sigma(x)$ for~$x$ a letter, all labels
are bounded above by $\base-1$. To the automaton is associated the
language~$\cL$ it recognizes, that is the set of words over the
alphabet $\left\{0,1,\ldots,\base-1\right\}$ which translate a
sequence of transitions from the initial state to a final state. This
language is regular because it is recognized by the automaton. The
simplest case is the case where the substitution is of constant
length, which means that all words $\sigma(x)$, $x\in\cX$, have the
same length $\base\geq 2$. In such a case the language~$\cL$ is merely
the free monoid $\left\{0,1,\ldots,\base-1\right\}^*$.

The automaton may be enriched by an output function~$f$ which
associates to each state, that is letter, a value. The output function
extends to words and the image of the fixed point~$\omega$ is an
infinite sequence of values. \citet{DuTh89} study the asymptotic
behaviour of the running sum
\[
S^f(N)=\sum_{0\leq n<N}f(\omega_n)
\]
where~$f$ is real-valued, under some hypotheses which we will
describe later.

To merge the framework of \citep{DuTh89} and ours it suffices to
consider a linear version of the automaton. For a given commutative
field~$\bK$, we consider the space $\bK^{\cX}$, whose canonical base
is $(e_x)_{x\in\cX}$. To each integer~$r$ between~$0$ and $\base-1$ we
associate a $d\times d$ matrix~$A_r$, where~$d$ is the size of the
alphabet, defined by $A_re_x=e_y$ if there is a transition
labeled~$r$ from state~$x$ to state~$y$ in the automaton, and
$A_re_x=0$ otherwise. The matrix $Q=A_0+\dotsb+A_{\base-1}$ is the
incidence matrix of the substitution; its entries $Q_{x,y}$ is the
number of occurrences of the letter~$y$ in the word $\sigma(x)$. We
insert in the description the column vector $C=e_a$ and the row
vector $L=(f(x))_{x\in\cA}$ and we have a linear representation. For a
word $w=w_0\dotsb w_{K-1}$ taken from the language $\cL$ of the
automaton, it is equivalent to follow from the initial state~$a$ the
transitions labeled $w_0$, $\ldots$, $w_{K-1}$ and to compute the
value $f(z)$ of the state~$z$ where ends the path, or to compute the
matrix product $LA_{w_{K-1}}\dotsb A_{w_0}C$. We are not far to the
definition of a radix-rational sequence, but the idea of a numeration
system is yet lacking.

The lengths of the iterates~$\sigma^n(a)$ are the elements of an
increasing sequence~$\cB=(\base_n)_{n\geq 0}$ which begins with
$\base_0=1$. This sequence defines a numeration system: a nonnegative
integer~$n$ admits an expansion $(a_{K-1}\ldots a_0)_{\cB}$ in the
system~$\cB$ if it writes
$n=a_0\base_0+a_1\base_1+\dotsb+a_{K-1}\base_{K-1}$. There is
always at least one expansion named the normal $\cB$-representation,
which is obtained by a greedy algorithm. The language~$\cS$ associated
to the system of numeration is the set of all normal
$\cB$-representations of nonnegative integers. (The representation
of~$0$ is the empty word.) Generally speaking, the language~$\cL$ is
not regular and is distinct form the language~$\cS$, but in the simple
case where the substitution is of constant length~$\base$, both
languages~$\cL$ and~$\cS$ are equal. As a consequence the sequence
$(f(\omega_n))$ is $\base$-rational. It is even $\base$-automatic: for
an nonnegative integer~$n$ we write its radix~$\base$ expansion and we
use this word to follow a path which begins with~$a$ in the automaton;
the end of the path is~$\omega_n$ and the value associated to~$n$
is~$f(\omega_n)$. It is proved that a $\base$-automatic sequence is a
$\base$-rational sequence which takes a finite number of values.

Nevertheless the previous system of numeration is not suitable in case
of a substitution which is not of constant length. The prefix-suffix
automaton~$\cA'_{\sigma}$ \citep{Mosse96,CaSi01a} is the right
automaton. Its states are the letters of~$\cX$. There is a transition
labeled $(p,x,s)\in\cX^*\times\cX\times\cX^*$ from~$x$ to~$y$ if
$\sigma(y)=pxs$ (hence the name of the automaton). All states are
initial and~$a$ is the only final states. Clearly~$\cA'_{\sigma}$ is a
sophisticated version of~$\cA$ (the transitions are reversed but this
is of no importance), which is equivalent in case of a substitution of
constant length. \citet{DuTh89} use a version of this automaton which
takes into account only the prefixes. The language~$\cL'$ recognized
by~$\cA'_{\sigma}$ may be defined by a set of prohibited patterns: a
pair of labels $((p,x,s),(p',y,s'))$ such that $\sigma(y)\neq psx$
does not occur in a word of~$\cL'$.  For every subword
$\omega^N=(\omega_n)_{0\leq n<N}$ there exists a unique path
$((p_k,x_k,s_k))_{0\leq k<K}$ of length~$K$, determined by
$\base_K\leq N<\base_{K+1}$, which ends at~$a$. In other terms there
exists a unique word $((p_k,x_k,s_k))_{0\leq k<K}$ from~$\cL'$ of
length~$K$ such that $x_{K-1}=a$. It defines a cutting of
$\omega^N=(\omega_n)_{0\leq n<N}$ which writes
$\omega^N=\sigma^{K-1}(p_{K-1})\sigma^{K-2}(p_{K-2})\ldots\sigma^0(p_0)$. Moreover
the prefixes~$p_k$ lie in the finite set of strict prefixes of the
words~$\sigma(x)$ with $x$ a letter. Such a result demands hypotheses
\citep{Mosse92}: the substitution is primitive, which means that there
exists a~$k$ such that for every pair of letters~$x$ and~$y$ the
letter~$x$ occurs in $\sigma^k(y)$; the fixed point~$\omega$ is not a
periodic sequence. Taking into account the lengths of the words, we
obtain
$N=|\sigma^{K-1}(p_{K-1})|+|\sigma^{K-2}(p_{K-2})|+\dotsb+|\sigma^0(p_0)|$. This
is the expansion of the integer~$N$ according to the numeration system
associated with the substitution~$\sigma$ and~$a$. (Note that there
does not exist a general definition of what is a numeration system.) 
The proof of the existence of the expansion is anew a greedy
algorithm~\citep[Lemma~1.3]{DuTh89}: $p_{K-1}$ is the longest prefix
of~$\omega$ such that $\sigma^{K-1}(p_{K-1})$ is a prefix
of~$\omega^N$.

The asymptotic study of the sequence~$(S^f(N))_{N\geq 0}$ is based on
the following remarks. The function~$f$ is defined on the letters of
the alphabet~$\cX$ and extends additively to the word
of~$\cX^*$. Particularly we have $S^f(N)=f(\omega^N)$. Let~$L$ the row
vector whose entries are the values of~$f$ on the letters of the
alphabet; let $C_w$ the column whose entries are the number of
occurrences of each letter in the word~$w$. A consequence of these
definition is the formula $f(w)=LC_w$. Let~$Q$ be the incidence matrix of
the substitution~$\sigma$. We see immediately the recursion
$C_{\sigma(w)}=QC_w$. As a consequence we have two analogous formul\ae\, 
\[
S^f(N)=f(\omega^N)=L\sum_{k=0}^{K-1}Q^kC_{p_k},\qquad
N=U\sum_{k=0}^{K-1}Q^kC_{p_k},
\]
where~$U$ is the row vector whose all entries are equal to~$1$. Both
formul\ae\, render plausible the following result
\citep[Th.~2.6]{DuTh89},
\[
S^f(N)\mathop{=}_{N\to+\infty}
L\Lambda \,N
+\log_{\theta}^{\alpha}NN^{\beta}G(N)+o(\log_{\theta}^{\alpha}NN^{\beta}).
\]
It is proved under the following hypotheses and through the following
assertions. Because~$Q$ is a primitive nonnegative matrix, its
spectral radius~$\theta$ is a dominant eigenvalue. The column
vector~$\Lambda$ is its positive eigenvector, normalized by
$U\Lambda=1$. The existence of a sub-dominant eigenvalue $\theta'$ is
assumed. It dominates all others eigenvalues apart~$\theta$, that is
$\theta>\theta'>|\theta''|$ for all eigenvalues~$\theta''$ except
$\theta$ and $\theta'$. Moreover it is assumed $\theta'>1$. The
integer~$\alpha$ is defined such that $\alpha+1$ is the multiplicity
of~$\theta'$ as root of the minimal polynomial of~$Q$ and the real
number~$\beta$ is defined by $\beta=\log_{\theta}\theta'$. The
function $G$ is defined as the limit
\[
\forall x> 0,\quad G(x)=x^{-\beta}\lim_{k\to+\infty}
\frac{S^f(\lfloor\theta^kx\rfloor)-L\Lambda \lfloor\theta^kx\rfloor}{k^{\alpha}\theta'{}^k}.
\]\
It satisfies $G(\theta x)=G(x)$ for~$x$ positive and is H\"older of
exponent~$\beta$.

To summarize \citet{DuTh89} study the mean asymptotic behaviour of
real-valued sequences defined by substitutions and for substitutions
with constant length this reduces to automatic sequences. We study
complex-valued sequences associated to complex rational formal power
series and for series which takes only a finite number of values this
reduces to automatic sequences.

\subsection{Linear representation insensitive to the leftmost zeroes}
In the use of Formula~\eqref{Du08vers09:eq:runningsumu}, we may
introduce a simplifying hypothesis which turns out to be non
restrictive.  To a rational formal series we associate a
radix-rational series through a numeration system. But the knowledge
of a radix-rational sequence does not determine a unique rational
formal series, because the expansions of the integers does not use the
words which begin with a zero. A first idea which comes in mind to
complete the definition of the formal series is to decide that words
which begin with zero give a null value. But there is another way to
extend a radix-rational sequence into a formal series and it is more
natural after all. Let us say that a linear representation of a
radix-rational series is {\em insensitive to the leftmost zeroes\/} if
it satisfies $LA_0=L$. For such a representation the formal power
series is completely determined by the radix-rational series, because
the adding of some zeroes on the left of the expansion of an integer
does not change the value associated to this word.
\begin{lemma}
  Every radix-rational sequence has a (reduced) linear insensitive to
  the leftmost zeroes representation, that is a representation such
  that $LA_0=L$.
\end{lemma}
This point is described in~\citep[sec.~4.2]{Dumas93a} (where an
insensitive to the leftmost zeroes representation is termed standard),
but we give a proof because it explains how most of the linear
representations in this paper are obtained.
\begin{proof}
  For the sake of simplicity, let us assume that the radix is
  $\base=2$. The hypothesis is that all subsequences $(u_{2^kn+r})$,
  $k\geq 0$, $0\leq r<2^k$, of the radix-rational sequence~$(u_n)$
  generate a finite dimensional vector space~$U$. Let~$d$ be the
  dimension of that vector space. (For the null sequence, we have
  $d=0$. We exclude this case.) We consider the subsequences $(u_n)$,
  $(u_{2n})$, $(u_{2n+1})$, $(u_{4n})$, $(u_{4n+1})$, $(u_{4n+2})$,
  $\ldots$ in that order and the dimension~$d_j$ of the vector space
  generated by the~$j$ first subsequences (with $d_{-1}=0$). We select
  the subsequences $v^1$, $v^2$, $\ldots$, $v^d$ such that the
  dimension increases by~$1$ when they encountered in the previous
  list. The family $(v^{\ell})_{1\leq \ell\leq d}$ is a basis
  of~$U$. The sequence~$(u_n)$ expresses as a linear combination of
  $v^1$, $v^2$, $\ldots$, $v^d$ (the sequence $(u_n)$ is nothing
  but~$v^1$), and this gives the column vector~$C$ of the linear
  representation we are building. Next we consider the action of~$0$
  and~$1$ over the sequences $v^{\ell}$ defined by
  $(0.v^{\ell})_n=v^{\ell}_{2n}$ and $(1.v^{\ell})_n=v^{\ell}_{2n+1}$
  and we express the images in the basis $(v^{\ell})_{1\leq \ell\leq
  d}$. This gives the square matrices~$A_0$ and~$A_1$ of the
  representation. At last the row vector~$L$ is the vector of initial
  values~$v^{\ell}_0$. In that way, we have a linear representation of
  the sequence~$(u_n)$. This may be readily verified by considering
  the sequence of row vectors $\Lambda_n=(v^1_n,\, v^2_n,\ \ldots\
  v^d_n)$ which satisfies $u_n=\Lambda_n C$, $\Lambda_{2n}=\Lambda_n
  A_0$, $\Lambda_{2n+1}=\Lambda_nA_1$, $L=\Lambda_0$. (Note that
  \citet{AlSh03} use a column vector for~$\Lambda_n$, hence the
  transposed matrices for~$A_0$, $A_1$.)  Moreover the equation
  $(0.v^{\ell})_0=v^{\ell}_{0}$, that is $\Lambda_0A_0=\Lambda_0$,
  renders evident the formula $LA_0=L$.  The proof would be more
  enlightening by considering a binary tree and a prefix part of the
  monoid of words, but we will not insist on this point.
\end{proof}
We do not have used this hypothesis of insensitivity in the
proofs because it is unsuitable for the vector-valued version we have
elaborated, but most of the examples we have given have this property
and this saves us to compute the first term of
Formula~\eqref{Du08vers09:eq:runningsumu}.

\subsection{Periodicity versus pseudo-periodicity}
The occurrence of periodic functions in logarithmic scale is common in
the asymptotic study of radix-rational
sequences~\citep{Coquet83,Delange75,DuSiTh99,FlGo94,FlGrKiPrTi94,OsSh89}. Nevertheless
the rising of periodic functions, if true in practical examples, is
wrong in full generality. The reason is the following: in common
examples the eigenvalues of the matrix~$Q$ have arguments which are
commensurable with~$\pi$. As we have already explained, a change of
radix permits to consider that the useful eigenvalues are all positive
numbers. As a consequence functions $\Phi$ are periodic, because
they depend only on the fractional part of the variable~$t$. But in
the general case they write something like $\Phi(t)=\omega^{\lfloor
t\rfloor}\rho^{1-\{t\}} F(\base^{\{t\}-1})$ and if~$\omega$ does not
write $\omega=e^{i\vartheta}$ with $\vartheta/\pi$ a rational number,
the function $\Phi$ is not periodic but only pseudo-periodic. Below we
exemplify this phenomenon.

\begin{figure}
\begin{center}
\includegraphics[width=0.5\linewidth]{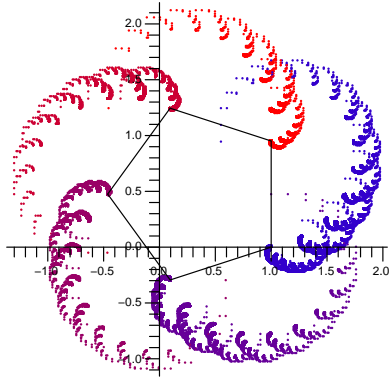}
\end{center}
\caption{\label{Du08vers09:fig:rosette5}The empirical rosette of
  Ex.~\ref{Du08vers09:ex:rosette} with $\vartheta=2\pi/5$. For each
  frond of the rosette, the index~$N$ goes from $2^K$
  to~$2^{K+1}-1$. Here $K$ goes from~$11$ to~$15$ and the color from
  red to blue. The vertices of the pentagon are the
  points~$\boldGamma(K)$ with~$K$ integer.}
\end{figure}

\begin{example}[Rosettes]\label{Du08vers09:ex:rosette}
Let us consider the linear representation for radix $\base=2$ and
dimension~$2$,
\[
A_0=\left(\begin{array}{cc}
\cos{\vartheta} & 0 \\ 0 & \cos{\vartheta}
	\end{array}\right),\qquad
A_1=\left(\begin{array}{cc}
0 & -\sin{\vartheta} \\ \sin{\vartheta} & 0
	\end{array}\right),
\]
where~$\vartheta$ is a real number restricted by the condition
$\vartheta\not\equiv0\bmod \pi/2$ to avoid degeneracy. We do not
define the row vector~$L$, because we are interested in the
vector-valued sequence
\[
\boldSigma_N=\sum_{0\leq n\leq N}A_wC
\]
($w$ is the binary expansion of the integer~$n$). With regard to the
column vector~$C$, we note that both matrices~$A_0$ and~$A_1$ write
$A_0=\cos\vartheta\,\operatorname{I}_2$,
$A_1=\sin\vartheta\,R_{\pi/2}$ (notation of
Ex.~\ref{Du08vers09:ex:geometricex}: $R_{\varphi}$ is the rotation
matrix with angle~$\varphi$).  As a consequence they commute with all
rotation matrices and if~$C$ is changed into $R_{\varphi}C$, then
$\boldSigma_N$ is changed into $R_{\varphi}\boldSigma_N$. Evidently we have a
similar formula with dilations and so we may take as vector~$C$ the
first vector~$E_1$ of the canonical basis.

The joint spectral radius is
$\lambdajsr=\max(|\cos{\vartheta}|,|\sin{\vartheta}|)$, because
$\norm{A_0}_1=|\cos{\vartheta}|$, $\norm{A_1}_1=|\sin{\vartheta}|$
and~$A_0^2$ and~$A_1^2$ admit respectively the
eigenvalues~$\cos^2{\vartheta}$ and~$-\sin^2{\vartheta}$.  The
matrix~$Q=A_0+A_1$ is the rotation matrix~$R_{\vartheta}$ and its
eigenvalues are $\rho\omega_{\pm}$ with $\rho=1$ and
$\omega_{\pm}=e^{\pm i\vartheta}$.  We take as eigenvectors
respectively for $\omega_+$ and $\omega_-$ the vectors
$V_+=\left(\begin{array}{cc} 1 & -i\end{array}\right)^\trpse$, $
V_-=\left(\begin{array}{cc} 1 & i
\end{array}\right)^\trpse$. The vector~$C$ expands as $C=(V_++V_-)/2$.
We have $\lambdajsr<\rho$ and
Theorem~\ref{Du08vers09:thm:sequenceasymptoticexpansion} applies. The
asymptotic expansion for the sum~$\boldSigma_{+,N}$ associated with $V_+$
writes
\[
\boldSigma_{+,N}=(\operatorname{I}_2-A_0)\sum_{k=0}^KQ^kV_++
e^{i(K+1)\vartheta}\boldF_+(2^{\{t\}-1})+O(\lambda_*^K),
\]
where~$\boldF_{\pm}$ is defined by the dilation equation
\begin{multline*}
\left\{\begin{array}{l}
F_{{\pm},1}(x)=\cos{\vartheta}e^{\mp i\vartheta}F_{{\pm},1}(2x),\\
F_{{\pm},2}(x)=\cos{\vartheta}e^{\mp i\vartheta}F_{{\pm},2}(2x),
\end{array}\right.\quad\text{for $0\leq x<\frac{1}{2}$;}\\
\left\{\begin{array}{l}
F_{{\pm},1}(x)=\cos{\vartheta}e^{\mp i\vartheta}-\sin{\vartheta}e^{\mp i\vartheta}F_{{\pm},2}(2x-1),\\
F_{{\pm},2}(x)=\mp i\cos{\vartheta}e^{\mp i\vartheta}+\sin{\vartheta}e^{\mp i\vartheta}F_{{\pm},1}(2x-1),
\end{array}\right.\quad\text{for $\frac{1}{2}\leq x<1$,}
\end{multline*}
with the boundary conditions $\boldF_\pm(0)=0$,
$\boldF_\pm(1)=V_\pm$. Both vector-valued functions $\boldF_\pm$ are
conjugate. Emphasizing their real and imaginary parts~$\boldU$
and~$\pm\boldV$ and carrying on with the computation we obtain for the
sum~$\boldSigma_N$ associated with $C=E_1$ the asymptotic expansion
\[
\boldSigma_N\mathop{=}_{N\to+\infty}
\boldGamma(t)+O(\lambda_*^K),
\]
with $\boldGamma:t\mapsto \Omega+\boldA(t)+\boldB(t)$,
\[
\Omega=
\sin\frac{\vartheta}{2}
\left(\begin{array}{c}
\sin\vartheta/2\\
\cos\vartheta/2
\end{array}\right),\qquad
\boldA(t)=
\sin\frac{\vartheta}{2}
\left(\begin{array}{c}
\cos((K+1/2)\vartheta-\pi/2)\\
\sin((K+1/2)\vartheta-\pi/2)
\end{array}\right),
\]
\[
\boldB(t)=
\cos ((K+1)\vartheta)\,\boldU(2^{\{t\}-1})
-
\sin ((K+1)\vartheta)\,\boldV(2^{\{t\}-1}).
\]
Obviously an increasing of~$t$ (and consequently of~$K$) by~$1$
rotates the vector~$\boldA(t)$ by~$\vartheta$. The change for~$\boldB$
is less obvious. From the real point of view, the dilation equation
for~$\boldF_+$ rewrites
\[
\boldPhi(x)=T_0\boldPhi(2x)\qquad\text{for $0\leq x<\frac{1}{2}$},\qquad
\boldPhi(x)=T_0W+T_1\boldPhi(2x-1)\qquad
\text{for $\frac{1}{2}\leq x<1$,}
\]
where~$\boldPhi$ is the $4$-dimensional vector-valued function
$\boldPhi=\left(\begin{array}{cc}\boldU & \boldV
\end{array}\right)^\trpse$. The matrices~$T_0$ and~$T_1$ have the
following expressions
\[
T_0=
\left(\begin{array}{cccc}
\cos^2{\vartheta} & 0 & \cos{\vartheta}\sin{\vartheta} & 0 \\
0 & \cos^2{\vartheta} & 0 & \cos{\vartheta}\sin{\vartheta} \\
-\cos{\vartheta}\sin{\vartheta} & 0 & \cos^2{\vartheta} & 0 \\
0 & -\cos{\vartheta}\sin{\vartheta} &  0 & \cos^2{\vartheta}
	  \end{array}\right)
=\cos{\vartheta}\,R_{-\vartheta}\otimes
\operatorname{I}_2,
\]
\[
T_1=\left(
\begin{array}{cccc}
0 & -\cos{\vartheta}\sin{\vartheta} & 0 & -\sin^2{\vartheta} \\
\cos{\vartheta}\sin{\vartheta} & 0 & \sin^2{\vartheta} & 0 \\
 0 & \sin^2{\vartheta} &0 & -\cos{\vartheta}\sin{\vartheta}\\
-\sin^2{\vartheta} & 0 &\cos{\vartheta}\sin{\vartheta} & 0
\end{array}
\right)=\sin{\vartheta}\,R_{-\vartheta}\otimes R_{\pi/2}.
\]
The boundary conditions are $\boldPhi(0)=0$ and
$\boldPhi(1)=W$ with $W=\left(\begin{array}{cccc} 1 & 0 & 0 & -1\end{array}\right)^\trpse$.
Let~$P$ be the matrix
\[
P=\left(\begin{array}{cccc}
0 & 0 & 0 & -1\\
0 & 0 & 1 & 0 \\
0 & 1 & 0 & 0\\
-1 & 0 & 0 & 0
	\end{array}\right).
\]
We verify $P^{-1}T_0P=T_0$, $P^{-1}T_1P=T_1$, $PW=W$. As a result
$P\boldPhi$ is a solution of the dilation equation. But the
dilation equation has a unique solution and we conclude that
$\boldPhi$ satisfies $P\boldPhi=\boldPhi$. This means
$-\boldV=R_{\pi/2}\boldU$. This shows the formula
$\boldB(t+1)=R_{\vartheta}\boldB(t)$. Eventually we see that
$\boldGamma(t+1)$ is the image of~$\boldGamma(t)$ by the rotation of
angle~$\vartheta$ about~$\Omega$.

Let us assume that~$\vartheta$ is commensurable with~$\pi$, and let us
write $\vartheta/\pi=p/q$ where $p/q$ is in lowest terms. We see
immediately that the arc $\boldGamma$ is $2q$-periodic, and even
$2^{1-\kappa}q$-periodic if~$\kappa$ is the dyadic valuation of~$p$. This
phenomenon is illustrated by
Fig.~\ref{Du08vers09:fig:rosette5}. Besides it is easy to verify
$\boldGamma(t+2^{-\kappa}q)+\boldGamma(t)=2\Omega$.

To the contrary when the ratio $\vartheta/\pi$ is irrational, the
rotation~$R_{\vartheta}$ leaves invariant the image of the arc
$\boldGamma$ but its order is infinite. The non-periodic character
of the arc~$\boldGamma$ is evident when we consider its value at integers,
\[
\boldGamma(K)=\Omega+
\cos\frac{\vartheta}{2}\left(\begin{array}{c}
\cos(K-1/2)\vartheta\\
\sin(K-1/2)\vartheta
\end{array}\right).
\]
Hence we see that the functions~$\Phi$ are not necessarily periodic,
even if they are frequently periodic in concrete examples.
\end{example}

\subsubsection*{\bf Acknowledgment} I am most grateful towards {\sc
  Fr\'ed\'eric Chyzak} ({\sc Algorithms Project, INRIA}) for his ever
patient ear during my effort to understand radix-rational sequences.

\footnotesize
\bibliographystyle{plainnat}
\bibliography{Du08}

\normalsize

\end{document}